\documentclass[superscriptaddress,nofootinbib,notitlepage]{revtex4-1}
\pdfoutput=1

\usepackage{graphicx}
\usepackage{amsthm}
\usepackage{thmtools}
\usepackage{mathtools}
\usepackage{thm-restate}
\usepackage{amsmath}
\usepackage{algorithm}
\usepackage{algorithmic}
\usepackage{amssymb}
\usepackage{bm}
\usepackage{xcolor}
\usepackage{placeins}
\usepackage[normalem]{ulem}
\usepackage[colorlinks]{hyperref}
\hypersetup{
	pdfstartview={FitH},
	pdfnewwindow=true,
	colorlinks=true,
	linkcolor=blue,
	citecolor=blue,
	filecolor=blue,
	urlcolor=blue}
\usepackage{tikz}
\usetikzlibrary{calc,backgrounds}
\usepackage{physics}
\usepackage{cancel}
\usepackage{multirow}
\usepackage{pgffor}
\usepackage{url}
\usepackage{hypcap}
\usepackage{dsfont}
\usepackage{soul}
\usepackage{inconsolata}

\usepackage{hyperref}
\hypersetup{pdfpagemode=UseNone}

\allowdisplaybreaks

\newcommand{\thm}[1]{\hyperref[thm:#1]{Theorem~\ref*{thm:#1}}}
\newcommand{\defn}[1]{\hyperref[defn:#1]{Definition~\ref*{defn:#1}}}
\newcommand{\lem}[1]{\hyperref[lem:#1]{Lemma~\ref*{lem:#1}}}
\newcommand{\prop}[1]{\hyperref[prop:#1]{Proposition~\ref*{prop:#1}}}
\newcommand{\fig}[1]{\hyperref[fig:#1]{Figure~\ref*{fig:#1}}}
\newcommand{\tab}[1]{\hyperref[tab:#1]{Table~\ref*{tab:#1}}}
\renewcommand{\sec}[1]{\hyperref[sec:#1]{Section~\ref*{sec:#1}}}
\newcommand{\app}[1]{\hyperref[app:#1]{Appendix~\ref*{app:#1}}}
\newcommand{\cor}[1]{\hyperref[cor:#1]{Corollary~\ref*{cor:#1}}}
\newcommand{\obs}[1]{\hyperref[obs:#1]{Observation~\ref*{obs:#1}}}
\newcommand{\nn}{\nonumber \\}
\newcommand{\append}[1]{\hyperref[append:#1]{Appendix~\ref*{append:#1}}}

\usepackage[capitalise]{cleveref} 

\newtheorem{theorem}{Theorem}
\newtheorem{definition}[theorem]{Definition}

\newtheorem{lemma}[theorem]{Lemma}
\newtheorem{proposition}[theorem]{Proposition}

\newtheorem{innercustomthm}{Theorem}
\newenvironment{customthm}[1]
{\renewcommand\theinnercustomthm{#1}\innercustomthm}
{\endinnercustomthm}

\input{Qcircuit}

\newcommand{\MQ}{\affiliation{%
Department of Physics and Astronomy,
Macquarie University, Sydney, NSW 2109, AU} }

\newcommand{\Google}{\affiliation{%
Google Quantum AI, Venice, CA 90291, USA}}

\newcommand{\UTS}{\affiliation{%
Centre for Quantum Software and Information,
University of Technology Sydney, Sydney, NSW 2007, AU}}

\newcommand{\UMD}{\affiliation{%
Joint Center for Quantum Information and Computer Science, University of Maryland, College Park, MD 20742, USA}}          

\begin{document}
\title{Optimal scaling quantum linear systems solver via discrete adiabatic theorem}
\date{\today}
\author{Pedro C.~S.~Costa} \MQ
\author{Dong An} \UMD \Google
\author{Yuval R.~Sanders}  \MQ \UTS
\author{Yuan Su} \Google
\author{Ryan Babbush} \Google
\author{Dominic W.~Berry} \MQ

\begin{abstract}
Recently, several approaches to solving linear systems on a quantum computer have been formulated in terms of the quantum adiabatic theorem for a continuously varying Hamiltonian. Such approaches enabled near-linear scaling in the condition number $\kappa$ of the linear system, without requiring a complicated variable-time amplitude amplification procedure. However, the most efficient of those procedures is still asymptotically sub-optimal by a factor of $\log(\kappa)$. Here, we prove a rigorous form of the adiabatic theorem that bounds the error in terms of the spectral gap for intrinsically discrete time evolutions.
We use this discrete adiabatic theorem to develop a quantum algorithm for solving linear systems that is asymptotically optimal, in the sense that the complexity is strictly linear in $\kappa$, matching a known lower bound on the complexity.
Our $\mathcal{O}(\kappa\log(1/\epsilon))$ complexity is also optimal in terms of the combined scaling in $\kappa$ and the precision $\epsilon$.
Compared to existing suboptimal methods, our algorithm is simpler and easier to implement.
Moreover, we determine the constant factors in the algorithm, which would be suitable for determining the complexity in terms of gate counts for specific applications.
\end{abstract}

\maketitle

\tableofcontents

\section{Introduction}

Finding the solution to a system of linear equations is a fundamental task that underlies many areas of science and technology. A classical linear system solver takes time proportional to the number of unknown variables even to write down the solution, and thus has a prohibitive computational cost for solving large linear systems. However, a quantum computer with a suitable input access can produce a quantum state that encodes the problem solution much faster than its classical counterpart. The first quantum algorithm for the quantum linear system problem (QLSP) was proposed by Harrow, Hassidim, and Lloyd (HHL) \cite{Harrow_2009}, and has been subsequently refined by later work. Due to the ubiquitous nature of the problem, quantum algorithms for QLSP have found a variety of applications, such as computing electromagnetic scattering \cite{Clader13}, solving differential equations \cite{BerryJPA14,BerryCMP17}, data fitting \cite{Wiebe2012}, machine learning \cite{supvec,lloyd2013quantum}, and more general solution of partial differential equations \cite{Ashley2016}.

Specifically, the goal of QLSP is to produce a quantum state $\ket{x}$ proportional to the solution of linear system $Ax=b$, where $A$ is an $N$-by-$N$ matrix.
The complexity of solving QLSP depends on various input parameters, such as the problem size $N$, the sparsity (for sparse linear systems), the norm of the coefficient matrix $A$, the condition number $\kappa$, and the error $\epsilon$ in the solution.
To simplify the discussion, we assume that $\norm{A}=1$ and hence $\norm{A^{-1}}=\kappa$, where $\norm{\cdot}$ denotes the spectral norm.
To further simplify the analysis, we assume that we have a block encoding of the coefficient matrix $A$ and a given operation to prepare the target vector $\ket{b}$, and consider the number of queries to these oracles. 
One can also consider the complexity in terms of the number of calls to entries of a sparse matrix, as in \cite{Harrow_2009}, but there are standard methods to block encode sparse matrices \cite{rootd}, so our result can be easily applied to that case.
These simplifications mean that the only relevant parameters which our algorithm depends on are $\kappa$ and $\epsilon$.

The original algorithm proposed by HHL has a complexity scaling quadratically with the condition number $\kappa$ and linearly with the inverse accuracy $1/\epsilon$ \cite{Harrow_2009}. The scaling with the condition number was improved by Ambainis using ``variable time amplitude amplification'' \cite{ambainis2010variable}; the resulting algorithm has a near-linear dependence on $\kappa$ but a much worse dependence on $1/\epsilon$. A further improvement was provided in \cite{CKS}, which yields a complexity logarithmic in the allowable error $\epsilon$. 
Unfortunately, algorithms based on variable time amplitude amplification \cite{ambainis2010variable,CKS} perform multiple rounds of recursive amplitude amplifications and can be challenging to implement in practice.

To address this, recent work has suggested alternative approaches based on adiabatic quantum computing (AQC). AQC is a universal model of quantum computing that has been shown to be polynomially equivalent to the standard gate model~\cite{farhi2000quantum,Aharonov2007}.
In AQC one encodes the solution to a computational problem in the ground state of a Hamiltonian $H_1$. Then, one initializes a quantum system in the ground state of an easy-to-prepare Hamiltonian $H_0$ and slowly deforms from the ground state of $H_0$ to the ground state of $H_1$ under a time-dependent Hamiltonian that interpolates between the two, such as $H(s)=(1-s)H_0+sH_1$. The advantage of using the adiabatic approach to solve QLSP as in \cite{PhysRevLett.122.060504} is that it naturally provides complexity close to linear in $\kappa$, without the highly complicated variable time amplitude amplification procedure. That work was further improved in \cite{an2019quantum} then \cite{Lin2020optimalpolynomial}, which gave complexity logarithmic in $\epsilon$ by using eigenstate filtering. We summarize key developments reducing the complexity in \tab{linear_systems_history}.

\begin{table*}[ht]
\centering
\begin{tabular}{|c|c|c|c|}
\hline
Year
& Reference
& Primary innovation
& Query complexity\\
\hline\hline
2008
& Harrow, Hassidim, Lloyd \cite{Harrow_2009}
& first quantum approach
& ${\cal O}(\kappa^2 / \epsilon)$\\
2012
& Ambainis \cite{ambainis2010variable}
& variable-time amplitude amplification
& ${\cal O}(\kappa (\log (\kappa) /\epsilon)^3)$\\
2017
& Childs, Kothari, Somma \cite{CKS}
& Fourier/Chebyshev fitting using LCU
& ${\cal O}(\kappa\, {\rm polylog} (\kappa/\epsilon))$\\
2018
& Subasi, Somma, Orsucci \cite{PhysRevLett.122.060504}
& adiabatic randomization method
& ${\cal O}((\kappa \log \kappa) / \epsilon)$\\
2019
& An, Lin \cite{an2019quantum}
& time-optimal adiabatic method
& ${\cal O}(\kappa \, {\rm polylog}(\kappa/\epsilon)) $\\
2019
& Tong, Lin \cite{Lin2020optimalpolynomial}
& Zeno eigenstate filtering
& ${\cal O}(\kappa \log (\kappa/\epsilon))$\\
2021
& this paper
& discrete adiabatic theorem
& ${\cal O}(\kappa \log(1/\epsilon))$\\
\hline
\end{tabular}
\caption{\label{tab:linear_systems_history} History of the lowest scaling algorithms for solving linear systems of equations on a quantum computer. Specifically, the problem is to prepare the state $\ket{x}$ given oracular access to the matrix $A$ and the ability to prepare the initial state $\ket{b}$ encoding a vector $b$ with the relation $A x = b$. Here, $\kappa$ is the condition number of $A$ and $\epsilon$ is the target precision to which we would like to prepare the state $\ket{x}$. However, the cost of a query for all classical algorithms is expected to scale polynomially in $N$ (the dimension of the matrix $A$), whereas on a quantum computer it is possible to make queries in complexity scaling as ${\cal O}(\textrm{polylog}(N))$ when $A$ is a sparse matrix. Query complexity of $\Omega(\kappa \log(1/\epsilon))$ is a known lower bound on the complexity.}
\end{table*}

It is known that a quantum algorithm must make at least $\Omega(\kappa \log(1/\epsilon))$ queries in general to solve the sparse QLSP problem \cite{RobinAram}. Therefore, the method in \cite{Lin2020optimalpolynomial} is already optimal in the scaling with solution accuracy $\epsilon$. However, a question left open was: how can we achieve an optimal scaling with the condition number $\kappa$, or is it possible to prove a lower bound ruling out this improvement? From the algorithmic perspective, finding a quantum algorithm with linear $\kappa$-scaling is technically challenging. Previous fast linear system solvers depend on polynomial approximations to implement the inverse function $1/x$ on $x\in[1/\kappa,1]$ \cite{CKS}, or truncations of the Dyson series to implement the continuous adiabatic evolution \cite{PhysRevLett.122.060504,an2019quantum,Lin2020optimalpolynomial}. In either case, an extra $\mathrm{polylog}(\kappa)$ factor is required to suppress the truncation or approximation error, resulting in a superlinear scaling with the condition number.

In this work, we develop a quantum algorithm for solving systems of linear equations with complexity $\mathcal{O}(\kappa\log(1/\epsilon))$. That is, we achieve a strictly linear scaling with $\kappa$, while maintaining the logarithmic scaling with $1/\epsilon$ from the best previous algorithms. Combining with the lower bound of \cite{RobinAram}, we establish for the first time a quantum linear system algorithm with optimal scaling in the condition number. 
It is also optimal in the combined scaling with $\kappa$ and $\epsilon$, because one cannot for example reduce the scaling to $\mathcal{O}(\kappa+\log(1/\epsilon))$.
We formally state our result in \sec{filter} and preview it here.
\begin{customthm}{}[QLSP with linear dependence on $\kappa$]
Let $Ax=b$ be a system of linear equations, where $A$ is an $N$-by-$N$ matrix with $\norm{A}=1$ and $\norm{A^{-1}}=\kappa$. Given an oracle block encoding the operator $A$ and an oracle preparing $\ket{b}$, there exists a quantum algorithm which produces the normalized state $\ket{A^{-1}b}$ to within error $\epsilon$ using a number
\begin{equation}
    \mathcal{O}(\kappa\log(1/\epsilon))
\end{equation}
of oracle calls.
\end{customthm}

Our algorithm is conceptually simple and easy to describe.
All it requires is a sequence of quantum walk steps corresponding to a qubitised form of the Hamiltonian used in prior work.
It completely avoids the heavy mechanisms of variable time amplitude amplification or the truncated Dyson-series subroutine from previous methods. 
Moreover, we provide a bound on the constant prefactor for our approach, that will allow estimation of the complexity in terms of the number of gates for specific applications.
We expect that our estimate of the prefactor can be tightened, and our algorithm will be the most efficient for the early fault-tolerant regime of quantum computation as well as having the best asymptotic scaling for large-scale applications.

The new insight that allows us to establish the optimal $\kappa$-scaling is the use of a discrete quantum adiabatic theorem, a result proved by Dranov, Kellendonk, and Seiler (DKS)~\cite{DKS98}. Unlike the continuous version, the discrete adiabatic theorem is formulated based on a quantum walk operator $W_T(s)$. Provided that the steps of quantum walk vary sufficiently slowly, the eigenstates of the walk operator can be approximately preserved throughout the entire discrete adiabatic evolution. Indeed, DKS showed that the error in the evolution scales as $\mathcal{O}(1/T)$ for $T$ steps of the walk. However, their analysis overlooked the scaling with other parameters, in particular, the spectral gap dependence.
In the case of solving QLSP, the gap depends on $\kappa$, so the result of DKS is not sufficient to give the $\kappa$-dependence of the algorithm.
Here, we give a complete analysis of the discrete adiabatic theorem, keeping track of all the parameters of interest while fixing several minor mistakes in the original proof. 

In developing our quantum linear system algorithm, we provide an improved method of filtering the final state which may be of independent interest.
Prior methods were based on singular value processing \cite{Lin2020optimalpolynomial}, which requires a sequence of rotations to be found by a numerically demanding procedure~\cite{Childs18,Haah2019product,Dong21,Chao2020}.
Our method is just as efficient, but the sequence of operations needed is easily determined.
Combining the discrete adiabatic theorem on the qubitised quantum walk with the improved eigenstate filtering then gives our result on the solution of linear systems.

The remainder of the paper is organized as follows.
In the following we give more detailed background and summarise our result in \sec{summary}.
Then in \sec{adtheo} we give a thorough proof of the discrete adiabatic theorem.
We base our method on the approach of DKS, but make many of the details rigorous and provide a strict bound on the error including constant factors.
We apply the discrete adiabatic theorem to the QLSP in \sec{linsys}.
In \sec{filter} we provide our general method of filtering that is just as efficient as that based on singular value processing.

\section{Discrete adiabatic theorems}
\label{sec:summary}

\subsection{Background}
\label{sec:summary_background}

Before presenting our results, let us present the main ideas of the DKS bound on the error in discrete-time adiabatic evolution \cite{DKS98}.
In this proposal the model of the adiabatic evolution is based on a sequence of $T$ walk operators $\{W_T(n/T): n \in \mathbb{N}, 0\leq n \leq T-1\}$.
That is, the system is initially prepared in a state $\ket{\psi_0}$, then the sequence of unitary transformations $W_T(n/T)$ have the effect 
$\ket{\psi_0} \mapsto \ket{\psi_1} \mapsto \cdots$.
To model this evolution, with $s := n/T$ we can write
\begin{equation}
\label{eq:U_as_product_of_W}
    U_T(s) = \prod_{n = 1}^{sT-1} W_T\left(n/T\right),
\end{equation}
and $U_T(0) \equiv \mathds{1}$, such that $\ket{\psi_n} = U_T(s)\ket{\psi_0}$.
The adiabatic limit is then the limit $T \to \infty$.
Alternatively we can construct the total unitary evolution recursively as
\begin{equation}
\label{eq:discrete_time_schroedinger_eqn_rescaled}
  U_T (s + 1/T) = W_T (s+1/T) U_T (s), \quad U_T (0) = \mathds{1}.
\end{equation}

The adiabatic limit is then the limit $T \to \infty$.
For the purpose of quantum algorithm design,
we are trying to choose $U_T$ so that
$\lim_{T \to \infty} U_T (1) \ket{\psi_0} = \ket{\psi_{\rm target}}$,
where $\ket{\psi_{\rm target}}$ is a desired `target' state that
enables us to solve a computational problem.
In order for this to be an accurate adiabatic evolution yielding the target state,
$U(n/T) \ket{\psi_0}$ should be approximately an eigenstate of $W (n/T)$ for all $n$.

We need to establish some terminology before we can present the statement of the result from \cite{DKS98}.
For each $T \in \mathbb{R}$ and $n \in \mathbb{N}$, introduce a projector $P_T (s)$ (with $s \equiv n/T$ as before) called the
\textit{spectral projection}, which projects onto the eigenspace of interest.
In addition, the \textit{complementary spectral projection} $Q_T (s)=\openone-P_T (s)$ projects onto all orthogonal eigenvectors.
An operator representing the ideal adiabatic evolution is denoted $U_T^A(s)$, so that
\begin{equation}
    P_T(s)=U_T^A(s) P_T (0) U_T^{A\dagger}(s) .
\end{equation}
That is, evolving the original eigenspace to step $n=sT$ under the ideal adiabatic evolution gives the corresponding eigenspace for the walk operator $W_T(s)$.

The adiabatic theorem is a statement about how close the ideal adiabatic evolution $U_T^A(s)$ is to the real evolution $U_{T}(s)$ at a given time.
Beginning with the initial state $\ket{\phi(0)}$ in the subspace of interest so $(P_T(0)\ket{\phi(0)}=\ket{\psi(0)})$, the goal is to bound the \textit{error}  between $U_{T}$ and its ideal evolution $U_T^A$ by an expression of the form
\begin{equation}
\label{eq:goal}
    \left\|\left(U_T(s)-U_T^A(s)\right)\ket{\phi(0)}\right\| \leq  \left\|U_T(s)-U_T^A(s)\right\| \leq \frac{\theta}{T} ,
\end{equation}
where $\| \cdot \|$ is the spectral norm.
Proving this result shows that increasing the number of steps reduces the error.
The constant $\theta$ in \cref{eq:goal} is a constant independent of the total time $T$, but depends on the gap $\Delta(s)$ between the eigenspace of interest and the complementary eigenspace.

In \cite{DKS98} it was shown that the error is $\mathcal{O}(1/T)$, which means that there exists some constant $\theta$, but that constant and its dependence on the gap were not determined.
That is a crucial difficulty in applying the result to the QLSP,
because the gap in using the adiabatic approach to the QLSP depends on the condition number $\kappa$.
Therefore, to determine the complexity of the algorithm in terms of $\kappa$, we need to know the dependence of the error on the gap.
In particular, we will show that the error scales as $\mathcal{O}(\kappa/T)$, which means that to obtain the solution to fixed error one can use $T=\mathcal{O}(\kappa)$ steps.
Then complexity linear in $\kappa$ and logarithmic in $1/\epsilon$ can be obtained using filtering.
To show this result we cannot simply use the result as given \cite{DKS98}, and need to derive the bound for the error far more carefully in order to give the dependence on the gap.

\subsection{Our result}
\label{sec:summary_result}

Our main goal here is to provide the explicit dependence on the gap in the discrete adiabatic theorem in order to improve the version given in \cite{DKS98}. 
The starting point is to replace the general order scaling
\begin{equation}
   W_T\left(s+1/T\right)-W_T\left(s\right)\approx \mathcal{O}(T^{-1}),
   \end{equation}
with an upper bound with explicit schedule dependence, 
\begin{equation}
\label{eq:main_ass}
   \left\|W_T\left(s+1/T\right)-W_T\left(s\right)\right\|\leq \frac{c(s)}{T}.
 \end{equation}
Implicit in this definition is the assumption that the behaviour of $W_T (s)$ is sufficiently smooth that $c(s)$ can be chosen independently of $T$.
This will need to be shown for the given applications.
More generally, we will need to consider higher-order differences, which result in values of $c_k(s)$ given in the following definition.
 
 \begin{definition}[Multistep Differences]
 \label{def:difs}
 For a positive integer $k$, the $k^{\rm th}$ difference of $W_T$ is
    \begin{equation}
        D^{(k)} W_T(s) :=
       D^{(k-1)} W_T \left( s + \frac{1}{T} \right) - D^{(k-1)} W_T(s), \quad
       D^{(1)} W_T(s) := DW_T(s) =
       W_T \left( s + \frac{1}{T} \right) - W_T (s).
    \end{equation}
    For $T > 0$,
    we define the function $c_k(s)$, which is implicit dependent of $T$, such that
    \begin{equation}\label{eqn:assump1}
        \left\|D^{(k)} W_T(s)\right\| \leq \frac{c_k(s)}{T^{k}}.
    \end{equation}
We then define the $\hat{c}_k(s)$ taking into account neighbouring steps as
\begin{equation}
\label{eq:chat}
    \hat{c}_k(s) = \max_{s' \in \left\{s-1/T,s,s+1/T\right\} \cap [0,1-k/T] } c_k(s') .
\end{equation}
\end{definition}

The principle of the gap is that it separates the eigenvalues of $W_T(s)$ into two groups that depend on the time parameter $s$.
Since $W_T(s)$ is unitary, these are groups on the unit circle in the complex plane.
Because it is on the unit circle, we need to separate these groups of eigenvalues with gaps in two locations.
We will denote one set of eigenvalues as $\sigma_P(s)$ and the other as $\sigma_Q(s)$, with corresponding projectors $P_T(s)$ and $Q_T(s)$, respectively.

We also need to account for the gaps for successive operators $W_T(s)$ and $W_T(s+1/T)$.
That is, there needs to be a gap between $\sigma_P(s)\cup \sigma_P(s+1/T)$ and $\sigma_Q(s)\cup \sigma_Q(s+1/T)$.
Moreover, we need to ensure that these regions are non-interleaved.
To do this, we will define arcs that contain the eigenvalues, such that
\begin{equation}
    \sigma^{(1)}_P \supseteq \sigma_P(s)\cup \sigma_P(s+1/T), \qquad     \sigma^{(1)}_Q \supseteq \sigma_Q(s)\cup \sigma_Q(s+1/T).
\end{equation}
To rule out interleaved regions, these arcs cannot intersect, and we consider the gap between these arcs.
We are interested in the case where this only has a small effect on the gap.
In turn this means that $T$ should not be too large, so we introduce a lower bound $T^*$ on the values of $T$ allowed.
We therefore define the multistep gaps as follows.

\begin{definition}[Multistep Gap]
\label{def:gaps}
For $T \in \mathbb{N}$ and $k$ a non-negative integer, $\Delta_k(s)$ is defined to be the minimum angular distance between arcs $\sigma^{(k)}_P$ and $\sigma^{(k)}_Q$, which satisfy
\begin{equation}\label{eqn:assump2}
    \sigma^{(k)}_P \supseteq \bigcup_{l=0}^k \sigma_P(s+l/T), \qquad 
    \sigma^{(k)}_Q \supseteq \bigcup_{l=0}^k \sigma_Q(s+l/T),
\end{equation}
for $T\ge T_*$.
The gap $\Delta(s)$, which is also implicitly dependent on $T$, is then in most cases the minimum gap for three successive steps, except in the cases at the boundaries:
\begin{equation}\label{eq:minGaps}
    \Delta(s) = \begin{cases}
        \Delta_2(s), & 0 \leq s \leq 1-2/T, \\
        \Delta_1(s), & s = 1-1/T, \\
        \Delta_0(s), & s = 1.
    \end{cases}
\end{equation}
Finally, $\check{\Delta}(s)$ is an adjustment for $\Delta(s)$ at neighbouring points:
\begin{equation}
\label{eq:fhat}
    \check{\Delta}(s) = \min_{s' \in \left\{s-1/T,s,s+1/T\right\} \cap [0,1] } \Delta(s').
\end{equation}
\end{definition}

Note that we have freedom to choose larger arcs than necessary, so these are lower bounds on the gap, though we will often call them the ``gap'' for convenience.
Also, given $\Delta_2(s)$, one can always choose arcs $\sigma^{(k)}_P$ and $\sigma^{(k)}_Q$ for $k=0,1$ such that $\Delta_k(s)\ge \Delta_2(s)$.
This means that in \cref{eq:minGaps} we can simply use $\Delta_2(s)$, rather than taking the minimum of $\Delta_k(s)$ for $k \in \{0,1,2\}$.

We have proven two forms of the discrete adiabatic theorem.
One is highly complicated, so we give it explicitly in later in \cref{sec:dat1}.
Here we instead give a simplified but looser form of the discrete adiabatic theorem.

\begin{theorem}[The Second Discrete Adiabatic Theorem]\label{cor:adia}
Suppose that the operators $W_T (s)$ satisfy
 $\left\|D^{(k)}W_T(s)\right\|\leq c_k(s)/T^k$ for $k = 1,2$, as per \cref{eqn:assump1}, and $T \geq \max_{s\in [0,1]} (4\hat{c}_1(s)/\check{\Delta}(s))$, 
 Then for any time $s$, s.t., $sT \in \mathbb{N}$, we have
 \begin{align}
     \quad \|U_T(s) - U_T^{A}(s)\| 
       & \leq \frac{12\hat{c}_1(0)}{T\check{\Delta}(0)^2}+  \frac{12\hat{c}_1(s)}{T\check{\Delta}(s)^2} + \frac{6\hat{c}_1(s)}{T\check{\Delta}(s)} + 305\sum_{n=1}^{sT-1}\frac{\hat{c}_1(n/T)^2}{T^2\check{\Delta}(n/T)^3} \nonumber \\
       & \quad  + 44\sum_{n=0}^{sT-1} \frac{ \hat{c}_1(n/T)^2}{ T^2 \check{\Delta}(n/T)^2} +  32\sum_{n=1}^{sT-1}\frac{\hat{c}_2(n/T)}{T^2\check{\Delta}(n/T)^2}, 
\end{align}
where $\hat{c}_k(s)$ and $\check{\Delta}(s)$ are defined in \cref{def:difs} and \cref{def:gaps}, respectively.
\end{theorem}

Note that this theorem depends on the first and second differences, described by $\hat{c}_1(s)$ and $\hat{c}_2(s)$, respectively.
These are analogous to the first and second derivatives in the continuous form of the adiabatic theorem, so we can see that these results are analogous.
We have three single terms with $1/T$ scaling, and three sums with $1/T^2$ scaling, which gives overall scaling of the complexity as $\mathcal{O}(1/T)$.
We also have a cubic dependence in the inverse gap $1/\Delta$ in the first sum given.
In choosing the quantum walk, one would aim to schedule the variation of $W_T$ such that they vary more slowly where the gap is small, making $\hat{c}_1$ smaller. 

\section{The first adiabatic theorem}
\label{sec:adtheo}

In this section, we prove our first form of the discrete adiabatic theorem, given later as \cref{theoAdia}, and then use it to prove \cref{cor:adia}. Following the general method and notation of \cite{DKS98}, we use the \textit{wave operator}
\begin{equation}
\label{eq:waveOp}
  \Omega_T (s) := U^{A\dagger}_T(s) U_T (s) .
\end{equation}
The aim of the discrete adiabatic theorem, the first and the second, is to prove that $\Omega_T (s)$ is close to the identity because
\begin{equation}
    \left\|U_T(s)-U_T^A(s)\right\| = \left\|U_{T}^{A\dagger}(s)U_T(s)-\mathds{1}\right\| = \left\|\Omega_T (s)-\mathds{1}\right\|.
\end{equation}
In \cite{DKS98} it was shown that $\Omega_T (s) = \mathds{1} + \order{1/T}$, but we instead aim to provide the explicit bounds dependent on the gap.

To prove the bound, one can define a \emph{kernel function} $K_T (s)$ as well, which corresponds to the difference of a single step of $\Omega_T (s)$ from the identity.
The wave operator at step $n$ is then given by
        \begin{equation}
        \label{eq:waveOp_kernel}
            \Omega_T(n/T) = I - \frac{1}{T}\sum_{m=0}^{n-1} K_T(m/T)\Omega_T(m/T).
        \end{equation}
The goal is then to show that the sum is small.
This is done with a discrete form of the summation by parts formula, giving our first discrete adiabatic theorem.

\subsection{Operator definitions}

\begin{figure}
    \centering
    \includegraphics[width = 0.5\textwidth]{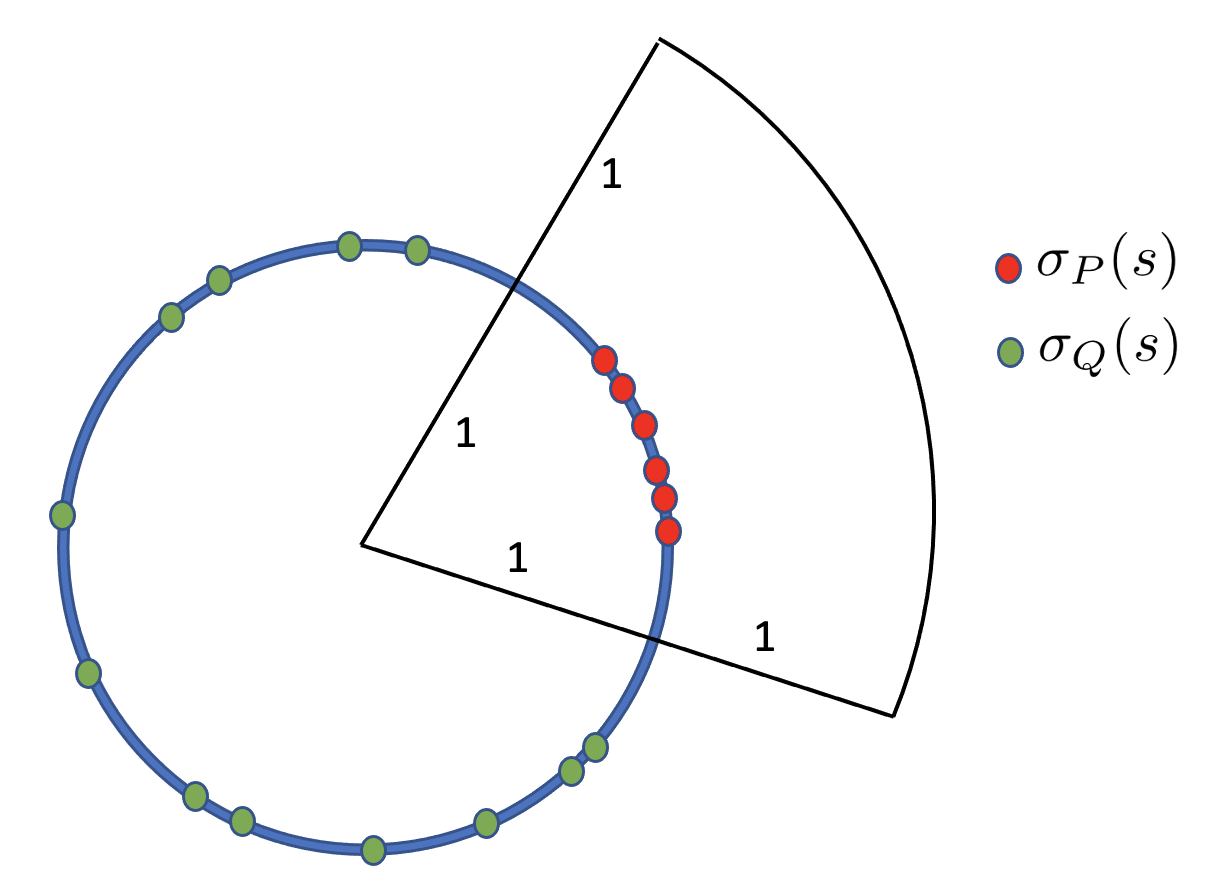}
    \caption{Illustration of the choice of the contour $\Gamma_T(s,k)$ for $k = 0$. 
    That is, we consider the eigenvalues for only a single step of the walk.
    The red dots indicate the spectrum of interest, which will often just be a single eigenvalue, for example for a ground state.
    The contour around the spectrum of interest is used to obtain a projector onto the spectrum of interest.
    For the illustration, we use a contour with radius $2$, but in practice, we take the limit that the radius goes to infinity. }
    \label{fig:contour_1}
\end{figure}

We next define the operators that are needed to understand the proof.
Let $\Gamma_T(s)$ be a sector contour enclosing the spectrum of interest, for example see \cref{fig:contour_1}.
Then the spectral projection $P_T(s)$ onto the spectrum of interest is given by the integral
\begin{equation}
\label{eq:ProjectCount}
    P_T(s)=\frac{1}{2 \pi i} \oint_{\Gamma_T(s)} R_T(s,z) dz, 
\end{equation}
where
\begin{equation}
\label{eq:resolv}
    R_T(s,z)\coloneqq\left(W_T(s)-z I\right)^{-1},
\end{equation}
is the resolvent of $W_T(s)$.
  Let 
 \begin{align}\label{eq:STdef}
        S_T(s,s')&\coloneqq P_T(s)P_T(s')+Q_T(s)Q_T(s'), \\
     v_T(s,s') &\coloneqq \sqrt{S_T(s,s')S^{\dagger}_T(s,s')} 
      = \sqrt{I-\left(P_T(s)-P_T\left(s'\right)\right)^2},
      \end{align}
      and
      \begin{equation}
      \label{eq:V}
           V_{T}(s,s') \coloneqq v_T(s,s')^{-1}S_T(s,s'),
      \end{equation}
      which is the unitary of the left polar decomposition of $S_T(s,s')$ (see Eq.~(11) of \cite{DKS98}). 
      We use the shorthand notations $S_T(s) = S_T(s+1/T,s)$, $v_T(s) = v_T(s+1/T,s)$, $V_T(s) = V_T(s+1/T,s)$, and define 
      (see Eqs.~(7) and (10) of \cite{DKS98})
    \begin{align}
    \label{eq:Wa}
         W_T^A(s) & \coloneqq V_T(s)W_T(s), \\
     U_T^A\left(s+\frac{1}{T}\right) & \coloneqq W_T^{A}(s)U_T^A(s),\\
     U^A_T(0) & \coloneqq I.
    \end{align}
    It can be checked from the definition that $V_T(s)$ is a unitary operator, and thus $W_T^{A}$ and $U_T^A$ are unitary. 
    In fact, $W_T^{A}$ is exactly the adiabatic walk operator, and the corresponding $U_T^{A}$ is the corresponding adiabatic evolution operator.

    To describe the proof we use the \emph{wave operator}
    \begin{equation}\label{eq:Omegadef}
        \Omega_T(s) \coloneqq U_T^{A\dagger}(s)U_T(s),
    \end{equation}
    which describes the difference between the actual evolution given by $U_T(s)$ and the ideal adiabatic evolution $U_T^{A}(s)$.
    To demonstrate that the evolution is close to adiabatic, we should have $\Omega_T(s)$ close to $I$.
    The \emph{ripple operator} is defined as 
    \begin{equation}\label{eq:ripple}
        \Theta_T (s) \coloneqq \Omega_T (s + 1/T) \Omega_T^\dagger (s), 
    \end{equation}
    so the wave operator is a product of ripple operators for each time step.
    The kernel function is defined as 
    \begin{equation}
    \label{eq:K}
        K_T(s) \coloneqq T( I - \Theta_T(s)),
    \end{equation}
    and should be close to zero for the evolution to be close to adiabatic.
    
\subsection{Sequence of Lemmas}
    Next we give an outline of the lemmas that will be proven and how they fit together to provide the final theorem.
    In the proof of the discrete adiabatic theorems, we keep with the the definition giving in \cref{eq:waveOp_kernel} for the wave operator. 
    The goal is then to upper bound the sum to show that $\Omega_T(n/T)$ is close to the identity, and therefore the evolution is close to adiabatic.
    In the sum we substitute the identity being equal to the sum of projections onto the desired subspace and the orthogonal subspace.
    That gives us four sums.
    Two of these are ``diagonal'' sums with two projections onto the same subspace, and two are ``off-diagonal'' sums with two different projections.
    
    The diagonal sums are relatively easily bounded, whereas for the off-diagonal sums are more difficult,
    for those we use the ``summation by parts formula'' in \cref{lem:sum_by_parts}.
    In that Lemma we use operators $X$ and $Y$, which will be taken to be $T(I-V_T^\dagger)$ and $\Omega_T$.
    In \cref{lem:sum_by_parts} we define operators $\tilde X$, $A$, $B$ and $Z$.
    The operator, $\tilde X$, is defined as a contour integral including $X$, then we have $B$ and $Z$ defined in terms of $\tilde X$.

    The sequence of lemmas used to prove the first adiabatic theorem (\cref{theoAdia}) are listed below.
    \begin{itemize}
        \item In \cref{lem:P} we bound the norms $\|DP_T\|$ and $\|D^{(2)}P_T\|$.
        The quantity $DP_T(s)$ is the difference in $P_T$ at successive time steps, and $D^{(2)}P_T$ is the difference in $DP_T$.
        These quantities are bounded in terms of the bounds on $DW_T$ and $D^{(2)}W_T$.
        \item \cref{lem:V_v2} gives $V_T$ in terms of $P_T$ and a new operator $\mathcal{F}_T$.
        \item \cref{lem:V_bound2} uses $\mathcal{F}_T$ to place an upper bound on the norm of $V_T-I$.
        Because $X$ is $T(I-V_T^\dagger)$, that enables us to place an upper bound on $X$.
        \item \cref{lem:DV_v2} provides an upper bound on the norm of $D\mathcal{F}_T$, which enables us to place an upper bound on $DV_T$.
        This uses the upper bounds on $\|DP_T\|$ and $\|D^{(2)}P_T\|$ from \cref{lem:P}.
        \item \cref{lem:WAv3} places an upper bound on $\|DW_T^A\|$ using the upper bound on $DV_T$ from \cref{lem:DV_v2}.
        Recall that $W_T^A$ is the ideal step for adiabatic evolution.
        \item \cref{lem:DOmega2} places an upper bound on $\|D\Omega_T\|$ using the upper bound on $\|V_T-I\|$ from \cref{lem:V_bound2}.
        \item \cref{lem:Xt} places upper bounds on $\|\tilde X\|$ and $\|D\tilde X\|$ in terms of $\|X\|$ and $\|DX\|$.
        Recall that $X$ will correspond to $T(I-V^\dagger)$.
        \item \cref{lem:ABZ2_v2} places upper bounds on the norms of the $A$, $B$ and $Z$ operators.
        The bound on $\|A\|$ uses the bound on $\|V_T-I\|$ from \cref{lem:V_bound2}.
        The bounds on $\|B\|$ and $\|Z\|$ use the bounds on $\|\tilde X\|$ and $\|D\tilde X\|$ from \cref{lem:Xt} as well as the bound on $\|DW_T^A\|$ from \cref{lem:WAv3}.
    \end{itemize}
    Finally, the bounds on $\|A\|$, $\|B\|$ and $\|Z\|$ from \cref{lem:ABZ2_v2} are used in the summation by parts formula in \cref{lem:sum_by_parts} to prove \cref{theoAdia}.
    
    \subsection{Properties of operators}
    Before giving the detailed lemmas, we provide some properties of the operators from \cite{DKS98}.
    First we give properties of the adiabatic operators and the projectors onto the subspaces.
    \begin{proposition}
        For any integers $T, n,m$ and the corresponding discrete time $s = n/T, s' = m/T$, we have $W_T^{A}(s)$ and $U_T^A(s)$ are unitary, and
\begin{align}\label{eq:projU}
                U_T^A(s) P_T(0) &= P_T(s)U_T^A(s), \\
                P_T(s+1/T)W_T^A(s) &= W_T^A(s)P_T(s), \\
  P_T (s) U_T (s) U_T^\dagger (s') P_T(s') &= 
    P_T (s) v_{T} (s, s') U_T^A (s)
    {U_T^A}^\dagger(s') P_T(s'), \label{eq:DKS_ansatz_P} \\
  Q_T (s) U_T (s) U_T^\dagger (s') Q_T(s') &= 
    Q_T (s) v_{T} (s, s') U_T^A (s)
    {U_T^A}^\dagger(s') Q_T(s'), \label{eq:DKS_ansatz_Q} \\
\label{eq:MainEq}
  P_T (s + 1/T) W_T (s) P_T (s) &=
    P_T (s + 1/T) v_T (s) W_T^A (s) P_T (s),\\
  Q_T (s + 1/T) W_T (s) Q_T (s) &=
    Q_T (s + 1/T) v_T (s) W_T^A (s) Q_T (s).
\end{align}
    \end{proposition}
    
See Eqs.~(8) to (9) and the accompanying discussion in \cite{DKS98} for explanation of these properties.
    Next we consider properties of the \emph{wave operator} $\Omega_T(s)$, the \emph{ripple operator} $\Theta_T (s)$ and the kernel function $K_T(s)$.
    One can simply prove that the ripple operator is a rotation of the operator $V_T$, and $\Omega_T$ satisfies a discrete form of the Volterra equation. 
    The key results are as in the following proposition, which is equivalent to Eqs.~(19) and (20) from \cite{DKS98}.
    \begin{proposition}\label{prop:volterra}
        For any integer $T,n$ and the discrete time $s = n/T$, we have 
        \begin{equation}\label{eqn:Theta_V}
      \Theta_T (s) =
      {U_T^A}^\dagger (s + 1/T) V_T^\dagger (s) {U_T^A} (s + 1/T).
      \end{equation}
        and the Volterra equation
        \begin{equation}\label{eqn:Volterra}
            \Omega_T(n/T) = I - \frac{1}{T}\sum_{m=0}^{n-1} K_T(m/T)\Omega_T(m/T).
        \end{equation}
    \end{proposition}
    
    Finally, we provide the definitions of the $\Tilde{X}$, $A$, $B$ and $Z$ operators.
    They are 
     \begin{align}
    \label{eq:Xtilde}
        \Tilde{X}(s) & \coloneqq -\frac{1}{2\pi i}\oint_{\Gamma_T(s)}R_T(s,z)X(s)R_T(s,z)dz,\\
    \label{eq:A}
        A(s) & \coloneqq \left(V_T^{\dagger}\left(s\right)-I\right)W_T^A(s),\\
    \label{eq:B}
         B(s) & \coloneqq D\Tilde{X}\left(s\right)W_T^A\left(s\right)+DW_T^A\left(s-1/T\right)\Tilde{X}\left(s\right),\\
    \label{eq:Z}
        Z(s) & \coloneqq T\left(\left[A\left(s\right),\Tilde{X}\left(s\right)\right]+B\left(s\right)\right),
    \end{align}
    where $R_T(s,z)$ is defined in \cref{eq:resolv}.
    At this point the intuition behind these operators is not clear, but we will see later that these operators are related to a summation by parts formula and can simplify the notation.
    The definitions of $\Tilde{X}(s)$ and $Z(s)$ are equivalent to Eqs.~(21) and (25) of \cite{DKS98}, and $A(s)$ and $B(s)$ are defined in unnumbered equations in the proof of Theorem 1 of that work.

\subsection{Bounding the operators}

Here we show the bounds for operators of interest with explicit dependence in terms of the gap. 
A key part of the method is that we will need to consider a contour $\Gamma_T(s)$ that encloses the spectrum of interest for successive steps of the walk.
In particular, we will use the notation $\Gamma_T(s,k)$ to indicate a contour that encloses the spectrum of interest for $k+1$ successive steps of the walk.
Moreover, for $\Gamma_T(s,k)$ we will take the specific contour that passes in straight lines from the center through the gaps in the spectrum, as shown in \cref{fig:contour_1} and \cref{fig:contour_2}.
Those figures indicate that the contour is closed by an arc at radius 2.
We will take the closure of the contour to be at a distance that approaches infinity for the contours $\Gamma_T(s,k)$.
The results can be obtained by taking the closure at a finite radius then taking that radius to infinity, but for simplicity of the explanation we will not give that limit explicitly except for one illustrative example.
Note that we will only take this limit when the integrand approaches zero more quickly than $1/|z|$.
That will be true for all the contour integrals we consider except that for $P_T(s)$.

\begin{figure}
    \centering
    \includegraphics[width = 0.5\textwidth]{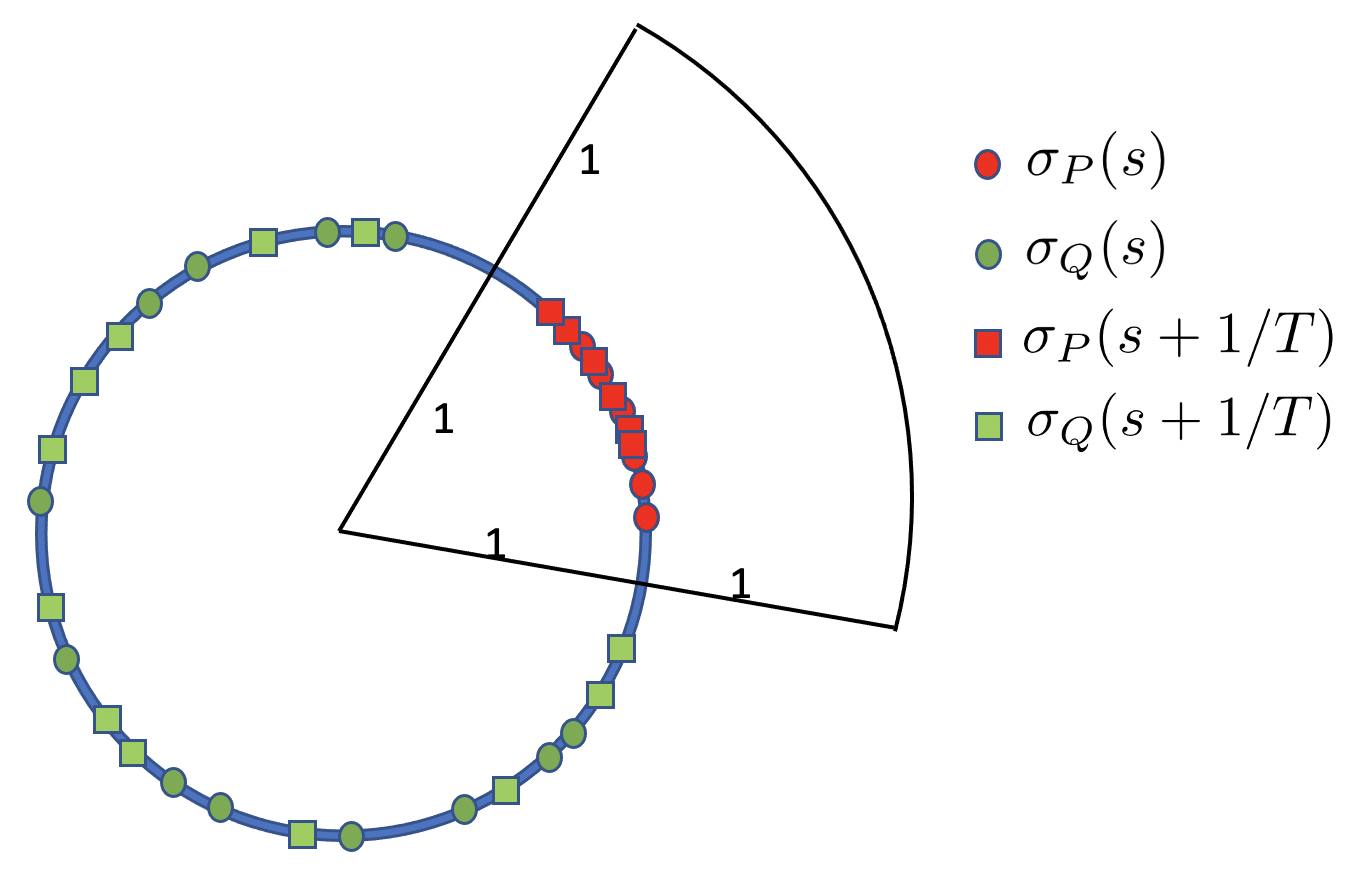}
    \caption{Illustration of the choice of the contour $\Gamma_T(s,k)$ for $k = 1$.
    That is, there are two successive steps of the walk, and we would need to consider the spectrum for both.
    We need to be able to use a contour that separates out the spectrum of interest for both steps of the walk.
    This ensures that we have projectors onto the spectrum of interest that are consistent for both steps, with a gap between the contour and the eigenvalues.
    We do not allow eigenvalues of interest to cross the gap between one step and the next.
    Again we show a contour with radius $2$, but we would take the infinity limit of the radius. }
    \label{fig:contour_2}
\end{figure}

We start with the bounds for $DP_T$ and $D^{(2)}P_T$, which can be obtained by direct calculations from the definitions. 
\begin{lemma}\label{lem:P}
    For any integer $T$ and $n$ and the corresponding discrete time $s=n/T$ we have 
    \begin{equation}
    \label{eq:Dp}
        \|DP_T(s)\|\leq \frac{2c_1(s)}{T\Delta_1(s)},
    \end{equation}
    and 
    \begin{equation}\label{eqn:D2P}
        \|D^{(2)}P_T(s)\| \leq \frac{\mathcal{G}_{T,1}(s)}{T^2}, 
    \end{equation}
    with
    \begin{equation}\label{eq:GT1}
        \mathcal{G}_{T,1}(s) \coloneqq \frac{ c_1(s)^2+c_1(s)c_1(s+1/T)}{\pi (1-\cos(\Delta_2(s)/2))}  + \frac{2c_2(s)}{\Delta_2(s)}.
    \end{equation}
    \end{lemma}
    See \cref{ap:proofDP} for the lengthy proof of this result.
Now we move on to bounding the finite difference of the kernel function and the adiabatic walk operator. 
The key here is to express and bound the operator $V_T$, because it is related to both the kernel and the adiabatic walk operator.
First, we re-express $V_T$ in terms of $P_T$.

\begin{lemma}\label{lem:V_v2}
 For a discrete time $s$, we have  
   \begin{equation}\label{eqn:def_remainder_V}
       V_T(s) = \mathcal{F}_T(s)\left[I + DP_T(s)(2P_T(s)-I)\right], 
   \end{equation}
where
\begin{equation}
\label{eq:def_F}
    \mathcal{F}_T(s) \coloneqq \left[I-\left(DP_T(s)\right)^2 \right]^{-1/2}.
\end{equation}
\end{lemma}
\begin{proof}
    By the definition of $V_T(s)$ in \cref{eq:V}, 
\begin{align}
\label{eq:Vtaprox}
  V_T(s)&=\left[I-\left(P_T(s)-P_T\left(s+1/T\right)\right)^2\right]^{-1/2} S_T\left(s+1/T,s\right) \nonumber\\
  &=\left[I-\left(P_T(s)-P_T\left(s+1/T\right)\right)^2\right]^{-1/2} \left[I-P_T(s)-P_T\left(s+1/T\right)+2P_T\left(s+1/T\right) P_T(s)\right] \nonumber\\
  &=\left[I-\left(P_T(s)-P_T\left(s+1/T\right)\right)^2 \right]^{-1/2} \left[I+\left(P_T\left(s+1/T\right)-P_T(s)\right)\left(2P_T(s)-I\right)\right]\nonumber \\
  &=\mathcal{F}_T(s)\left[I + DP_T(s)(2P_T(s)-I)\right].
\end{align}
\end{proof}

That enables us to bound the difference of $V_T$ from the identity.

\begin{lemma}\label{lem:V_bound2}
    For a discrete time $s$  
    \begin{equation}
       \|V_T(s) - I\| \leq \|\mathcal{F}_T(s) - I\| +  \|DP_T(s)\|\|\mathcal{F}_T(s) \|.
   \end{equation}
\end{lemma}
\begin{proof}
From \cref{lem:V_v2} we have
\begin{equation}
V_T(s) - I  = \mathcal{F}_T(s) + \mathcal{F}_T(s)(DP_T(s)(2P_T(s)-I)) -I.
\end{equation}
Then the triangle inequality gives 
\begin{align}
\|V_T(s) - I\|  &\leq \|\mathcal{F}_T(s)-I\| + \|\mathcal{F}_T(s)(DP_T(s)(2P_T(s)-I))\| \nn
&= \|\mathcal{F}_T(s)-I\| + \|\mathcal{F}_T(s)DP_T(s)\| ,
\end{align}
where in the second line we have used the fact that $2P_T(s)-I$ is a unitary reflection operator.
The inequality $\|\mathcal{F}_T(s)DP_T(s)\| \leq \|\mathcal{F}_T(s)\| \, \|DP_T(s)\|$ then gives the bound required.
\end{proof}

Next we bound the change in $V_T$.

\begin{lemma}\label{lem:DV_v2}
For a discrete time $s$,
\begin{align}
    \left\|DV_T(s)\right\| &\leq \left(1+\|DP_T(s+1/T)\|\right)\left\|D^{(2)}P_T(s)\right\| \mathcal{D}_3\left( \max(\|DP_T(s+1/T)\|,\|DP_T(s)\|) \right)\nonumber\\
    &\quad+\|\mathcal{F}_T(s)\|\left(\left\|D^{(2)}P_T(s)\right\| + 2 \|DP_T(s)\|^2\right),
\end{align}
with
\begin{equation}
    \mathcal{D}_3(z)\coloneqq \frac{z}{(1-z^2)^{3/2}}.
\end{equation}
\end{lemma}
\begin{proof}
For the difference operator $D$ there is the product rule $D(X(s)Y(s))=DX(s)Y(s+1/T)+X(s)DY(s)$ for any two operators $X(s)$ and $Y(s)$. Then, from \cref{lem:V_v2}  we have 
\begin{align}
    D V_T(s) &= D\mathcal{F}_T(s) [I + DP_T(s+1/T)(2P_T(s+1/T)-I)] + \mathcal{F}_T(s)D\left[DP_T(s)(2P_T(s)-I)\right]\nonumber\\
    &=D\mathcal{F}_T(s) [I + DP_T(s+1/T)(2P_T(s+1/T)-I)]\nonumber\\  &\quad+ \mathcal{F}_T(s)\left[D^{(2)}P_T(s)(2P_T(s+1/T)-I)+2(DP_T(s))^2\right].
\end{align}
By the triangle inequality and using the fact that the reflection operator is unitary, we get
\begin{equation}\label{eq:l9tmp}
    \|D V_T(s)\|\leq\|D\mathcal{F}_T(s)\| \left(1 + \|DP_T(s+1/T)\|\right) + \|\mathcal{F}_T(s)\|\left(\|D^{(2)}P_T(s)\| + 2\|DP_T(s)\|^2\right).
\end{equation}
Now, from the Taylor expansion of $\mathcal{F}_T$
\begin{equation}\label{eqn:F_Taylor}
    \mathcal{F}_T(s) = I + \sum_{k=1}^{\infty} \frac{\Pi_{j=1}^{k}(2j-1)}{2^k k!}(DP_T(s))^{2k}, 
\end{equation}
the bound of $D\mathcal{F}_T$ in terms of $P_T$ can be computed, i.e.
\begin{align}
    \| \mathcal{F}_T(s+1/T) - \mathcal{F}_T(s) \| &= \left\| \sum_{k=1}^{\infty} \frac{\Pi_{i=1}^{k}(2i-1)}{2^k k!} \left[ \left(DP_T(s+1/T)\right)^{2k} - \left(DP_T(s)\right)^{2k}\right] \right\| \nn
    &= \left\| \sum_{k=1}^{\infty} \frac{\Pi_{i=1}^{k}(2i-1)}{2^k k!}  \sum_{j=0}^{2k-1}\left(DP_T(s+1/T)\right)^{j} \left[ DP_T(s+1/T)-DP_T(s)\right] \left(DP_T(s)\right)^{2k-1-j} \right\| \nn
    &\le \sum_{k=1}^{\infty} \frac{\Pi_{i=1}^{k}(2i-1)}{2^k k!} \left\|D^{(2)}P_T(s)\right\|\sum_{j=0}^{2k-1}\|\left(DP_T(s+1/T)\right)^{j}  \left(DP_T(s)\right)^{2k-1-j}\| \nn
    &\le \left\|D^{(2)}P_T(s)\right\|\sum_{k=1}^{\infty} \frac{\Pi_{i=1}^{k}(2i-1)}{2^k k!} (2k) \left[ \max(\|DP_T(s+1/T)\|,\|DP_T(s)\|) \right]^{2k-1} \nn
    &=  \left\|D^{(2)}P_T(s)\right\| \mathcal{D}_3\left( \max(\|DP_T(s+1/T)\|,\|DP_T(s)\|) \right),
\end{align}
where we have used the Taylor expansion of the function $\mathcal{D}_3$.
Substituting this into \cref{eq:l9tmp} gives the bound required.
\end{proof}

\begin{lemma}\label{lem:WAv3}
For any discrete time $s$ with $W_T^A(s)$ defined as in \cref{eq:Wa}, we have
\begin{equation}
\label{eqn:DWA}
\|DW_T^A(s)\| \leq \frac{c_1(s)}{T}+ \|DV_T(s)\|. 
\end{equation}
\end{lemma}
\begin{proof}
According to the definition of $W_T^A(s)$ as $V_T(s)W_T(s)$,
\begin{align}
DW_T^A(s) = DV_T(s) W_T(s+1/T) + V_T(s) DW_T(s). \end{align}
Since $W_T$ and $V_T$ are unitary, and using the triangle inequality, we have
\begin{equation}
    \left\|DW_T^A(s)\right\| \leq \left\|DV_T(s)\right\| + \left\|DW_T(s)\right\|. \nonumber\\
   \label{eq:upperDWa}
    \end{equation}
Using \cref{eqn:assump1} with $k=1$ for $\left\|DW_T(s)\right\|$ then gives \cref{eqn:DWA}.
\end{proof}

\begin{lemma}\label{lem:DOmega2}
    For any discrete time $s$, $\Omega_T$ as defined in \cref{eq:Omegadef}, and $\mathcal{F}_T$ defined as in \cref{eq:def_F}, we have the upper bound on $D\Omega_T(s)$
    \begin{equation}
    \|D\Omega_T(s)\|  \leq \|\mathcal{F}_T(s) - I\| +  \|DP_T(s)\|\|\mathcal{F}_T(s) \|. 
    \end{equation}
\end{lemma}
\begin{proof}
    From Proposition~\ref{prop:volterra} and \cref{eqn:Volterra} we see that the difference between $\Omega_T\left(s\right)$ and $\Omega_T\left(s+1/T\right)$ is the term in the sum with $m=n$.
    That gives
    \begin{align}
        D\Omega_T\left(s\right) &= - \frac{1}{T} K_T(s)\Omega_T(s) \nonumber\\
        &=  (\Theta_T(s)-I)\Omega_T(s) \nonumber\\
        &= {U_T^A}(s+1/T)^{\dagger} (V_T(s)-I)^{\dagger}U_T^A(s+1/T)\Omega_T(s). 
    \end{align}
    In the third line we have used \cref{eqn:Theta_V} for $\Theta_T(s)$.
    Since $U_T^A$ and $\Omega_T$ are unitary, we have
    \begin{equation}
        \| D\Omega_T\left(s\right) \| = \|(V_T(s)-I)^{\dagger}\| = \|V_T(s)-I\|.
    \end{equation}
    Then the desired bound follows from \cref{lem:V_bound2}. 
\end{proof}

Finally we summarize the bounds for the operators related to the summation by parts formula. 
\begin{lemma}\label{lem:Xt}
    For any discrete time $s$ in $\tilde{X}(s)$ as defined in \cref{eq:Xtilde}, and any bounded operator $X(s)$, we have
    \begin{equation}\label{eqn:Xt_bound}
    \left\|\tilde{X}(s)\right\| \leq \frac{2}{\Delta_0(s)}\|X(s)\|,
    \end{equation}
    and 
    \begin{equation}
    \left\|D\Tilde{X}(s)\right\| \leq \frac{2}{\Delta_1(s)}\left\|DX(s)\right\| + \frac{2c_1(s)}{\pi T (1-\cos(\Delta_1(s)/2))} \|X(s)\|. 
    \end{equation}
\end{lemma}
\begin{proof}
    The bound for $\tilde{X}$ directly follows from the definition \cref{eq:Xtilde} and choosing an appropriate contour $\Gamma_T(s,0)$.
    As shown in \cref{fig:contour_1}, the contour passes in a straight line from the centre through both gaps, and has a circular arc of radius 2 between these two straight lines.
    That is, \cref{eq:Xtilde} gives
    \begin{align}
        \|\Tilde{X}(s)\| & \le \frac{1}{2\pi}\oint_{\Gamma_T(s,0)} \| R_T(s,z)\|^2 \|X(s)\| \, |dz| \nn
       & \le \frac{1}{2\pi} \|X(s)\| \frac{4\pi}{\Delta_0(s)} = \frac{2}{\Delta_0(s)}\|X(s)\|,
    \end{align}
    where we have used \cref{eq:CountInt1}, but replaced $\Delta_1(s)$ with $\Delta_0(s)$ because we need only consider the eigenvalues for a single step of the walk.
    
    For $D\Tilde{X}(s)$, using $\Gamma_T(s,1)$ (for two consecutive steps of the walk), and using
    \cref{eq:Xtilde} we have 
    \begin{align}
    D\Tilde{X}(s)&=-\frac{1}{2\pi i}\oint_{\Gamma_T(s,1)}\left(R_T\left(s+\frac{1}{T},z\right)X\left(s+\frac{1}{T}\right)R_T\left(s+\frac{1}{T},z\right)-R_T\left(s,z\right)X\left(s\right)R_T\left(s,z\right)\right)dz\nonumber\\
    &=-\frac{1}{2\pi i}\oint_{\Gamma_T(s,1)}R_T\left(s+\frac{1}{T},z\right)DX\left(s\right)R_T\left(s+\frac{1}{T},z\right)dz\nonumber\\
    &\quad -\frac{1}{2\pi i}\left(\oint_{\Gamma_T(s,1)}R_T\left(s,z\right)X\left(s\right)DR_T\left(s,z\right)dz+\oint_{\Gamma_T(s,1)}DR_T\left(s,z\right)X\left(s\right)R_T\left(s+\frac{1}{T},z\right)dz\right).
    \end{align}
Using \cref{eq:contProj} we have
\begin{equation}
    D R_T(s,z) = - R_T\left(s+\frac{1}{T},z\right)DW_T\left(s\right) R_T\left(s,z\right),
\end{equation}
so
\begin{align}
    \left\|DR_T\left(s,z\right)\right\| &\leq 
    \left\| R_T\left(s+\frac{1}{T},z\right)\right\| \left\|DW_T\left(s\right)\right\|  \left\| R_T\left(s,z\right)\right\| \\
& \leq    \frac{c_1(s)}{T}\left\|R_T\left(s+\frac{1}{T},z\right)\right\|\left\|R_T\left(s,z\right)\right\| .
\end{align}
We can therefore write an upper bound as
\begin{align}
        \|D\Tilde{X}(s)\| &\le \frac{\left\| DX\left(s\right)\right\|}{2\pi} \oint_{\Gamma_T(s,1)}\left\|R_T\left(s+\frac{1}{T},z\right)\right\|^2   |dz| \nonumber\\
    &\quad +\frac{\left\| X\left(s\right)\right\|}{2\pi}\left(\oint_{\Gamma_T(s,1)}\left\| R_T\left(s,z\right)\right\|^2  \left\|R_T\left(s+\frac 1T,z\right)\right\| |dz| + \oint_{\Gamma_T(s,1)}\left\| R_T\left(s,z\right)\right\|  \left\|R_T\left(s+\frac{1}{T},z\right)\right\|^2 |dz|\right).
\end{align}
Using the bounds on the contour integrals given in \cref{eq:CountInt1} and \cref{eq:CountInt2}, we then get
\begin{align}
    \left\|D\Tilde{X}(s)\right\| &\leq \frac{1}{2\pi} \left\|DX(s)\right\|\frac{4\pi}{\Delta_1(s)} + \frac{1}{\pi} \|X(s)\| \frac{c_1(s)}{T} \frac{2}{1-\cos(\Delta_1(s)/2)} \nonumber\\
    & =  \frac{2}{\Delta_1(s)}\left\|DX(s)\right\| + \frac{2c_1(s)}{\pi T (1-\cos(\Delta_1(s)/2))} \|X(s)\|. 
\end{align}
\end{proof}

\begin{lemma}\label{lem:ABZ2_v2}
    For a discrete time $s$ in $A(s)$, $B(s)$ and $Z(s)$ defined in \cref{eq:A,eq:B,eq:Z} respectively, and any bounded operator $X(s)$, we have
    \begin{equation}
        \label{eq:boundA}
     \|A(s)\|\leq \|\mathcal{F}_T(s) - I\| +  \|DP_T(s)\|\|\mathcal{F}_T(s) \|, 
    \end{equation}
    \begin{align}
    \label{eq:BoundB}
    \left\|B\left(s\right)\right\| &\leq \frac{2}{\Delta_1(s)}\left\|DX(s)\right\| + \frac{2c_1(s)}{\pi T (1-\cos(\Delta_1(s)/2))} \|X(s)\|\nonumber\\ 
    &\quad+ \frac{2}{\Delta_0(s)}\left (\frac{c_1(s-1/T)}{T}+ \|DV_T(s-1/T)\|\right)\|X(s)\|, 
   \end{align}
   and 
   \begin{align}\label{eq:Zbound}
       \|Z(s)\| &\leq \frac{4T}{\Delta_0(s)} \left(\|\mathcal{F}_T(s) - I\| +  \|DP_T(s)\|\|\mathcal{F}_T(s) \|\right)\|X(s)\| + \frac{2T}{\Delta_1(s)}\left\|DX(s)\right\| \nonumber\\
       &\quad+ \frac{2c_1(s)}{\pi  (1-\cos(\Delta_1(s)/2))} \|X(s)\| + \frac{2}{\Delta_0(s)}\left (c_1(s-1/T)+ T\|DV_T(s-1/T)\|\right)\|X(s)\|. 
   \end{align}
\end{lemma}
\begin{proof}
    From the definition of $A$ in \cref{eq:A}, we have
\begin{equation}
    \|A(s)\|=\left\|\left(V_T\left(s\right)^{\dagger}-I\right)W_T^A(s) \right\|\leq \left\|\left(V_T\left(s\right)^{\dagger}-I\right) \right\| \left\|W_T^A(s) \right\|
    = \left\|\left(V_T\left(s\right)^{\dagger}-I\right) \right\|. 
\end{equation}
    The bound for $\|A\|$ follows from \cref{lem:V_bound2}. 
 
 For $B$, from \cref{eq:B}, we have that
\begin{equation}
    \left\|B\left(s\right)\right\| \leq \left\| D\Tilde{X}\left(s\right)W_T^A\left(s\right)\right\|+\left\| DW_T^A\left(s-1/T\right)\Tilde{X}\left(s\right)\right\|
    \leq \left\| D\Tilde{X}\left(s\right)\right\|+\left\| DW_T^A\left(s-1/T\right)\right\| \left\|\Tilde{X}\left(s\right)\right\|.
\end{equation}
By inserting the bounds previously computed for $D\tilde{X}$ (in \cref{lem:Xt}) and $DW_T^A$ (in \cref{lem:WAv3}), the desired bound for $B$ is established. 

Finally, for $Z$ the definition \cref{eq:Z} immediately gives
\begin{equation}
    \left\|Z\left(s\right)\right\|\leq T\left(2\left\|A\left(s\right)\right\| \left\|\tilde{X}\left(s\right)\right\| + \left\|B\left(s\right)\right\|\right).
\end{equation}
The bounds previously computed in \cref{eq:BoundB,eqn:Xt_bound,eq:boundA} then give the upper bound required.
\end{proof}

The summation by parts formula, which is presented as our next lemma here, is given in Theorem 1 of \cite{DKS98} with a typo in the sign of both operators $\mathcal{S}$ and $\mathcal{B}$.
Here we correct the sign slightly differently for the two quantities, taking $\mathcal{B}$ to be the negative of the $\mathcal{B}$ defined in \cite{DKS98}, and taking $\mathcal{S}$ to be the same but placing a minus sign in the statement of the theorem (so there is $\mathcal{B}-\mathcal{S}/T$).
Throughout the lemma and its proof we will encounter slight shift of the discrete time very frequently. To simplify the notation, for any positive integer $n$, we define $n_+ = n+1$ and $n_- = n-1$. 

\begin{lemma}[Summation by parts formula]
\label{lem:sum_by_parts}
Let $W_T(s)$, $s\in \mathbb{Z}/T$, be a sequence of unitaries, 
and suppose that $X(s)$ and $Y(s)$ are sequences of operators. 
Then
\begin{equation}
    \sum_{n=1}^{l}Q_0 U^{A\dagger}_T\left(\frac{n}{T}\right )X\left(\frac{n}{T}\right) U^{A}_T\left(\frac{n}{T}\right) P_0 Y\left(\frac{n}{T}\right) = \mathcal{B} - \frac{1}{T}\mathcal{S},
\end{equation}
where $P_0=P_T(0)$ and $Q_0=Q_T(0)$,
\begin{equation}
\label{eq:Bou}
    \mathcal{B} = Q_0U^{A\dagger}_T\left(\frac{l}{T}\right)\Tilde{X}\left(\frac{l_+}{T}\right)U^{A}_T\left(\frac{l_+}{T}\right)P_0Y\left(\frac{l_+}{T}\right) - Q_0U_T^{A\dagger}(0)\Tilde{X}\left(\frac{1}{T}\right)U^{A}_T\left(\frac{1}{T}\right)P_0Y\left(\frac{1}{T}\right),
\end{equation}
is the boundary term and
\begin{equation}
\label{eq:Sum}
\mathcal{S}=\sum_{n=1}^{l} Q_0 U^{A\dagger}_T\left(\frac{n}{T}\right) \left(Z\left(\frac{n}{T}\right) U_T^{A}\left(\frac{n}{T}\right) P_0 Y\left(\frac{n}{T}\right) + \Tilde{X}\left(\frac{n_+}{T}\right) W^A_T\left(\frac{n}{T}\right) U_T^A\left(\frac{n}{T}\right) P_0 T DY\left(\frac{n}{T}\right)\right),
\end{equation}
is the sum. 
\end{lemma}
As we are making a correction to the theorem, and it is quite lengthy, we give a proof in \cref{ap:sum_parts}.

\subsection{The first discrete adiabatic theorem}
\label{sec:dat1}

We now give the complete explicit form of the discrete adiabatic theorem.

 \begin{theorem}[The First Discrete Adiabatic Theorem]
 \label{theoAdia} 
 Let $U_T(s) = \prod_{l = 0}^{sT-1} W_T\left(l/T\right)$ for $s\in \mathbb{Z}/T$ be a product of unitary operators $W_T\left(l/T\right)$ as per
 \cref{eq:U_as_product_of_W}, and let $U_T^A(s)$ be the corresponding ideal adiabatic evolution that maps an eigenstate of $W_T(0)$ to the corresponding eigenstate of $W_T(s)$.
 Suppose further that the operators $W_T (s)$ satisfy
 $\left\|D^{(k)}W_T(s)\right\|\leq c_k(s)/T^k$ for $k = 1,2$, as per \cref{def:difs}, 
 we consider the gaps $\Delta_k(s)$ as defined in \cref{def:gaps}, 
 and $T \geq \max_{s\in [0,1]} (2c_1(s)/\Delta_1(s))$.
 Then for any time $s$, we have
 \begin{align}
 \label{eq:mainEqTheo}
            & \quad \left\|U_T (s) - U_T^{A}(s) \right\| \nn
            & \leq\frac{4}{\Delta_0(1/T)}\mathcal{D}_2\left(\frac{2c_1(0)}{T\Delta_1(0)}\right) +  \frac{4}{\Delta_0(s)}\mathcal{D}_2\left(\frac{2c_1(s-1/T)}{T\Delta_1(s-1/T)}\right) + 2\mathcal{D}_2\left(\frac{2c_1(s-1/T)}{T\Delta_1
       (s-1/T)}\right) \nonumber\\
       & \quad + \sum_{n=1}^{sT-1}4\left(\frac{1}{\Delta_0(n_+/T)} + \frac{2}{\Delta_0(n/T)}\right)\mathcal{D}_2\left(\frac{2c_1(n/T)}{T\Delta_1(n/T)}\right)\mathcal{D}_2\left(\frac{2c_1(n_-/T)}{T\Delta_1(n_-/T)}\right) \nonumber\\
       &\quad + \sum_{n=1}^{sT-1}\frac{4\mathcal{G}_{T,3}(n_-/T)}{T^2\Delta_1(n/T)} + \sum_{n=1}^{sT-1} \frac{4c_1(n/T)}{\pi T (1-\cos(\Delta_1(n/T)/2))} \mathcal{D}_2\left(\frac{2c_1(n_-/T)}{T\Delta_1
       (n_-/T)}\right)\nonumber\\
       &\quad+ \sum_{n=1}^{sT-1}\frac{4 \mathcal{G}_{T,4}(n_-/T)}{T\Delta_0(n/T)}\mathcal{D}_2\left(\frac{2c_1(n_-/T)}{T\Delta_1
       (n_-/T)}\right)  + \sum_{n=0}^{sT-1} \frac{24c_1(n/T)^2}{T^2\Delta_1(n/T)^2} +  \sum_{n=0}^{sT-1}\frac{4c_1(n/T)^2}{T^2\Delta_1(n/T)^2}\left(1-\frac{2c_1(n/T)}{T\Delta_1(n/T)}\right)^{-1}, \end{align}
where
\begin{align}
\label{eq:D_i}
        \mathcal{D}_1(z) &\coloneqq \frac{1}{\sqrt{1-z^2}}, \quad \mathcal{D}_2(z) \coloneqq \sqrt{\frac{1+z}{1-z}} - 1, \quad \mathcal{D}_3(z)\coloneqq \frac{z}{(1-z^2)^{3/2}},\\
    \label{eq:G1}
        \mathcal{G}_{T,1}(s) &\coloneqq \frac{ c_1(s)^2+c_1(s)c_1(s+1/T)}{\pi (1-\cos(\Delta_2(s)/2))}  + \frac{2c_2(s)}{\Delta_2(s)}, \\
        \label{eq:G2}
        \mathcal{G}_{T,2}(s) &\coloneqq \mathcal{G}_{T,1}(s)\mathcal{D}_3\left( \max\left(\frac{2c_1(s+1/T)}{T\Delta_1(s+1/T)},\frac{2c_1(s)}{T\Delta_1(s)}\right) \right), \\
        \label{eq:G3}
         \mathcal{G}_{T,3}(s) &\coloneqq \mathcal{G}_{T,2}(s)\left(1 + \frac{2c_1(s)}{T\Delta_1(s)}\right) +  \mathcal{D}_1\left(\frac{2c_1(s)}{T\Delta_1(s)}\right)\left(\mathcal{G}_{T,1}(s) +  \frac{8c_1(s)^2}{\Delta_1(s)^2}\right), \\
         \label{eq:G4}
         \mathcal{G}_{T,4}(s) &\coloneqq \frac{\mathcal{G}_{T,3}(s)}{T} + c_1(s). 
    \end{align}
\end{theorem} 

\begin{proof}
Starting from the definition of $K_T$ and Proposition~\ref{prop:volterra}, for any discrete time $s$, 
\begin{align}
    \|U_T(s)-U_T^A(s)\| &= \|\Omega_T(s)-I\| \nonumber\\
    &= \left\|\frac{1}{T}\sum_{n=0}^{sT-1}K_T\left(\frac{n}{T}\right)\Omega_T\left(\frac{n}{T}\right)\right\| \nonumber\\
    &= \left\|\sum_{n=0}^{sT-1}\left(I - \Theta_T\left(\frac{n}{T}\right)\right)\Omega_T\left(\frac{n}{T}\right)\right\| \nonumber\\
    &= \left\|\sum_{n=1}^{sT}{U_T^A}^{\dagger}\left(\frac{n}{T}\right)\left(I - V_T^{\dagger}\left(\frac{n_-}{T}\right)\right){U_T^A}\left(\frac{n}{T}\right)\Omega_T\left(\frac{n_-}{T}\right)\right\|. 
\end{align}
Note that in the summation by parts formula only the ``off-diagonal'' term is considered. 
This motivates us to further split the sum into ``diagonal'' and ``off-diagonal'' terms as 
    \begin{align}
    \|U_T(s)-U_T^A(s)\|&= \left\|\sum_{n=1}^{sT}{(P_0+Q_0)U_T^A}^{\dagger}\left(\frac{n}{T}\right)\left(I - V_T^{\dagger}\left(\frac{n_-}{T}\right)\right){U_T^A}\left(\frac{n}{T}\right)(P_0+Q_0)\Omega_T\left(\frac{n_-}{T}\right)\right\| \label{eqn:main_bound_split}\\ &\leq \left\|\sum_{n=1}^{sT}P_0{U_T^A}^{\dagger}\left(\frac{n}{T}\right)\left(I - V_T^{\dagger}\left(\frac{n_-}{T}\right)\right){U_T^A}\left(\frac{n}{T}\right)P_0\Omega_T\left(\frac{n_-}{T}\right)\right\| \label{eq:diag1}\\
    & \quad + \left\|\sum_{n=1}^{sT}Q_0{U_T^A}^{\dagger}\left(\frac{n}{T}\right)\left(I - V_T^{\dagger}\left(\frac{n_-}{T}\right)\right){U_T^A}\left(\frac{n}{T}\right)Q_0\Omega_T\left(\frac{n_-}{T}\right)\right\| \label{eq:diag2}\\
    & \quad + \left\|\sum_{n=1}^{sT}Q_0{U_T^A}^{\dagger}\left(\frac{n}{T}\right)\left(I - V_T^{\dagger}\left(\frac{n_-}{T}\right)\right){U_T^A}\left(\frac{n}{T}\right)P_0\Omega_T\left(\frac{n_-}{T}\right)\right\| \label{eq:offdiag1}\\
    & \quad + \left\|\sum_{n=1}^{sT}P_0{U_T^A}^{\dagger}\left(\frac{n}{T}\right)\left(I - V_T^{\dagger}\left(\frac{n_-}{T}\right)\right){U_T^A}\left(\frac{n}{T}\right)Q_0\Omega_T\left(\frac{n_-}{T}\right)\right\|\label{eq:offdiag2},
    \end{align}
where \cref{eq:diag1,eq:diag2} are the diagonal  and  \cref{eq:offdiag1,eq:offdiag2} are the off-diagonal components. For the ``diagonal'' term, it is possible to show
\begin{align}
        &\quad \left\|\sum_{n=1}^{sT}P_0{U_T^A}^{\dagger}\left(\frac{n}{T}\right)\left(I - V_T^{\dagger}\left(\frac{n_-}{T}\right)\right){U_T^A}\left(\frac{n}{T}\right)P_0\Omega_T\left(\frac{n_-}{T}\right)\right\| \nonumber\\
        &\leq \sum_{n=0}^{sT-1} \left\|I-\mathcal{F}_T\left(\frac{n}{T}\right)\right\|\left(1+\left\|DP_T\left(\frac{n}{T}\right)\right\|\right) + 3\sum_{n=0}^{sT-1}\left\|DP_T\left(\frac{n}{T}\right)\right\|^2. 
\end{align}
This is shown in \cref{ap:sec_diag}, where the result is given in \cref{eq:diag_p2}.
The reasoning for the term with $Q_0$ is identical and gives the same result.
Using \cref{lem:V_v2} one can show
\begin{equation}
    \|\mathcal{F}_T(s) - I\|\leq\mathcal{D}_1\left(\frac{2c_1(s)}{T\Delta_1(s)}\right)-1 .
\end{equation}
The steps for deriving the above bound are given in \cref{eq:bounF-i} and the function $\mathcal{D}_1$ is defined in \cref{eq:D_i}. Therefore, one obtains the following bound for the ``diagonal'' term
\begin{align}
    & \quad \left\|\sum_{n=1}^{sT}P_0{U_T^A}^{\dagger}\left(\frac{n}{T}\right)\left(I - V_T^{\dagger}\left(\frac{n_-}{T}\right)\right){U_T^A}\left(\frac{n}{T}\right)P_0\Omega_T\left(\frac{n_-}{T}\right)\right\| \nonumber\\
    &\leq \sum_{n=0}^{sT-1} \frac{12c_1(n/T)^2}{T^2\Delta_1(n/T)^2} +  \sum_{n=0}^{sT-1}\left(\mathcal{D}_1\left(\frac{2c_1(n/T)}{T\Delta_1(n/T)}\right)-1\right)\left(1+\frac{2c_1(n/T)}{T\Delta_1(n/T)}\right) \nonumber\\
    & \leq \sum_{n=0}^{sT-1} \frac{12c_1(n/T)^2}{T^2\Delta_1(n/T)^2} +  \sum_{n=0}^{sT-1}\frac{2c_1(n/T)^2}{T^2\Delta_1(n/T)^2}\left(1-\frac{2c_1(n/T)}{T\Delta_1(n/T)}\right)^{-1}, \label{eqn:diagonal} 
\end{align}
where in the last line we use the inequality $[(1-z^2)^{-1/2}-1](1+z)\le z^2/[2(1-z)]$ for all $0 \leq z < 1 $.
The exact same bound holds for the second diagonal term with $Q_0$ in \cref{eq:diag2}.
The reason is that we have treated the eigenspace of interest and the complementary eigenspace completely symmetrically.
Therefore exactly the same bounds hold with $P_T$ replaced with $Q_T$, and the above bound must continue to hold.

For the ``off-diagonal'' term we can similarly consider only the term with $Q_0$ on the left and $P_0$ on the right as in \cref{eq:offdiag1}, and exactly the same bound will hold for the other off-diagonal term in \cref{eq:offdiag2}.
Using \cref{lem:sum_by_parts} with $X(s) = T(1-V_T^{\dagger}(s-1/T))$ and $Y(s) = \Omega_T(s-1/T)$ (note the slight shift in time), it is possible to show
\begin{align}
    & \quad \left\|\sum_{n=1}^{sT}Q_0{U_T^A}^{\dagger}\left(\frac{n}{T}\right)\left(I - V_T^{\dagger}\left(\frac{n_-}{T}\right)\right){U_T^A}\left(\frac{n}{T}\right)P_0\Omega_T\left(\frac{n_-}{T}\right)\right\|\nonumber\\ 
     & \leq \frac{1}{T}\left\|\tilde{X}\left(\frac{1}{T}\right)\right\| + \frac{1}{T}\left\|\tilde{X}\left(s\right)\right\| + \frac{1}{T}\left\|X\left(s\right)\right\|
     + \frac{1}{T^2}\sum_{n=1}^{sT-1} \left\|Z\left(\frac{n}{T}\right)\right\|  + \frac{1}{T}\sum_{n=1}^{sT-1} \left\|\tilde{X}\left(\frac{n_+}{T}\right)\right\| \left\|DY\left(\frac{n}{T}\right)\right\|. 
\end{align}
See \cref{ap:offD} for the derivation and the result in \cref{eq:offDiag_B}.
By using \cref{lem:P,lem:DOmega2,lem:V_bound2,lem:DV_v2,lem:Xt,lem:ABZ2_v2}, we can show
\begin{align}
    \left\|X\left(\frac{n}{T}\right)\right\| &\leq T\mathcal{D}_2\left(\frac{2c_1(n_-/T)}{T\Delta_1
       (n_-/T)}\right), \\
      \left\|\tilde{X}\left(\frac{n}{T}\right)\right\| &\leq \frac{2T}{\Delta_0(n/T)}\mathcal{D}_2\left(\frac{2c_1(n_-/T)}{T\Delta_1(n_-/T)}\right), \\
        \left\|Z\left(\frac{n}{T}\right)\right\| &\leq \frac{4T^2}{\Delta_0\left(n/T\right)} \mathcal{D}_2\left(\frac{2c_1(n/T)}{\Delta_1(n/T)}\right)\mathcal{D}_2\left(\frac{2c_1(n_-/T)}{\Delta_1(n_-/T)}\right) + \frac{2\mathcal{G}_{T,3}(n_-/T)}{\Delta_1(n/T)} \nonumber\\
       &\quad+ \frac{2Tc_1(n/T)}{\pi  (1-\cos(\Delta_1(n/T)/2))} \mathcal{D}_2\left(\frac{2c_1(n_-/T)}{T\Delta_1(n_-/T)}\right) + \frac{2T\mathcal{G}_{T,4}(n_-/T)}{\Delta_0(n/T)}\mathcal{D}_2\left(\frac{2c_1(n_-/T)}{T\Delta_1(n_-/T)}\right), \\
        \left\|DY\left(\frac{n}{T}\right)\right\| &\leq \mathcal{D}_2\left(\frac{2c_1(n_-/T)}{T\Delta_1(n_-/T)}\right),
\end{align}
with $\mathcal{D}_2(x)$, $\mathcal{G}_{T,3}(n/T)$ and $\mathcal{G}_{T,4}(n/T)$ given in \cref{eq:D_i,eq:G3,eq:G4} respectively. The bounds of these operators are derived in \cref{ap:offD}; see \cref{eq:Xbnd,eq:Xtildebnd,eq:Zbnd,eq:Ybnd}.

Therefore
\begin{align}
    & \quad \left\|\sum_{n=1}^{sT}Q_0{U_T^A}^{\dagger}\left(\frac{n}{T}\right)\left(I - V_T^{\dagger}\left(\frac{n_-}{T}\right)\right){U_T^A}\left(\frac{n}{T}\right)P_0\Omega_T\left(\frac{n_-}{T}\right)\right\| \nonumber\\
    & \leq \frac{2}{\Delta_0(1/T)}\mathcal{D}_2\left(\frac{2c_1(0)}{T\Delta_1(0)}\right) +  \frac{2}{\Delta_0(s)}\mathcal{D}_2\left(\frac{2c_1(s-1/T)}{T\Delta_1(s-1/T)}\right) + \mathcal{D}_2\left(\frac{2c_1(s-1/T)}{T\Delta_1
       (s-1/T)}\right) \nonumber\\
       & \quad + \sum_{n=1}^{sT-1}2\left(\frac{1}{\Delta_0(n_+/T)} + \frac{2}{\Delta_0(n/T)}\right)\mathcal{D}_2\left(\frac{2c_1(n/T)}{T\Delta_1(n/T)}\right)\mathcal{D}_2\left(\frac{2c_1(n_-/T)}{T\Delta_1(n_-/T)}\right) \nonumber\\
       &\quad + \sum_{n=1}^{sT-1}\frac{2\mathcal{G}_{T,3}(n_-/T)}{T^2\Delta_1(n/T)} + \sum_{n=1}^{sT-1} \frac{2c_1(n/T)}{\pi T (1-\cos(\Delta_1(n/T)/2))} \mathcal{D}_2\left(\frac{2c_1(n_-/T)}{T\Delta_1
       (n_-/T)}\right)\nonumber\\ 
       &\quad+ \sum_{n=1}^{sT-1}\frac{2 \mathcal{G}_{T,4}(n_-/T)}{T\Delta_0(n/T)}\mathcal{D}_2\left(\frac{2c_1(n_-/T)}{T\Delta_1
       (n_-/T)}\right). \label{eqn:offdiagonal}
\end{align}
Finally, by using \cref{eqn:diagonal} and \cref{eqn:offdiagonal} in \cref{eqn:main_bound_split}, we obtain the required overall bound in \cref{eq:mainEqTheo}. 
\end{proof}

\subsection{Proof of the second adiabatic theorem}

Because the first form of the discrete adiabatic theorem is quite complicated, we give a simplified but looser form in \cref{cor:adia}.
In this subsection we provide the proof of that result.

\begin{proof}
The key ideas to obtain \cref{cor:adia} from \cref{theoAdia} are: replace the functions $c_1(s)$ and $c_2(s)$ by \cref{eq:chat}, which take into account neighbouring steps; replace the gaps $\Delta_k(s)$ by $\check{\Delta}(s)$ as defined in \cref{eq:fhat}, which takes into account the minimum gap in neighbouring steps; and bounding the higher-order terms by lower-order terms with a slightly more strict assumption on $T$, that it is no less than $\max_s (4\hat{c}_1(s)/\check{\Delta}(s))$. 

We first bound the functions $\mathcal{D}_k$ with simpler expressions. 
Recall that the definitions of $\mathcal{D}_k$ are
    \begin{equation}
        \mathcal{D}_1(z) = \frac{1}{\sqrt{1-z^2}}, \qquad \mathcal{D}_2(z) = \sqrt{\frac{1+z}{1-z}} - 1, \qquad \mathcal{D}_3(z)= \frac{z}{(1-z^2)^{3/2}}.
    \end{equation}
Notice that in \cref{theoAdia} all the arguments in $\mathcal{D}_k$ are in the form of $2c_1/(T\Delta_1)$, then, under the assumption on $T$, we are only interested in the case $0 \leq z \leq 1/2$. 
Then we have 
\begin{align}
\label{eq:upperB_D}
    \mathcal{D}_1(z) \leq \xi_1, \qquad
    \mathcal{D}_2(z) \leq \xi_2 z, \qquad
    \mathcal{D}_3(z)  \leq \xi_3 z
\end{align}
with constants $\xi_1 = 2/\sqrt{3}, \xi_2 = 2\sqrt{3}-2, \xi_3 = 8/(3\sqrt{3})$. 

Now we move on to the functions $\mathcal{G}_{T,k}$. 
For any positive integer $n$, from \cref{theoAdia} we need to bound $\mathcal{G}_{T,3}(n_-/T)$ and $\mathcal{G}_{T,4}(n_-/T)$, which in turn depend on $\mathcal{G}_{T,1}(n_-/T)$ and $\mathcal{G}_{T,2}(n_-/T)$. 
Using the inequality $1-\cos(\theta/2) = 2\sin^2(\theta/4) \geq \theta^2/\pi^2$ for all $0\leq \theta \leq \pi$, it is possible to show that
\begin{equation}\label{eq:GT1body}
    \mathcal{G}_{T,1}(n_-/T) \leq  \frac{2\pi\hat{c}_1(n/T)^2}{\check{\Delta}(n/T)^2} + \frac{2\hat{c}_2(n/T)}{\check{\Delta}(n/T)}. 
\end{equation}
That leads to the following upper bounds for the other main functions, 
\begin{equation}\label{eq:GT2body}
    \mathcal{G}_{T,2}(n_-/T)  \leq  \frac{4\pi\xi_3 \hat{c}_1(n/T)^3}{T\check{\Delta}(n/T)^3} + \frac{4\xi_3 \hat{c}_1(n/T)\hat{c}_2(n/T)}{T\check{\Delta}(n/T)^2}, 
\end{equation}
\begin{equation}\label{eq:GT3body}
    \mathcal{G}_{T,3}(n_-/T) \leq  \left(3\pi\xi_3/2 + (2\pi+8)\xi_1\right)\frac{\hat{c}_1(n/T)^2}{\check{\Delta}(n/T)^2} + \left(3\xi_3/2+2\xi_1\right)\frac{\hat{c}_2(n/T)}{\check{\Delta}(n/T)}, 
\end{equation}
and 
\begin{align}\label{eq:GT4body}
    \mathcal{G}_{T,4}(n_-/T) \leq \left(3\pi\xi_3/2 + (2\pi+8)\xi_1\right)\frac{\hat{c}_1(n/T)^2}{T\check{\Delta}(n/T)^2} + \left(3\xi_3/2+2\xi_1\right)\frac{\hat{c}_2(n/T)}{T\check{\Delta}(n/T)} + \hat{c}_1(n/T). 
\end{align}
See \cref{sec:uppMainF} for the details, where these bounds are given in \cref{eq:GT1app,eq:GT2app,eq:GT3app,eq:GT4app}.
Plugging all these estimates back to \cref{theoAdia} and using $1-\cos(\theta/2) = 2\sin^2(\theta/4) \geq \theta^2/\pi^2$ again, it is possible to show that
\begin{align}\label{eq:Uerrbndbody}
     \|U_T(s) - U_T^A(s)\| 
    & \leq \frac{8\xi_2\hat{c}_1(0)}{T\check{\Delta}(0)^2}+  \frac{8\xi_2\hat{c}_1(s)}{T\check{\Delta}(s)^2} + \frac{4\xi_2\hat{c}_1(s)}{T\check{\Delta}(s)} +  \sum_{n=1}^{sT-1} \left(48\xi_2^2 + 6\pi \xi_3 + (8\pi+32)\xi_1 + 8\pi \xi_2\right) \frac{\hat{c}_1(n/T)^2}{T^2\check{\Delta}(n/T)^3} \nonumber \\
       & \quad + \sum_{n=1}^{sT-1}(6\xi_3+8\xi_1)\frac{\hat{c}_2(n/T)}{T^2\check{\Delta}(n/T)^2}+ \sum_{n=0}^{sT-1} \frac{(32+8\xi_2) \hat{c}_1(n/T)^2}{ T^2 \check{\Delta}(n/T)^2} \nonumber \\
       & \quad + \sum_{n=1}^{sT-1}\left(12\pi\xi_2\xi_3+(16\pi+64)\xi_1\xi_2\right) \frac{\hat{c}_1(n/T)^3}{T^3\check{\Delta}(n/T)^4}  + \sum_{n=1}^{sT-1} \left(12\xi_2\xi_3+16\xi_1\xi_2\right)\frac{\hat{c}_1(n/T)\hat{c}_2(n/T)}{T^3\check{\Delta}(n/T)^3} . 
       \end{align}
This result is derived in \cref{sec:uppMainF}, \cref{eq:Uerrbndapp}.
Finally, for a clear representation in terms of the gap, we slightly modify some terms with $T^3$ on the denominator to $T^2$ by using the bounds $\hat{c}_1(s)/(T\check{\Delta}(s)) \leq 1/4$, then 
\begin{align}
     \|U_T(s) - U_T^A(s)\| 
    & \leq \frac{8\xi_2\hat{c}_1(0)}{T\check{\Delta}(0)^2}+  \frac{8\xi_2\hat{c}_1(s)}{T\check{\Delta}(s)^2} + \frac{4\xi_2\hat{c}_1(s)}{T\check{\Delta}(s)} +  \sum_{n=1}^{sT-1} \left(48\xi_2^2 + 6\pi \xi_3 + (8\pi+32)\xi_1 + 8\pi \xi_2\right) \frac{\hat{c}_1(n/T)^2}{T^2\check{\Delta}(n/T)^3} \nonumber \\
       & \quad + \sum_{n=1}^{sT-1}(6\xi_3+8\xi_1)\frac{\hat{c}_2(n/T)}{T^2\check{\Delta}(n/T)^2}+ \sum_{n=0}^{sT-1} \frac{(32+8\xi_2) \hat{c}_1(n/T)^2}{ T^2 \check{\Delta}(n/T)^2}  \nonumber \\
       & \quad + \sum_{n=1}^{sT-1}\left(3\pi\xi_2\xi_3+(4\pi+16)\xi_1\xi_2\right) \frac{\hat{c}_1(n/T)^2}{T^2\check{\Delta}(n/T)^3}  + \sum_{n=1}^{sT-1} \left(3\xi_2\xi_3+4\xi_1\xi_2\right)\frac{\hat{c}_2(n/T)}{T^2\check{\Delta}(n/T)^2}  \nonumber \\
       & = \frac{8\xi_2\hat{c}_1(0)}{T\check{\Delta}(0)^2}+  \frac{8\xi_2\hat{c}_1(s)}{T\check{\Delta}(s)^2} + \frac{4\xi_2 \hat{c}_1(s)}{T\check{\Delta}(s)} \nonumber \\
       & \quad + \left(48\xi_2^2 + 6\pi \xi_3 + (8\pi+32)\xi_1 + 8\pi \xi_2 + 3\pi\xi_2\xi_3+(4\pi+16)\xi_1\xi_2\right)\sum_{n=1}^{sT-1}\frac{\hat{c}_1(n/T)^2}{T^2\check{\Delta}(n/T)^3} \nonumber \\
       & \quad + \left(32+8\xi_2\right)\sum_{n=0}^{sT-1} \frac{ \hat{c}_1(n/T)^2}{ T^2 \check{\Delta}(n/T)^2} +  (6\xi_3+8\xi_1+3\xi_2\xi_3+4\xi_1\xi_2)\sum_{n=1}^{sT-1}\frac{\hat{c}_2(n/T)}{T^2\check{\Delta}(n/T)^2} \nonumber \\
       & \leq \frac{12\hat{c}_1(0)}{T\check{\Delta}(0)^2}+  \frac{12\hat{c}_1(s)}{T\check{\Delta}(s)^2} + \frac{6 \hat{c}_1(s)}{T\check{\Delta}(s)} \nonumber \\
       & \quad + 305\sum_{n=1}^{sT-1}\frac{\hat{c}_1(n/T)^2}{T^2\check{\Delta}(n/T)^3}  + 44\sum_{n=0}^{sT-1} \frac{ \hat{c}_1(n/T)^2}{ T^2 \check{\Delta}(n/T)^2} +  32\sum_{n=1}^{sT-1}\frac{\hat{c}_2(n/T)}{T^2\check{\Delta}(n/T)^2},
\end{align}
where the last inequality is derived by plugging the concrete values of $\xi_k$ into the bound and rounding the resulting constants to the closest integers greater than or equal to them.

\end{proof}

\section{Application: solving linear systems}
\label{sec:linsys}

\subsection{Preparing the walker}

In this section we apply \cref{theoAdia}, about adiabatic evolution in the discrete setting, to solve the quantum linear system problem. 
In adiabatic quantum computation, one usually uses a Hamiltonian that is a combination of two Hamiltonians as
\begin{equation}
\label{eq:Ham_prob}
    H(s) = (1-f(s))H_0 + f(s)H_1,
\end{equation}
where the function $f(s): [0,1] \rightarrow [0,1]$ is called the schedule function.
Normally $H_0$ is the Hamiltonian where the ground state is easy to prepare, and $H_1$ is the one where the ground state encodes the solution of the problem that we are trying to determine.
For the case of linear systems solvers, the ground state of $H(1)$ should encode the normalized solution for a linear system. In other words, for $A\in \mathbb{C}^{N\times N}$ an invertible matrix with $\|A\|=1$, and a normalized vector $\ket{b}\in \mathbb{C}^N$ the goal is to prepare a normalized state $\ket{\tilde{x}}$, which is an approximation of $\ket{x}=A^{-1}\ket{b}/\|A^{-1}\ket{b}\|$.
For precision $\epsilon$ of the approximation, we require $\|\ket{\tilde{x}} - \ket{x}\|\leq \epsilon$.
One can also bound the error in terms of $\|\ket{\tilde{x}}\bra{\tilde{x}} - \ket{x}\bra{x}\|$ as was done in some prior work \cite{PhysRevLett.122.060504,an2019quantum}, which is asymptotically equal (for small error).
Translating this problem to our theorem for the adiabatic evolution, $\ket{\tilde{x}}$ would be the state achieved from the steps of the walk, and $\ket{x}$ would be obtained from the ideal adiabatic evolution.

Beginning with the simplest case, where $A$ is Hermitian and positive definite, one takes the Hamiltonians  \cite{an2019quantum}
\begin{equation}
\label{eq:H0}
    H_0\coloneqq\begin{pmatrix} 0 &  Q_b\\
 Q_b & 0
\end{pmatrix},
\end{equation}
and
\begin{equation}
\label{eq:H1}
    H_1\coloneqq \begin{pmatrix} 0 &  AQ_b\\
 Q_bA & 0
\end{pmatrix},
\end{equation}
where $Q_b=I_N-\ket{b}\bra{b}$. 
The state $\ket{0,b}$ is an eigenstate of $H_0$ with eigenvalue 0, and one would aim for this to evolve adiabatically to eigenstate $\ket{0,A^{-1}b}$ of $H_1$.
There is also eigenstate $\ket{1,b}$ or both $H_0$ and $H_1$ with the same eigenvalue 0, but it is orthogonal and we will show that there is no crossover in the ideal adiabatic evolution using the walk.

Denoting the condition number of the matrix as $\kappa$, a lower bound for the gap of $H(s)$ is \cite{an2019quantum}
\begin{equation}
\label{eq:gapSc}
    \Delta_0(s)= 1- f(s) + f(s)/\kappa.
\end{equation}
Note that according to \cref{def:gaps}, $\Delta_0(s)$ is a lower bound on the exact gap between the eigenvalues, so we use an equality here rather than an inequality.

Since the goal is to get a schedule function which slows down the evolution as the gap becomes small, a standard condition for the schedule is \cite{jansen2007bounds}
\begin{equation}
\label{eq:gapCon}
    \dot{f}(s) = d_p\Delta_0^p(s),
\end{equation}
 where $f(0)=0$, $p>0$ and $d_p=\int_0^1\Delta_0^{-p}(u)\, du$ is a normalization constant chosen so that $f(1)=1$. It is possible to show that \cite{an2019quantum}
\begin{equation}
\label{eq:sched1}
f(s) = \frac{\kappa}{\kappa - 1}\left[1-\left(1+s\left(\kappa^{p-1}-1\right)\right)^{\frac{1}{1-p}}\right], 
\end{equation}
satisfies \cref{eq:gapCon}, but with $\Delta_0(s)$ replaced with the lower bound on the gap from \cref{eq:gapSc}. This schedule function has two properties that have useful applications to estimate the upper bounds for the difference between consecutive walker operators, namely that $f(s)$ is monotonic increasing and that $\dot{f}(s)$ is monotonic decreasing. 

Distinct from the continuous version of the adiabatic theorem, in our discrete version of the theorem, we have to take into account the gap between the different groups of eigenvalues of $W_T(s)$ for $s,s+1/T$ and $s+2/T$, as described in \cref{eqn:assump2}.
From the property that the gap function is monotonically increasing we have
\begin{equation}
\label{eq:Gaps}
    \Delta_k(s) = 1- f(s+k/T) + f(s+k/T)/\kappa, \quad k=0,1,2.
\end{equation}

In the case where $A$ is not positive definite but Hermitian, the Hamiltonians $H_0$ and $H_1$ can be modified.
In \cite{PhysRevLett.122.060504} the authors show that one can use a larger Hilbert space for both $H_0$ and $H_1$ as 
\begin{equation}
    H_0=\sigma_+ \otimes \left[(\sigma_z \otimes I_N)Q_{+,b}\right] + \sigma_- \otimes \left[Q_{+,b}(\sigma_z \otimes I_N)\right],
\end{equation}
where $Q_{+,b} = I_{2N}-\ket{+,b}\bra{+,b}$ and $\ket{\pm}=\frac{1}{\sqrt{2}}(\ket{0}\pm\ket{1})$, and 
\begin{equation}
\label{eq:H1_nonP}
    H_1=\sigma_+ \otimes \left[(\sigma_x \otimes A)Q_{+,b}\right] + \sigma_- \otimes \left[Q_{+,b}(\sigma_x \otimes A)\right].
\end{equation}
Here two ancilla qubits are needed to enlarge the matrix block. 
The solution of the linear system problem can be obtained if we can prepare the zero-energy state $\ket{0,+,b}$ of $H_0$. 
By replacing the Hamiltonians in \cref{eq:Ham_prob} for solving the QLSP it is possible to show that the spectral gap of $H(s)$ is lower bounded as $\Delta_0(s) \geq \sqrt{(1-f(s))^2 + (f(s)/\kappa)^2}$ \cite{PhysRevLett.122.060504}.
To avoid the need to use this formula, you can use the relation that for $0\leq f(s) \leq 1$, 
\begin{equation}
\sqrt{(1-f(s))^2 + (f(s)/\kappa)^2} \geq (1-f(s) + f(s)/\kappa)/\sqrt{2}.
\end{equation}
When $A$ is not positive definite, you can keep using the same schedule function from \cref{eq:sched1}, but the spectral gaps can instead be lower bounded by
\begin{equation}
\label{eq:GapsA_non}
    \Delta'_k(s) =  \left(1- f(s+k/T) + f(s+k/T)/\kappa\right)/\sqrt{2}, \quad k=0,1,2.
\end{equation}

The standard approach when $A$ is non-Hermitian is to construct a Hermitian matrix as
\begin{equation}\label{eq:Avec}
    \mathbf{A}\coloneqq\begin{pmatrix} 0 &  A\\
 A^{\dagger} & 0
\end{pmatrix},
\end{equation}
and use
\begin{equation}\label{eq:bvec}
    \mathbf{b}\coloneqq\begin{pmatrix} 0\\
 b
\end{pmatrix}.
\end{equation}
In our case we adopt a slightly different approach, and instead of replacing $A$ in \cref{eq:H1}, we replace $\sigma_x \otimes A$ in \cref{eq:H1} by $\mathbf{A}$.
As we will show, this gives the same lower bound on the gap without further expanding the dimension.
Then you have the final Hamiltonian
\begin{equation}
    H_1=\sigma_+ \otimes \left[\mathbf{A} Q_{\mathbf{b}}\right] + \sigma_- \otimes \left[Q_{\mathbf{b}} \mathbf{A}\right].
\end{equation}
Now define
\begin{equation}\label{eq:Af}
    A(f) \coloneqq (1-f) \sigma_z \otimes I_N + f \mathbf{A} =
    \begin{pmatrix} (1-f)I &  f A\\
 f A^{\dagger} & -(1-f)I
\end{pmatrix}
\end{equation}
so
\begin{equation}\label{eq:Hsencoding}
    H(s) = (1-f(s)) H_0 + f(s) H_1 = \begin{pmatrix}
     0 & A(f(s)) Q_{\mathbf{b}} \\
     Q_{\mathbf{b}} A(f(s)) & 0
    \end{pmatrix}
\end{equation}
Then it is found that
\begin{equation}
    H^2(s) = \begin{pmatrix}
     A(f(s)) Q_{\mathbf{b}} A(f(s)) & 0 \\
     0 & Q_{\mathbf{b}} A^2(f(s)) Q_{\mathbf{b}}
    \end{pmatrix}
\end{equation}
As per the analysis in the Supplementary Material of \cite{PhysRevLett.122.060504}, the spectra of $A(f(s)) Q_{\mathbf{b}} A(f(s))$ and $Q_{\mathbf{b}} A^2(f(s)) Q_{\mathbf{b}}$ are identical.
Moreover, following that analysis, the gap of $A(f(s)) Q_{\mathbf{b}} A(f(s))$ is lower bounded by the minimum eigenvalue of $A^2(f(s))$.
In this case, since
\begin{equation}
    A^2(f)  =
    \begin{pmatrix} (1-f)^2I+f^2 A A^\dagger &  0 \\
 0  & (1-f)^2I+f^2 A^\dagger A
\end{pmatrix},
\end{equation}
the minimum eigenvalue is $(1-f)^2+(f/\kappa)^2$.
This translates to a minimum gap of $H(s)$ of $\sqrt{(1-f(s))^2+(f(s)/\kappa)^2}$ as in the Hermitian case.

By using the qubitised quantum walk for the implementation of $W_T$, we can avoid the logarithmic factor in the complexity that arises from using the Dyson series to simulate continuous Hamiltonian evolution.
In order to block encode the Hamiltonian $H(s)$, one can use block encodings of both $H_0$ and $H_1$, supplemented with an ancilla qubit that will be rotated to select between $H_0$ and $H_1$. The rotation is given by
\begin{equation}
\label{eq:C-rot}
R(s)=\frac{1}{\sqrt{\left(1-f(s)\right)^2+f(s)^2}}\begin{pmatrix} 1-f(s) &  f(s)\\
 f(s) & -(1-f(s))
\end{pmatrix}.
\end{equation}
To block encode $A(f(s))$, instead of using symmetric rotations, we use the initial rotation $R(s)$, then apply the controlled operations
\begin{equation}
     \textsc{sel} = \ket{0}\bra{0} \otimes U_0 + \ket{1}\bra{1} \otimes U_1,
\end{equation}
where $U_0$ and $U_1$ are unitaries used for the block encodings of $\sigma_z\otimes I_N$ and $\mathbf{A}$.
Then after this operation, instead of applying the inverse of $R(s)$, we simply perform a Hadamard.
This means that, instead of block encoding $A(f(s))$, we have block encoded
\begin{equation}
    \frac 1{\sqrt{2[(1-f(s))^2+f(s)^2]}} A(f(s)).
\end{equation}
This prefactor is between $1/\sqrt{2}$ and 1, and will reduce the gap.

For the qubitisation it is important that the block encoding is symmetric.
For the complete qubitisation of $H(s)$ as in \cref{eq:Hsencoding}, there are two blocks, one with $Q_\mathbf{b}$ followed by $A(f(s))$, and the other with the reverse order.
These blocks independently are asymmetric, but together they give a Hermitian Hamiltonian, and the block encoding is symmetric.
We can apply exactly the same procedure to account for the asymmetry between the $R(s)$ and Hadamard.
That is, for one block in $H(s)$, we can use $R(s)$ at the beginning and the Hadamard at the end, and the other we reverse the order.
The overall encoding is symmetric, as required for qubitisation.
We can similarly use this approach for combining $H_0$ and $H_1$ from \cref{eq:H0,eq:H0} for the case where $A$ is positive definite and Hermitian.
The qubitised operator $W_T(s)$ is then obtained by combining the block encoding of $H(s)$ with a reflection on the control qubits.
For a complete description of the procedure, see \cref{app:blockHs}.

\subsection{Choosing values for \texorpdfstring{$c_1(s)$}{c1(s)} and \texorpdfstring{$c_2(s)$}{c2(s)}}

Now, to apply \cref{cor:adia} for the QLSP two things that should be estimated are the functions $c_1(s)$ and $c_2(s)$, which in turn require upper bounds for $DW_T(s)$ and $D^{(2)}W_T(s)$.
In order to bound the difference in $W_T(s)$, we can use the fact that the only way $W_T(s)$ is dependent on $s$ is through $R(s)$.
The key feature of this operation is that it has $R(s)$ in two cross-diagonal blocks (in the matrix representation).
As a result, the spectral norm of the difference of operators is equal to the spectral norm of the difference of $R(s)$.

\begin{lemma}
\label{lem:DR} 
For any $0 \leq s \leq 1-1/T$, with $W_T(s)$ encoded using the block encoding of $H(s)$ from \cref{eq:Hsencoding} together with $R(s)$ given in \cref{eq:C-rot}, it is consistent with \cref{def:difs} to choose
    \begin{equation}
\label{eq:cs}
    c_1(s)=2T(f(s+1/T)-f(s)),
\end{equation}
and
\begin{equation}
\label{eq:c2}
    c_2(s) = \begin{cases}
       2\max_{\tau\in \{s,s+1/T,s+2/T\}}\left(2|f'(\tau)|^2 + |f''(\tau)|\right), & 0 \leq s \leq 1-2/T, \\
       2\max_{\tau\in \{s,s+1/T\}}\left(2|f'(\tau)|^2 + |f''(\tau)|\right), & s = 1-1/T.
    \end{cases}
\end{equation}
\end{lemma}
\begin{proof}
As discussed above, to bound the difference in $W_T(s)$, we can use the difference in the rotation operation $R(s)$, which can be upper bounded as
\begin{align}
\label{eq:DR}
    \left\| R(s+1/T)-R(s) \right\| &= \left\| \int_{s}^{s+1/T} \frac{dR}{ds} ds \right\| \nn
    &= \left\| \int_{s}^{s+1/T} \frac{dR}{df}\frac{df}{ds} ds \right\| \nn
    &\le  \int_{s}^{s+1/T} \left\|\frac{dR}{df}\frac{df}{ds} \right\| ds \nn
    &\le 2 \int_{s}^{s+1/T} \left|\frac{df}{ds} \right| ds \nonumber \\
    &=2 \left(f(s+1/T)-f(s)\right) .
\end{align}
Here we have upper bounded the norm of $dR/df$ by 2.
To show this, the derivative of $R(s)$ with respect to $f$
\begin{equation}
\frac{dR}{df}=\frac{1}{[\left(1-f(s)\right)^2+f(s)^2]^{3/2}}\begin{pmatrix} -f(s) & 1- f(s)\\
 1-f(s) & f(s)
\end{pmatrix},
\end{equation}
so the norm of $dR/df$ is $1/[(1-f(s))^2+f(s)^2]$, which varies between 1 and 2.
We have also used the fact that $df/ds>0$.
Since $\|W_T(s+1/T)-W_T(s)\|=\|R(1+1/T)-R(s)\|$ and $c_1(s)$ required in \cref{def:difs} to satisfy
\begin{equation}
    \|W_T(s+1/T)-W_T(s)\| \le \frac{c_1(s)}T,
\end{equation}
we can take $c_1(s)$ as in \cref{eq:cs}.

Now for the second difference of the walk operator we use Taylor's theorem in two directions to give
    \begin{equation}
        W_T(s+2/T) = W_T(s+1/T) + \frac{1}{T}\frac{dW_T(s+1/T)}{ds} + \int_{s+1/T}^{s+2/T} (s+2/T-\tau)\frac{d^2W_T(\tau)}{d\tau^2}d\tau, 
    \end{equation}
    and 
    \begin{equation}
        W_T(s) = W_T(s+1/T) - \frac{1}{T}\frac{dW_T(s+1/T)}{ds} + \int_{s}^{s+1/T} (\tau-s)\frac{d^2W_T(\tau)}{d\tau^2}d\tau .
    \end{equation}
That gives
\begin{align}
        \|D^{(2)}R(s)\| &= \left\|\int_{s+1/T}^{s+2/T} (s+2/T-\tau)R''(\tau)d\tau + \int_{s}^{s+1/T} (\tau-s)R''(\tau)d\tau\right\| \nonumber \\
        & \leq \frac{1}{T^2} \max_{\tau \in [s,s+2/T]}\|R''(\tau)\|.
    \end{align} 
    
    Moving to the second derivative of $R(s)$ we have
    \begin{equation}
        \frac{d^2R(s)}{ds^2} =\frac{d^2R}{df^2}\left(\frac{df(s)}{ds}\right)^2 + \frac{dR}{df}\frac{df^2(s)}{ds^2},  
    \end{equation}
    where
\begin{equation}
\frac{d^2R}{df^2}=\frac{1}{[\left(1-f(s)\right)^2+f(s)^2]^{5/2}}\begin{pmatrix} (4f(s)-1)f(s) -1 & (4f(s)-7)f(s)+2\\
 (4f(s)-7)f(s)+2 & (1-4f(s))f(s) +1
\end{pmatrix},
\end{equation}
and its norm is
\begin{equation}
    \sqrt{\frac{16(f(s)-1)f(s) +5}{\left[2(f(s)-1)f(s) +1\right]^4}},
\end{equation}
which varies between $\sqrt{5}$ and 4. Then we conclude that
\begin{equation}
     \left\|D^{(2)}R\left(s\right)\right\|\leq \frac{2}{T^2}\max_{\tau\in \{s,s+1/T,s+2/T\}}\left(2|f'(\tau)|^2 + |f''(\tau)|\right).
\end{equation}
Now we have $\|D^{2}W(s)\|=\|D^{2}R(s)\|$ and \cref{def:difs} requires that $\|D^{2}W(s)\|\le c_2(s)/T^2$, so we can take $c_2(s)$ as in \cref{eq:c2}.
\end{proof}

\subsection{Linear \texorpdfstring{$\kappa$}{kappa} for \texorpdfstring{$p=3/2$}{p=3/2}}

Our next step is to show the strict linear dependence in $\kappa$ for the QLSP based on our discrete adiabatic theorem.
In the continuous case, it has been shown in~\cite{an2019quantum} that for all $1<p<2$, the corresponding AQC-based linear system solver can achieve $\mathcal{\kappa/\epsilon}$ scaling.
This suggests that taking $p$ as the midpoint of $3/2$ will give high efficiency.
Here we consider this case to estimate the constant factors in the algorithm.

\begin{theorem}[Strict linear dependence in $\kappa$]\label{theo:p15} Consider solving the QLSP $Ax=b$ for a normalised state $\ket{A^{-1}b}$, where $\|A\|=1$ and $\|A^{-1}\|=\kappa$.
By using $T \geq \kappa$ steps of a quantum walk and the schedule function of \cref{eq:sched1} with $p=3/2$,
in the case of a positive-definite and Hermitian $A$ the error in the solution may be bounded as
    \begin{equation}
    \label{eq:analyticWalk}
        \|U_T(s) - U_T^A(s)\| \leq 5632\frac{\kappa}{T} + {\cal O}\left(\frac{\sqrt{\kappa}}{T}\right),
    \end{equation}
using the encoding of $H_0$ and $H_1$ as in \cref{eq:H0,eq:H1}.
For general $A$, by using the encoding of $H(s)$ as in \cref{eq:Hsencoding} the error may be bounded by
\begin{equation}
\label{eq:analyticWalk2}
\|U_T(s) - U_T^A(s)\| \leq
15307\frac{\kappa}{T} + {\cal O}\left(\frac{\sqrt{\kappa}}{T}\right).
\end{equation}
\end{theorem}
\begin{proof}
In \cref{cor:adia} there are six terms to bound, three which are individual terms and three which are sums.
The details of the derivations of bounds on these are given in \cref{ap:upperB_cor}. Namely for the individual terms, it is shown in \cref{appsec:single} (\cref{eq:firstupperB,eq:SecondupperB,eq:thirdupperB}) that
\begin{align}
\label{eq:single_c(0)}
\frac{\hat{c}_1(0)}{T\check{\Delta}(0)^2}
&=\frac{4\sqrt{\kappa}}{T} + \mathcal{O}\left( \frac{\kappa}{T^2} \right), \\
\label{eq:single_c(1)}
\frac{\hat{c}_1(1)}{T\check{\Delta}(1)^2} &
=\frac{4\kappa}{T} + \mathcal{O}\left( \frac{\kappa}{T^2} \right), \\
\label{eq:single_c(1)_2}
\frac{\hat{c}_1(1)}{T\check{\Delta}(1)} &
=\frac{4}{T} + \mathcal{O}\left( \frac{1}{T^2} \right) ,
\end{align}
and for the sums with $c_1(s)$, it is shown in \cref{appsec:summc1}  that
\begin{align}
\label{eq:sumc1}
\sum_{n=1}^{T-1}\frac{\hat{c}_1(n/T)^2}{T^2\check{\Delta}(n/T)^3}
&= \frac{16 \kappa}{T}+ \mathcal{O} \left( \frac{\kappa^{3/2}}{T^2} \right),\\
\label{eq:sumc1_2}
 \sum_{n=0}^{T-1} \frac{ \hat{c}_1(n/T)^2}{ T^2 \check{\Delta}(n/T)^2} &\leq \frac{16}{T} +\mathcal{O} \left( \frac {\kappa}{T^2} \right).
\end{align}
Finally for the sum with $c_2(s)$, it is shown in \cref{appsec:sumc2} that 
\begin{equation}
\label{eq:sumc2}
\sum_{n=1}^{T-1} \frac{ \hat{c}_2(n/T)}{ T^2 \check{\Delta}(n/T)^2} \leq 
\frac {22\kappa}{T} + \mathcal{O}\left( \frac{\kappa^{3/2}}{T^2}\right).
\end{equation}

These results are for the case where $A$ is positive definite and Hermitian.
By adding all the inequalities above, and including the constant factors in \cref{cor:adia}, we obtain the total upper bound in \cref{eq:analyticWalk}.
Note that we only include the leading term proportional to $\kappa/T$, and terms with scalings such as $\kappa^{3/2}/T^2$ are order $\sqrt{\kappa}/T$ due to the requirement that $T>\kappa$.
For the case of general $A$ (which need not be positive definite or Hermitian), the spectral gap can be lower bounded using an extra factor of $1/\sqrt 2$ as in \cref{eq:GapsA_non}.
This means that the upper bounds on the terms with $\check{\Delta}(s)$,  $\check{\Delta}(s)^2$ and  $\check{\Delta}(s)^3$ in the denominator may be multiplied by $\sqrt{2}$, 2 and $2\sqrt{2}$, respectively, to give new upper bounds that hold in the case of general $A$.
Adding these terms together with the weightings from \cref{cor:adia} then gives the upper bound in \cref{eq:analyticWalk2}.

There are two further subtleties in using the adiabatic algorithm for the solution.
One is that the zero eigenvalue of the Hamiltonian is degenerate, with one giving the solution, and the other just the state $\ket{b}$ but with a bit flip in an ancilla.
Because these eigenstates are orthogonal (due to the bit flip), there is no crossover between them in the adiabatic evolution.
This means that the degeneracy has no effect on the quality of the solution; see \cref{app:phasefactor}.

A further subtlety is that the qubitised quantum walk yields two eigenstates for each eigenstate of the Hamiltonian.
For the case here, the eigenvalue of the Hamiltonian we are interested in is $0$, which gives eigenvalues $\pm 1$ of the walk operator.
We may use the discrete adiabatic theorem separately on each of these eigenvalues to show that the eigenstate is preserved in the discrete adiabatic evolution.
The problem is that we need to have a positive superposition of these two eigenstates, which means that there should be no relative phase factor introduced in the adiabatic evolution.
It is shown that there is no relative phase factor in \cref{app:phasefactor}.
Therefore, neither of these subtleties has an effect on the solution, and no adjustment to the bounds is required.
\end{proof}

\subsection{General \texorpdfstring{$p$}{p}}

In this subsection, we will show that the $\mathcal{\kappa/\epsilon}$ scaling also holds for all $1<p<2$ in the discrete setting.
This result is more general, but due to a number of approximations will not be as tight an estimate as that for the specific case of $p=3/2$.
Since here we do not assume a specific value of $p$, the direct computation approach in proving \cref{theo:p15} is not applicable.
Instead, we will approximate the upper bound of the discrete error by some continuous integrals and then bound both the integrals and the approximation errors. 
More precisely, we first notice that in \cref{cor:adia}, the dominant terms are the last three terms, the summations over equidistant discrete time steps. 
These summations are exactly in the form of Riemann sum and approximate some integrals. 
Then the dominant part of the discrete adiabatic errors can be bounded by some integrals plus the difference between the integrals and corresponding Riemann sums. 
Similar to what has been shown in~\cite{an2019quantum}, the integrals exactly scale $\mathcal{O}(\kappa/T)$. 
The difference between the integrals and Riemann sums is indeed of higher order according to the error bound of the first order quadrature formula. 
Combining all these together, we can prove that the discrete adiabatic error for general $1<p<2$ also scales as $\mathcal{O}(\kappa/T)$, which further implies an $\mathcal{O}(\kappa/\epsilon)$ complexity of the discrete AQC-based algorithm to solve the linear system problem within $\epsilon$ error. 

We summarize the main result in the following theorem. 
A detailed proof is given in \cref{app:qlsp_general_p}. 
\begin{theorem}[Linear dependence on $\kappa$ for general $p$]\label{thm:qlsp_general_p}
    Consider using $T$ steps of the discrete adiabatic evolution with the schedule function defined in \cref{eq:sched1} for $1 < p < 2$ to solve the QLSP with general matrix $A$. 
    Then 
    \begin{enumerate}
        \item for any $\kappa > 2$ and $T \geq 32d_p/3 = \mathcal{O}(\kappa^{p-1})$, there exists a positive constant $C_p$, which only depends on $p$, such that the difference between the discrete adiabatic evolution and the solution of the linear system problem can be bounded by 
    \begin{equation}
         C_p \left(\frac{\kappa}{T}+\frac{\kappa^{p-1}}{T} + \frac{\kappa}{T^2} + \frac{1}{T}\right), 
    \end{equation}
    \item in order to prepare an $\epsilon$-approximation of the solution of the linear system problem, it suffices to choose 
    \begin{equation}
        T = \mathcal{O}\left(\frac{\kappa}{\epsilon}\right). 
    \end{equation}
    \end{enumerate}
\end{theorem}

We give an explicit formula for $C_p$ in \cref{eq:Cpdef} in \cref{app:qlsp_general_p}.
We remark that \cref{thm:qlsp_general_p} only guarantees the asymptotic performance of the discrete AQC-based solvers, and the pre-constant $C_p$ is much larger than what we observe numerically. 
In particular, \cref{thm:qlsp_general_p} also holds for the case when $p=3/2$, but the pre-constant in \cref{thm:qlsp_general_p} is much larger than that in \cref{theo:p15}. 
This is because in \cref{thm:qlsp_general_p} we use a general proof strategy, which is applicable for all $1<p<2$ at a sacrifice of using potentially unnecessary inequalities to simplify the analysis. 
These inequalities are definitely not sharp and thus result in a worse pre-constant than that obtained by direct computations in the proof of \cref{app:qlsp_general_p}. 

\subsection{Numerical Results}

\begin{figure}[b]
	\includegraphics[trim= -160 0 0 0 80,clip,scale=0.6]{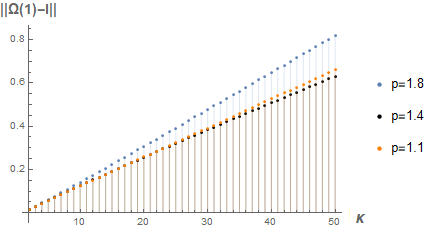}\\
    \caption{This figure shows the upper bound on the error in the adiabatic evolution versus the condition number $\kappa$ for a range of values of $p$ used in the scheduling function $f(s)$.
    The upper bound on the error is computed using \cref{theoAdia}, and
    in all cases the number of steps of the walk is $T=5\times 10^4$.
		\label{fig:blockerror}
	}
\end{figure}

\begin{figure}[tb]
	\includegraphics[trim= -160 0 0 0,clip,scale=0.60]{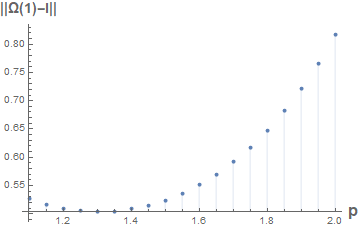}\\
    \caption{The upper bound on the error in the adiabatic evolution as a function of $p$ used in the scheduling function $f(s)$. In this plot we have used constant values $\kappa=40$ and $T=5\times10^4$.
		\label{fig:errBlockp}
 	}
\end{figure}

We first report the numerical results for the case where $A$ is a Hermitian and positive matrix. Rather than using the upper bounds for the first and second differences of the walk operator from \cref{lem:DR}, we exactly compute the norm of $DR$ and $D^{(2)}R$ in order to give values of $c_1(s)$ and $c_2(s)$ in \cref{theoAdia}.
We also account for the fact that the gap actual gap for the quantum walk operator is the $\arcsin$ of that in \cref{eq:Gaps}. 

In~\cref{fig:blockerror} we show the numerical results for the (upper bound on the) error as a function of the condition number $\kappa$ of the matrix $A$.
We use a fixed number of steps $T=5\times 10^4$ and three different values of $p$ for the schedule function of \cref{eq:sched1}.
In each case it can be seen that the error is approximately linear in $\kappa$, which is what results in an overall complexity which is linear in $\kappa$.
The different values of $p$ result in different scaling constants with values close to 1 or 2 giving poorer scaling, which is as expected since we require $1<p<2$.

To more clearly see the dependence of the error in $p$,
in \cref{fig:errBlockp} we show the error as a function of $p$ for constant $\kappa$ (of 40).
In this case it turns out that the smallest error is for $p=1.3$, which is on the lower side of the range $(1,2)$, and smaller than the value $p=3/2$ chosen for \cref{thm:qlsp}.
From \cref{fig:errBlockp} we can also estimate the constant factors for the $\kappa/T$ scaling of the error.
In the case with $p=3/2$, for instance, as was used in \cref{theo:p15}, we have $\|U_T(s)-U_T^A(s)\|\lesssim 638 \kappa/T$. The estimate of the constant factor in \cref{theo:p15} is around 9 times bigger.
This is not unreasonable considering the many approximations made, though it indicates that the constant factor in the analysis can be improved by a more careful analysis.

\section{Filtering for solving linear equations}
\label{sec:filter}

To provide a solution to linear equations using the adiabatic method, one can use the approach of \cite{Lin2020optimalpolynomial} where the initial adiabatic algorithm is used to find the solution to some constant error (independent of $\epsilon$), then the solution can be filtered.
The approach used in \cite{Lin2020optimalpolynomial} was to apply filtering by singular value processing (similar to quantum signal processing), which is efficient and only needs one ancilla qubit, but has the drawback that it requires a highly complicated procedure for finding the correct rotation angles.
Here we provide a method using a linear combination of unitaries with similar efficiency, and only requiring two ancilla qubits (one more than singular value processing).
This has the advantage that determining the sequence of gates needed is much simpler.

The filtering by a linear combination of unitaries is similar in principle to measuring the eigenvalue of the Hamiltonian to ensure that the system is still in the ground state.
However, that does not achieve quite what we want, because it will typically produce an estimate different from the required eigenvalue.
Using a linear combination of unitaries, it may be chosen such that the ``success'' case is obtained with high probability, and one then need only consider the amplitude for incorrect states in the final state.

This is related to the principle of using symmetric states in a linear combination of unitaries.
A phase measurement would be equivalent to using preparation with the desired amplitudes at the beginning, then inverse preparation on an equal superposition at the end.
In contrast, the approach that maximises the success probability for a linear combination of unitaries is to use symmetric preparation before and after the controlled operations.
Calling the desired weights $w_j$, we would initially prepare the control register in the state
\begin{equation}\label{eq:symmsta}
    \frac 1{\sqrt{\sum_j w_j}} \sum_j \sqrt{w_j} \ket{j} .
\end{equation}
Given that we are performing $j$ steps of the walk, and the input system state is an eigenvector of the walk with eigenvalue $e^{i\phi}$, the resulting state is
\begin{equation}
    \frac 1{\sqrt{\sum_j w_j}} \sum_j \sqrt{w_j} e^{ij\phi} .
\end{equation}
Then projecting on the same state as in \cref{eq:symmsta} gives
\begin{equation}
    \frac 1{\sum_j w_j} \sum_j w_j e^{ij\phi} .
\end{equation}

In practice, the target register will be a superposition of the eigenstates
\begin{equation}
    \sum_k \psi_k \ket{k},
\end{equation}
where we are using $\ket{k}$ to indicate the eigenstate of $W_T(1)$ corresponding to eigenvalue $\phi_k$.
The state after applying the linear combination of unitaries is then
\begin{equation}
    \frac 1{\sum_j w_j} \sum_{j,k} w_j \psi_k e^{ij\phi_k} \ket{k} =
    \sum_{k} \tilde w(\phi_k) \psi_k \ket{k},
\end{equation}
where
\begin{equation}\label{eq:tildew}
\tilde w(\phi)=    \frac 1{\sum_j w_j} \sum_{j} w_j e^{ij\phi}.
\end{equation}
Note that the state is not normalised, with the norm giving the probability of the success of this linear combination of unitaries.
We aim to have $\tilde w(\phi)$ for $\phi$ in the spectrum of interest.
Now, let us assume that the initial probability of the state on the spectrum of interest is at least $1/2$.
One can then show that the resulting normalised state obtained after the filtering has error, as quantified by the norm of the difference of states, upper bounded by
\begin{equation}
    \max_{k\in\{\perp\}} \tilde w(\phi_k).
\end{equation}
where $\perp$ is the set of $k$ such that $\phi_k$ is not in the spectrum of interest (so $e^{i\phi_k}\not\in \sigma_P$).
See \cref{ap:filtering} for the proof.

The result of this reasoning is that to bound the error in the filtering, we need to bound the maximum of $\tilde w(\phi_k)$, which is minimised by the Dolph-Chebyshev window.
This is obtained by taking the discrete Fourier transform of the Chebyshev polynomials, so that $\tilde w(\phi)$ is given by Chebyshev polynomials in a similar way as for \cite{Lin2020optimalpolynomial}.
In particular, one can take
\begin{equation}
    \tilde w(\phi) = \epsilon T_{\ell}\left( \beta \cos\left( \phi \right) \right)
\end{equation}
for $\phi$ taking discrete values $\pi k/\ell$ for $k$ from $-\ell$ to $\ell$, and
where $\beta=\cosh(\tfrac 1\ell \cosh^{-1}(1/\epsilon))$.
Taking the discrete Fourier transform of these values gives the window, and the Fourier transform simply yields the formula for $\tilde w(\phi)$ in terms of Chebyshev polynomials.
One obtains powers of $e^{2i\phi}$ from $-\ell/2$ to $+\ell/2$, which means we need a maximum power of $e^{i\phi}$ of $\ell$.
One can obtain the positive and negative powers simultaneously with negligible cost by simply controlling whether the reflection is performed in the qubitisation.
As a result, the cost in terms of calls to the block-encoded matrix is $\ell$, as compared to $2\ell$ for the singular value processing approach.

The peak for $\tilde w(\phi)$ will be at $0$ and $\pi$, which is what is needed because the qubitised operator produces duplicate eigenvalues at $0$ and $\pi$.
The width of the operator can be found by noting that the peak is for the argument of the Chebyshev polynomial equal to $\beta$, and the width is where the argument is 1, so $\beta\cos(\phi)=1$.
This gives us
\begin{equation}\label{eq:chebwid}
    \cosh(\tfrac 1\ell \cosh^{-1}(1/\epsilon)) \cos(\phi) = 1 .
\end{equation}
Now, because the width of the peak should be equal to the gap, and the gap is $1/\kappa$, we can replace $\phi$ with $1/\kappa$, and solving for $\ell$ gives
\begin{equation}
    \ell = \frac{\cosh^{-1}(1/\epsilon)}{\cosh^{-1}(1/\cos(1/\kappa))} \le \kappa \ln(2/\epsilon).
\end{equation}
Note that \cref{eq:chebwid} was for finding the width given an integer $\ell$, but solving for $\ell$ with a width of $1/\kappa$, we should round $\ell$ up to the nearest integer to provide a width no larger than $\kappa$.

In comparison, in \cite{Lin2020optimalpolynomial} the error is given as $2 e^{-\sqrt{2} \ell \Delta}$, which would imply that one can take $\ell \approx \sqrt{1/2} \kappa \ln (2/\epsilon)$.
Since the order of the polynomial is $2\ell$, which is also the number of applications of the block encoding needed, this would imply a cost of $\sqrt{2} \kappa \ln (2/\epsilon)$, which is greater than what we have here by a factor of $\sqrt{2}$.
However, it turns out that the scaling given in \cite{Lin2020optimalpolynomial} is overly conservative, and the actual scaling is $2 e^{-2 \ell \Delta}$, which then gives the same complexity as we have here.

Next we consider how to apply the linear combination of unitaries with minimum ancilla qubits.
To do this we first represent the control registers in unary.
That is, for each of the $\ell$ controlled operations, we use a single qubit which is one or zero depending on whether this operation is to be performed or not.
It may seem counterproductive to expand the size of the ancilla in this way, but it has the advantage that it has a simple state preparation procedure, where an initial qubit is rotated, then the following qubits are prepared by controlled rotations.
When doing this procedure, we can apply a just-in-time preparation procedure, where each qubit is prepared just as it is needed to be used as a control.
An example of this is shown in Fig.~\ref{fig:lcu1}.

\begin{figure}
\centerline{
\Qcircuit @R=1em @C=1em {
\lstick{\ket{0}} & \gate{R_0} & \ctrl{1} & \qw & \qw & \ctrl{4} & \qw & \qw & \qw & \qw & \qw & \ctrl{1} & \gate{R_0^\dagger} & \qw & \rstick{\!\!\!\!\!\bra{0}} \\
\lstick{\ket{0}} & \qw & \gate{R_1} & \ctrl{1} & \qw & \qw & \ctrl{3} & \qw & \qw & \qw & \ctrl{1} & \gate{R_1^\dagger} & \qw & \qw & \rstick{\!\!\!\!\!\bra{0}} \\
\lstick{\ket{0}} & \qw & \qw & \gate{R_2} & \ctrl{1} & \qw & \qw & \ctrl{2} & \qw & \ctrl{1} & \gate{R_2^\dagger} & \qw & \qw & \qw & \rstick{\!\!\!\!\!\bra{0}} \\
\lstick{\ket{0}} & \qw & \qw & \qw & \gate{R_3} & \qw & \qw & \qw & \ctrl{1} & \gate{R_3^\dagger} & \qw & \qw & \qw & \qw & \rstick{\!\!\!\!\!\bra{0}} \\
\lstick{\ket{\psi}} & \qw & \qw & \qw & \qw & \gate{W} & \gate{W} & \gate{W} & \gate{W} & \qw & \qw & \qw & \qw & \qw \\
}}
\caption{A linear combination of steps using control registers prepared in unary using a linear sequence of controlled rotations.
\label{fig:lcu1}}
\end{figure}
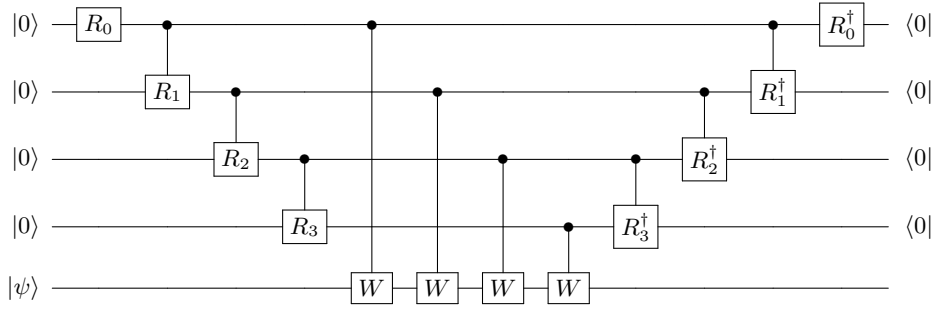

\begin{figure}
\centerline{
\Qcircuit @R=1em @C=1em {
\lstick{\ket{0}} & \gate{R_0} & \ctrl{1} & \qw & \qw & \ctrl{4} & \qw & \qw & \qw & \gate{R_0'} & \qw & \qw & \qw & \qw & \rstick{\!\!\!\!\!\bra{1}} \\
\lstick{\ket{0}} & \qw & \gate{R_1} & \ctrl{1} & \qw & \qw & \ctrl{3} & \qw & \qw & \ctrlo{-1} & \gate{R_1'} & \qw & \qw & \qw & \rstick{\!\!\!\!\!\bra{1}} \\
\lstick{\ket{0}} & \qw & \qw & \gate{R_2} & \ctrl{1} & \qw & \qw & \ctrl{2} & \qw & \qw & \ctrlo{-1} & \gate{R_2'} & \qw & \qw & \rstick{\!\!\!\!\!\bra{1}} \\
\lstick{\ket{0}} & \qw & \qw & \qw & \gate{R_3} & \qw & \qw & \qw & \ctrl{1} & \qw & \qw & \ctrlo{-1} & \gate{R_3'} & \qw & \rstick{\!\!\!\!\!\bra{1}} \\
\lstick{\ket{\psi}} & \qw & \qw & \qw & \qw & \gate{W} & \gate{W} & \gate{W} & \gate{W} & \qw & \qw & \qw & \qw & \qw \\
}}
\caption{A linear combination of steps using control registers prepared in unary using a linear sequence of controlled rotations, but with the inverse preparation performed with the linear sequence in the reverse order.
\label{fig:lcu2}}
\end{figure}
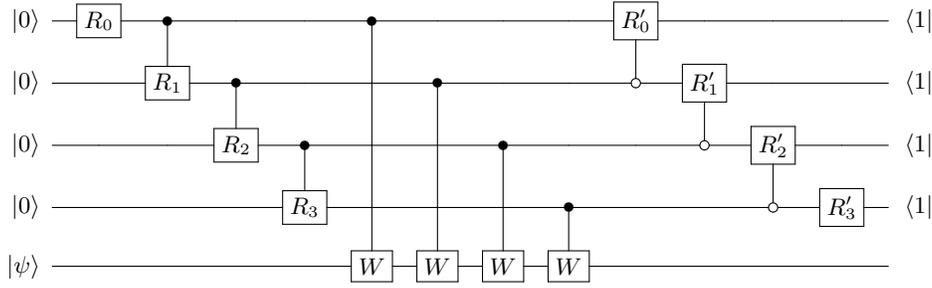

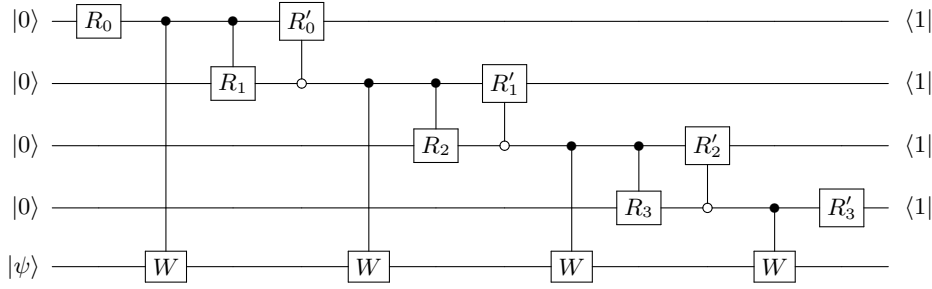
\begin{figure}
\centerline{
\Qcircuit @R=1em @C=1em {
\lstick{\ket{0}} & \gate{R_0} & \ctrl{4} & \ctrl{1} & \gate{R_0'} & \qw & \qw & \qw & \qw & \qw & \qw & \qw & \qw & \qw & \rstick{\!\!\!\!\!\bra{1}} \\
\lstick{\ket{0}} & \qw & \qw & \gate{R_1} & \ctrlo{-1} & \ctrl{3} & \ctrl{1} & \gate{R_1'} & \qw & \qw & \qw & \qw & \qw & \qw & \rstick{\!\!\!\!\!\bra{1}} \\
\lstick{\ket{0}} & \qw & \qw & \qw & \qw & \qw & \gate{R_2} & \ctrlo{-1} & \ctrl{2} & \ctrl{1} & \gate{R_2'} & \qw & \qw & \qw & \rstick{\!\!\!\!\!\bra{1}} \\
\lstick{\ket{0}} & \qw & \qw & \qw & \qw & \qw & \qw & \qw & \qw & \gate{R_3} & \ctrlo{-1} & \ctrl{1} & \gate{R_3'} & \qw & \rstick{\!\!\!\!\!\bra{1}} \\
\lstick{\ket{\psi}} & \qw & \gate{W} & \qw & \qw & \gate{W} & \qw & \qw & \gate{W} & \qw & \qw & \gate{W} & \qw & \qw \\
}}
\caption{A linear combination of steps using control registers prepared in unary, but with the order of the operations changed so we only need to use two ancilla qubits at a time.
\label{fig:lcu3}}
\end{figure}

Then in inverting the preparation, one could simply perform the reverse of all the controlled rotations as in Fig.~\ref{fig:lcu1}.
However, the trick is that the sequential state preparation procedure for the unary can be performed from either end.
The preparation could be achieved by performing rotations starting from the lart qubit, and working back to the first.
We do not do that for the preparation, but we do the reverse of that for the inverse preparation.
An example of this is shown in Fig.~\ref{fig:lcu2}.
When you reverse that form of preparation, you are working form the first qubit to the last, the same as for the preparation.
That means you only need to use two ancillas at once, by rearranging the operations as shown in Fig.~\ref{fig:lcu3}.
See \cref{ap:filtering} for a more explicit description of the sequence of rotations.

The major advantage of this procedure over singular value processing or quantum signal processing is that there is a very simple prescription for finding the sequence of operations.
A second advantage is that, instead of the measurement being performed at the end, measurements are performed sequentially, and a failure (the incorrect measurement result) can be flagged early.
That means that in cases where there will be a failure, it will on average be flagged halfway through, with the result that half the number of operations are needed since one can discard the state and start again.

Combining our result for the solution of QLSP with the filtering, we find that the overall complexity of the QLSP algorithm can be given as $\mathcal{O}(\kappa\log(1/\epsilon))$.
In particular, the result is as follows.
\begin{theorem}[QLSP with linear dependence on $\kappa$]\label{thm:qlsp}
Let $Ax=b$ be a system of linear equations, where $A$ is an $N$-by-$N$ matrix with $\norm{A}=1$ and $\norm{A^{-1}}=\kappa$. Given an oracle block encoding the operator $A$ and an oracle preparing $\ket{b}$, there exists a quantum algorithm which produces the normalized state $\ket{A^{-1}b}$ to within error $\epsilon$ using a number
\begin{equation}
    \mathcal{O}(\kappa\log(1/\epsilon))
\end{equation}
of oracle calls.
\end{theorem}
\begin{proof}
In this theorem, we use standard assumptions that access to the oracles includes forward, reverse, and controlled uses.
We initially apply the oracle for preparing $\ket{b}$ to prepare the initial state of the form in \cref{eq:bvec}.
This preparation is also used to construct the projection operator $Q_\mathbf{b}$.
Together with the oracle for block encoding $A$, we can construct the operator for block encoding $H(s)$ as described in detail in \cref{app:blockHs}.
A reflection on the ancillas yields the walk operator.

Now use the discrete adiabatic theorem for the QLSP as given in \thm{qlsp_general_p} for fixed precision, such as $1/2$.
That step has complexity $\mathcal{O}(\kappa)$, and the only error is the overlap with other states that are not the solution.
Next, use the filtering as described above, which has complexity $\mathcal{O}(\kappa\log(1/\epsilon))$.
In the case of success, one has produced the state $\ket{A^{-1}b}$ to within norm-distance $\epsilon$.
In the case of failure of the filtering, repeat the procedure.
Since the probability of success may be made at least $1/2$ by suitably choosing the fixed precision for the adiabatic procedure, the adiabatic and filtering steps need only be applied 2 times on average before success.
This gives a factor of $2$ to the total complexity of $\mathcal{O}(\kappa)$ plus $\mathcal{O}(\kappa\log(1/\epsilon))$.
The total complexity is therefore $\mathcal{O}(\kappa\log(1/\epsilon))$ as claimed.
\end{proof}

Perhaps surprisingly, in this complexity the largest asymptotic complexity is for the filtering step, because it has a factor of $\log(1/\epsilon)$ which is absent from the adiabatic step.
In practice, we have found quite large constant factors for the adiabatic evolution, so it is likely that the adiabatic step will still be the most costly part of the algorithm for realistic values of the parameters.
In particular, for the numerical calculation of the upper bound it was found that the scaling constant was about 638, so to obtain our requirement of initial probability on the spectrum of interest at least $1/2$ (corresponding to needing to repeat the algorithm twice on average) one would need about $834\kappa$ steps of the adiabatic evolution.
In contrast, the $\ln(2/\epsilon)$ factor is only about 20 for $\epsilon$ as one part in a billion.

\section{Conclusions}

In this work we have shown the first QLSP algorithm that scales optimally in terms of the condition number. We achieved this by adapting prior algorithms for the QLSP based on adiabatic evolutions so that they did not require the additional overhead of the Dyson series algorithm for precisely evolving under time-dependent Hamiltonians on a gate model quantum computer. Instead, we show that one can directly discretize the time-evolution using quantum walks and that the error in this procedure can be obtained using a discrete adiabatic theorem. We also obtain rigorous new error bounds on the performance of those discrete adiabatic theorems.

While this improvement is ``only'' by a log factor, the fact that we can asymptotically match the lower bound is of fundamental interest. Furthermore, there is widespread anticipation that compelling practical application of the QLSP may eventually be found and that error-corrected quantum computers capable of realizing those applications may eventually be realized. Should this occur, then it will be crucial to program those devices using the best possible scaling versions of these algorithms in order to have the fastest implementations requiring the least overhead due to error-correction. Our expectation is that the QLSP approach described in this paper would be more performant than any other approach in the literature both in terms of asymptotic scaling but also in terms of the constant factors associated with realizing finite instances. Thus, we also foresee practical value in these results.

As well as scaling optimally in the condition number, our algorithm scales optimally in terms of the combination of the condition number and the precision $\epsilon$.
As was recently proven, a lower bound to the complexity is $\mathcal{O}(\kappa \log(1/\epsilon))$ \cite{RobinAram}.
Our result matches this lower bound, showing that it is optimal.
It is interesting that the complexity is multiplicative between $\kappa$ and $\log(1/\epsilon)$, in contrast to Hamiltonian simulation which is additive between the time and $\log(1/\epsilon)$.
In this approach to solving linear equations, the $\log(1/\epsilon)$ factor only comes from the filtering step, which in practice would have lower complexity than the initial adiabatic step.

Another question is the scaling with the sparsity in the case where the matrix is sparse and given by oracles for positions of nonzero entries.
In this work we have given the complexity in terms of calls to a block encoding of the matrix, rather than those more fundamental oracles.
The lower bound in terms of those oracles has a multiplicative factor of $\sqrt{d}$ in the sparsity $d$.
One could get such a scaling if there were a way of block encoding the matrix with complexity $\sqrt{d}$, but standard methods are linear.
It is shown in \cite{rootd} how to simulate a Hamiltonian with complexity $\sqrt{d}$ up to logarithmic factors using a nested interaction picture approach.
One could use that combined with the adiabatic approach to obtain this scaling with sparsity, but it would reintroduce logarithmic factors, so the complexity would no longer be strictly linear in $\kappa$.

More generally, we expect that other quantum algorithms based on continuous time-evolutions might benefit from using discrete time adiabatic algorithms. For example, there are quantum algorithms for optimization that use adiabatic evolution.
There was some analysis demonstrating that discrete adiabatic evolution could be used in \cite{Sanders2020}, but our analysis here is far tighter.
There was also recent work showing that digital adiabatic simulation based on Trotter-type formulas is robust against discretization \cite{Changhao2021}, whereas our approach does not introduce any discretization error since we directly invoke the discrete adiabatic theorem.
Our analysis here could be tightened further in terms of the constant factors.
There is over an order of magnitude difference between the numerical results and the analytically proven scaling constants.
A more careful accounting for the inequalities could tighten this difference, but we have not  done that in this work because our analysis is already very lengthy.

\section*{Acknowledgements}

The authors thank Lin Lin, Kianna Wan and Nathan Wiebe for helpful discussions. DWB worked on this project under a sponsored research agreement with Google Quantum AI.
DWB is also supported by Australian Research Council Discovery Projects DP190102633 and DP210101367. 
DA acknowledges the support by the Department of Defense through the Hartree Postdoctoral Fellowship at QuICS, and is also partially supported by the Department of Energy under Grant No. DE-SC0017867.
YRS is supported by Australian Research Council Grant DP200100950.

\appendix

\section{List of variables}

Here we give a list of variable names with links to their definitions.

\subsection{List of variables presented in \texorpdfstring{\cref{sec:summary}}{Section II}} 
\begin{itemize}
\setlength\itemsep{0em}
    \item $W_T(s)$ - The discrete walk operator.
    \item $n$ - An integer index used for the discrete walk operators, so $s=n/T$.
    \item $U_T(s)$ - The product of walk operators up to $s$.
    item $P_T(s)$ - The projector onto the spectrum of interest.
    \item $Q_T(s)$ - The projector onto the complementary spectrum.
    \item $U_T^A(s)$ - The ideal adibatic evolution, in contrast to $U_T(s)$ given by the actual walk operators.
    \item $W_T^A(s)$ - Ideal adiabatic walk operators that exactly preserve eigenstates.
    \item $s$ - A variable used to index adiabatic evolution, starting from 0 and ending at 1.
    \item $T$ - An integer corresponding to the number of discrete walk operators in discrete adiabatic evolution.
    \item $T^*$ - A lower bound used for $T$ for the definition in \cref{def:gaps}.
    \item $D$ - A difference operator, so for example $DW_T(s) = W_T(s+1/T)-W_T(s)$.
    \item $D^{(k)}$ - The iterated difference operator.
    \item $c_k(s)$ - A bound on the norm of $D^{(k)}W_T(s)$ as in \cref{def:difs}.
    \item $\hat c_k(s)$ - The maximum of $c_k(s)$ over neighbouring time steps, as in \cref{eq:chat}.
    \item $\sigma_P(s)$ - The spectrum of interest.
    \item $\sigma_Q(s)$ - The complementary spectrum.
    \item $\sigma_P^{(k)}$ - An arc including the spectrum $\sigma_P(s)$ at $k+1$ successive steps, as in \cref{eqn:assump2}.
    \item $\sigma_Q^{(k)}$ - Similar to $\sigma_P^{(k)}$, but for the complementary spectrum.
    \item $\Delta_k(s)$ - The gap between the spectra accounting for $k+1$ successive steps; see \cref{def:gaps}.
    \item $\Delta(s)$ - The gap accounting for up to 3 successive steps as defined in \cref{eq:minGaps}.
    \item $\check \Delta(s)$ - The maximum of $\Delta(s)$ accounting for neighbouring steps; see \cref{eq:fhat}.
\end{itemize}

\subsection{List of variables presented in \texorpdfstring{\cref{sec:adtheo}}{Section III}}

\begin{itemize} \setlength\itemsep{0em}
 \item $R_T(s,z)$ - The resolvent of $W_T(s)$; see \cref{eq:resolv}.
    \item $S_T(s,s')$ - The operator exactly mapping from the spectrum at step $s'$ to $s$; see \cref{eq:STdef}.
    We also use $S_T(s)=S_T(s+1/T,s)$.
    \item $V_T(s,s')$ - The unitary obtained from a polar decomposition of $S_T(s,s')$; see \cref{eq:V}.
    We also use $V_T(s)=V_T(s+1/T,s)$.
    \item $v_T(s,s')$ - The correction to obtain $V_T(s,s')$ from $S_T(s,s')$.
    We also use $v_T(s)=v_T(s+1/T,s)$.
    \item $\Omega_T(s)$ - The wave operator, accounting for the difference between the ideal and adiabatic walk; see \cref{eq:waveOp}.
    \item $\Theta_T(s)$ - The ripple operator, corresponding to a step of $\Omega_T(s)$; see \cref{eq:ripple}.
    \item $K_T(s)$ - The kernel function, see \cref{eq:K}.
    \item $X(s)$ - Given by $T(1-V_T^{\dagger}(s-1/T))$ and used in the proof of the adiabatic theorem.
    \item $\tilde X(s)$ - Obtained from a contour integral of $X(s)$ as in \cref{eq:Xtilde}.
    \item $A(s)$ - A variable used in the proof of the discrete adiabatic theorem; see \cref{eq:A}.
    \item $B(s)$ - Used in the proof of the discrete adiabatic theorem; see \cref{eq:B}.
    \item $Z(s)$ - Used in the proof of the discrete adiabatic theorem; see \cref{eq:Z}.
    \item $\Gamma_T(s)$ - A contour that encloses the spectrum of interest.
    \item $\Gamma_T(s,k)$ - A contour that encloses the spectrum of interest for $k+1$ successive steps of the walk.
    \item $\mathcal{F}_T(s)$ - A function of $DP_T(s)$ used for expressing $V_T(s)$; see \cref{eq:def_F}.
    \item $\mathcal{B}$ - The boundary term used in \cref{lem:sum_by_parts}.
    \item $\mathcal{S}$ - The sum used in \cref{lem:sum_by_parts}.
    \item $n_\pm$ - We use $n_+=n+1$ and $n_-=n-1$. We also use this notation for $l$.
    \item $P_0$ - The initial projector onto the spectrum of interest, $P_0=P_T(0)$.
    \item $Q_0$ - Similarly for the complementary spectrum $Q_0=Q_T(0)$.
    \item $\mathcal{D}_j(x)$ - The simple scalar functions $\mathcal{D}_1(x),\mathcal{D}_2(x),\mathcal{D}_3(x)$ are defined in \cref{eq:D_i}.
    \item $\xi_j$ - Constants used for upper bounds on $\mathcal{D}_j(x)$ as in \cref{eq:upperB_D}.
    \item $\mathcal{G}_{T,j}(s)$ - These functions for $j=1,2,3,4$ are defined in \cref{eq:G1,eq:G2,eq:G3,eq:G4}.
\end{itemize}

\subsection{ List of variables presented in \texorpdfstring{\cref{sec:linsys}}{Section IV}}

\begin{itemize} \setlength\itemsep{0em}
    \item $A$ - The matrix in the QLSP $Ax=b$.
    \item $b$ - The vector in the QLSP.
    \item $x$ - This is usually used as the solution vector in $Ax=b$, but in \cref{ap:proofDP} as a real variable of integration.
    \item $N$ - The dimension of the QLSP.
    \item $\kappa$ - The condition number of $A$.
    \item $\epsilon$ - The allowable error in the solution.
    \item $H_0$ - The initial Hamiltonian in adiabatic evolution.
    \item $H_1$ - The final Hamiltonian in adiabatic evolution.
    \item $\ket{b}$ - The state with amplitudes proportional to the entries of $b$.
    \item $Q_b$ - The projector eliminating $\ket{b}$, given as $I_N-\ketbra{b}{b}$.
    \item $f(s)$ - Used for the scheduling function, which we take as in \cref{eq:sched1}.
    \item $d_p$ - A constant used in constructing $f(s)$; see \cref{eq:gapCon}.
    \item $p$ - An adjustable parameter used in the scheduling function, taking values in the range $(1,2]$.
    \item $\mathbf{A}$ - A matrix constructed from $A$ so as to be Hermitian; see \cref{eq:Avec}.
    \item $\mathbf{b}$ - A vector comprised of $b$ and a zero vector; see \cref{eq:bvec}.
    \item $A(f)$ - The intermediate value of $A$ used in the adiabatic evolution; see \cref{eq:Af}.
    \item $H(s)$ - The Hamiltonian constructed from $A(f)$; see \cref{eq:Hsencoding}.
    \item $R(s)$ - A rotation used in block encoding $H(s)$; see \cref{eq:C-rot}.
   \end{itemize}

\subsection{ List of variables presented in \texorpdfstring{\cref{sec:filter}}{Section V}}

\begin{itemize} \setlength\itemsep{0em}
 \item $w_j$ - Weights used for the linear combination of unitaries for filtering.
    \item $\phi_k$ - Used to label eigenvalues of the walk operator, so the eigenvalue is $e^{i\phi_k}$.
    \item $\tilde w(\phi)$ - A Fourier transform of $w_j$ as in \cref{eq:tildew}.
    \item $\perp$ - A set of $k$ such that $\phi_k$ is not in the spectrum of interest.
    \item $T_\ell$ - The Chebyshev polynomial of the first kind.
    \item $\ell$ - The order of the Chebyshev polynomial.
\end{itemize}

\section{Proof of \texorpdfstring{\cref{lem:P}}{Lemma 6}}\label{ap:proofDP}

In order to bound $DP_T(s)$, we first rewrite $D R_T(s,z)$ as
\begin{align}
\label{eq:contProj}
    D R_T(s,z) &=  \left(W_T\left(s+\frac{1}{T}\right)-z I\right)^{-1}-\left(W_T\left(s\right)-z I\right)^{-1} \nonumber\\
  &= \left(W_T\left(s+\frac{1}{T}\right)-z I\right)^{-1}\left(W_T\left(s\right)-z I\right)\left(W_T\left(s\right)-z I\right)^{-1}\nonumber\\
  &\quad -\left(W_T\left(s+\frac{1}{T}\right)-z I\right)^{-1}\left(W_T\left(s+\frac{1}{T}\right)-z I\right)\left(W_T\left(s\right)-z I\right)^{-1}\nonumber\\
  &= \left(W_T\left(s+\frac{1}{T}\right)-z I\right)^{-1}\left(W_T\left(s\right)-W_T\left(s+\frac{1}{T}\right) \right)\left(W_T\left(s\right)-z I\right)^{-1}\nonumber\\
  &= - R_T\left(s+\frac{1}{T},z\right)DW_T\left(s\right) R_T\left(s,z\right). 
  \end{align}

Using this expression, we can then express $DP_T(s)$ in terms of a contour integral as
\begin{align}
\label{eq:contProj2}
    DP_T(s)&= \frac{1}{2 \pi i} \oint_{\Gamma_T(s,1)} \left[R_T\left(s+\frac{1}{T},z\right)-R_T(s,z)\right] dz \nonumber \\
  &= -\frac{1}{2 \pi i} \oint_{\Gamma_T(s,1)} R_T\left(s+\frac{1}{T},z\right)DW_T\left(s\right) R_T\left(s,z\right)dz. 
  \end{align}
  When we consider $P_T(s)$, the integrand drops off as $1/|z|$, so the contour must be at a finite distance as illustrated in \cref{fig:contour_1_infty}.
  For $DP_T(s)$, we can use the same contour for both $P_T(s)$ and $P_T(s+1/T)$.
  The principle now is that the integrand falls off as $1/|z|^2$, so the contribution from the arc will fall to zero for large radius.
  We denote the contour as $\Gamma_T(s,1,a)$, which is a sector of radius $(a+1)$ for some real number $a$, and we will take the limit $a \rightarrow \infty$.
Then we have 
\begin{align}
    \|DP_T(s)\|
    &= \frac{1}{2\pi}\left\|\oint_{\Gamma_T(s,1,a)} R_T\left(s+\frac{1}{T},z\right)DW_T\left(s\right) R_T\left(s,z\right)dz\right\| \nonumber\\
    &\leq \frac{1}{2\pi}\oint_{\Gamma_T(s,1,a)} \left\|R_T\left(s+\frac{1}{T},z\right)\right\|\|DW_T\left(s\right) \|\|R_T\left(s,z\right)\||dz| \nonumber\\
    &\leq \frac{c_1(s)}{T}\frac{1}{2\pi}\oint_{\Gamma_T(s,1,a)} \left\|R_T\left(s+\frac{1}{T},z\right)\right\|\| R_T\left(s,z\right)\||dz|,
\end{align}
where in the last line we used the bound from \cref{eq:main_ass}. 

  \begin{figure}
    \centering
    \includegraphics[width = 0.6\textwidth]{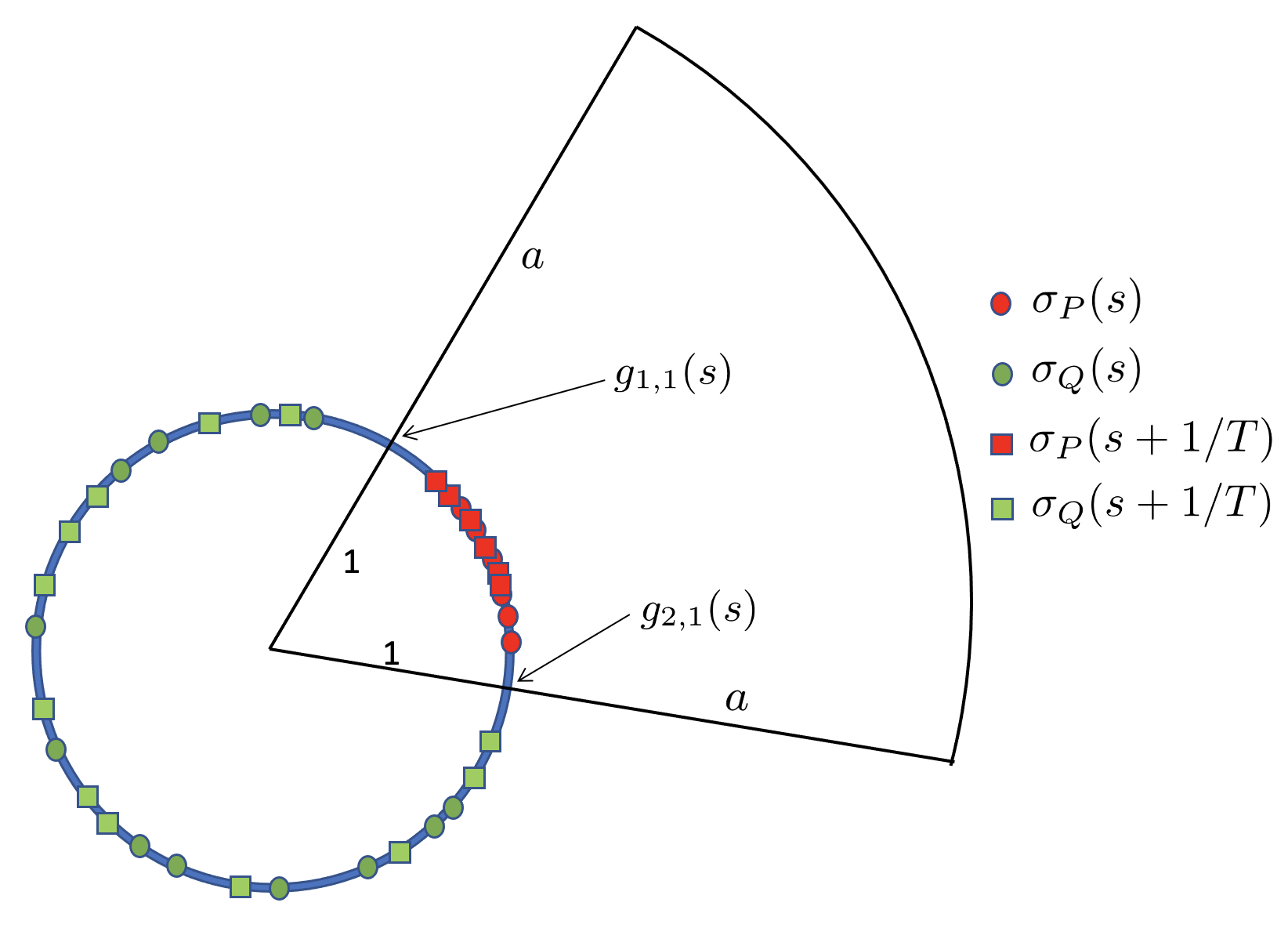}
    \caption{The contour $\Gamma_T(s,1,a)$ that passes through the two gaps, and has a closure of the contour via an arc at radius $a+1$.
    The centres of the two gaps are denoted $g_{1,1}(s)$ and $g_{2,1}(s)$, and the contour is taken to have two straight lines that are multiples of these complex numbers.}
    \label{fig:contour_1_infty}
\end{figure}

Since $R_T(s,z)$ is the resolvent of the unitary operator $W_T(s)$, we know that
\begin{equation}
\label{eq:dist}
    \left\|\left(W_T(s) - zI\right)^{-1}\right\|
    =
    \frac{1}{d\left(\sigma(W_T(s)),z\right)},
\end{equation}
where $d\left(\sigma(W_T(s)),z\right)$ is the distance between the spectrum of $W_T$ and $z$. Therefore, by separating the contour integral into three parts, two of them along the radius and one along the arc, we have that 
\begin{align}
        & \oint_{\Gamma_T(s,1,a)}  \left\|R_T\left(s+\frac{1}{T},z\right)\right\|\left\| R_T\left(s,z\right)\right\||dz| \nn
        & \leq 2 \int_0^{a+1} \frac{dx}{(x-\cos(\Delta_1(s)/2))^2 + (\sin(\Delta_1(s)/2))^2}
        + \int_{\text{arg}(g_{1,1}(s))}^{\text{arg}(g_{2,1}(s))} \frac{1}{a^2} (a+1)d\theta \nn
        & \leq 2 \int_0^{a+1} \frac{dx}{(x-\cos(\Delta_1(s)/2))^2 + (\sin(\Delta_1(s)/2))^2}
        + \frac{2\pi(a+1)}{a^2}. 
\end{align}
Here we have denoted the complex numbers in the centres of the gaps by $g_{1,1}(s)$ and $g_{2,1}(s)$.
By taking the limit $a \rightarrow \infty$, we have 
\begin{align}\label{eq:CountInt1}
    \lim_{a\to \infty} \oint_{\Gamma_T'(s,1,a)}  \left\|R_T\left(s+\frac{1}{T},z\right)\right\|\left\| R_T\left(s,z\right)\right\||dz| 
    &\le 2\int_0^{\infty}\frac{dx}{(x-\cos{\left(\Delta_1(s)/2\right)})^2+\sin{\left(\Delta_1(s)/2\right)^2}} \nn
    &= \frac{2\pi-\Delta_1(s)}{\sin{\left(\Delta_1(s)/2\right)}} \nn
    & \le \frac{4\pi}{\Delta_1(s)},
\end{align}
where in the last line we have used $(\pi-x)/\sin x \leq \pi/x$ for $0 < x \leq \pi/2$.
Note that taking the limit of $a\to\infty$ the contribution from the arc completely vanishes, and we have integrals to infinity along the two straight lines for the contour.
That gives a bound on $\|DP_T(s)\|$ as
\begin{equation}
\label{eq:CountInt}
    \|DP_T(s)\|\leq \frac{2c_1(s)}{T\Delta_1(s)}.
\end{equation}
A number of other integrals that can be obtained in a similar way.
In exactly the same way, we have
\begin{align}\label{eq:CountInt2}
    \oint_{\Gamma_T(s,1)}  \left\|R_T\left(s+\frac{1}{T},z\right)\right\|^2\left\| R_T\left(s,z\right)\right\||dz|
    &\le 2\int_0^{\infty}\frac{dx}{[(x-\cos{\left(\Delta_1(s)/2\right)})^2+\sin{\left(\Delta_1(s)/2\right)^2]^{3/2}}} \nn
    &= \frac 2{1-\cos(\Delta_1(s)/2)} .
\end{align}
The same bound holds for similar products of three terms.
Here we have written the integral as for the contour $\Gamma_T(s,1)$.
This contour can be regarded as the limit as $a\to\infty$ of the contour $\Gamma_T(s,1,a)$, but from now on we will omit the explicit procedure taking the limit.

Now we move on to the $D^{(2)}P_T$. Now since we are dealing with the second order difference, the contour should be chosen to be $\Gamma_T(s,2)$ which passes through the eigenvalue gap for three consecutive steps.
The reasoning for the contour integrals above is unchanged, except the gap $\Delta_1(s)$ is changed to $\Delta_2(s)$ for three consecutive steps.
    We therefore have 
    \begin{align}
            D^{(2)}P_T(s) & = D P_T\left(s+\frac{1}{T}\right) - D P_T\left(s\right) \nonumber\\
            & = -\frac{1}{2\pi i} \oint_{\Gamma_T(s,2)} \left[R_T\left(s+\frac{2}{T},z\right)DW_T\left(s+\frac{1}{T}\right) R_T\left(s+\frac{1}{T},z\right) - R_T\left(s+\frac{1}{T},z\right)DW_T\left(s\right) R_T\left(s,z\right)\right] dz \nonumber\\
            & = -\frac{1}{2\pi i} \oint_{\Gamma_T(s,2)} \left[R_T\left(s+\frac{2}{T},z\right) -R_T\left(s+\frac{1}{T},z\right) \right]DW_T\left(s+\frac{1}{T}\right) R_T\left(s+\frac{1}{T},z\right) dz \nonumber\\
            & \quad -\frac{1}{2\pi i} \oint_{\Gamma_T(s,2)} R_T\left(s+\frac{1}{T},z\right) D^{(2)}W_T\left(s\right) R_T\left(s+\frac{1}{T},z\right) dz \nonumber\\
            & \quad -\frac{1}{2\pi i} \oint_{\Gamma_T(s,2)} R_T\left(s+\frac{1}{T},z\right)DW_T\left(s\right) \left[R_T\left(s+\frac{1}{T},z\right) -R_T\left(s,z\right) \right] dz \nonumber\\ 
            & = \frac{1}{2\pi i} \oint_{\Gamma_T(s,2)} R_T\left(s+\frac{2}{T},z\right)DW_T\left(s\right) R_T\left(s+\frac{1}{T},z\right)DW_T\left(s+\frac{1}{T}\right) R_T\left(s+\frac{1}{T},z\right) dz \nonumber\\
            & \quad -\frac{1}{2\pi i} \oint_{\Gamma_T(s,2)} R_T\left(s+\frac{1}{T},z\right) D^{(2)}W_T\left(s\right) R_T\left(s+\frac{1}{T},z\right) dz \nonumber\\
            & \quad -\frac{1}{2\pi i} \oint_{\Gamma_T(s,2)} R_T\left(s+\frac{1}{T},z\right)DW_T\left(s\right) R_T\left(s+\frac{1}{T},z\right)DW_T\left(s\right) R_T\left(s,z\right) dz. 
        \end{align}
We can bound the first term as
\begin{align}
&    \left\| \frac{1}{2\pi i} \oint_{\Gamma_T(s,2)} R_T\left(s+\frac{2}{T},z\right)DW_T\left(s\right) R_T\left(s+\frac{1}{T},z\right)DW_T\left(s+\frac{1}{T}\right) R_T\left(s+\frac{1}{T},z\right) dz \right\| \nn
& \le  \frac{1}{2\pi} \oint_{\Gamma_T(s,2)} \left\|R_T\left(s+\frac{2}{T},z\right)\right\|\left\|DW_T\left(s\right)\right\| \left\|R_T\left(s+\frac{1}{T},z\right)\right\|\left\|DW_T\left(s+\frac{1}{T}\right)\right\| \left\|R_T\left(s+\frac{1}{T},z\right)\right\| |dz| \nn
& \le \frac{c_1(s) c_1(s+1/T)}{T^2}\frac{1}{2\pi} \oint_{\Gamma_T(s,2)} \left\|R_T\left(s+\frac{2}{T},z\right)\right\| \left\|R_T\left(s+\frac{1}{T},z\right)\right\| \left\|R_T\left(s+\frac{1}{T},z\right)\right\| |dz| \nn
& \le \frac{c_1(s) c_1(s+1/T)}{T^2}\frac{1}{\pi} \int_0^{\infty} \frac{dx}{[(x-\cos(\Delta_2(s)/2))^2+\sin(\Delta_2(s))/2)^2]^{3/2}} \nn
& = \frac{c_1(s) c_1(s+1/T)}{\pi T^2} \frac{1}{1-\cos(\Delta_2(s)/2)}.
\end{align}
For the second term we have the upper bound
\begin{align}
    &\left\| \frac{1}{2\pi i} \oint_{\Gamma_T(s,2)} R_T\left(s+\frac{1}{T},z\right) D^{(2)}W_T\left(s\right) R_T\left(s+\frac{1}{T},z\right) dz \right\| \nn
    &\le \frac{1}{2\pi} \oint_{\Gamma_T(s,2)} \left\|R_T\left(s+\frac{1}{T},z\right)\right\| \left\| D^{(2)}W_T\left(s\right)\right\| \left\| R_T\left(s+\frac{1}{T},z\right)\right\| |dz| \nn
    & \le \frac{c_2(s)}{T^2} \frac{1}{2\pi} \oint_{\Gamma_T(s,2)} \left\|R_T\left(s+\frac{1}{T},z\right)\right\| \left\| R_T\left(s+\frac{1}{T},z\right)\right\| |dz| \nn
    & \le \frac{c_2(s)}{\pi T^2}\int_0^{\infty} \frac{dx}{(x-\cos(\Delta_2(s)/2))^2+\sin(\Delta_2(s))/2)^2} \nn
    & = \frac{c_2(s)}{\pi T^2} \frac{\pi-\Delta_2(s)/2}{\sin(\Delta_2(s)/2)} \nn
    & \le \frac{2c_2(s)}{T^2} \frac{1}{\Delta_2(s)} .
\end{align}
For the third term we have identical reasoning as for the first term, except the $DW_T(s+1/T)$ is replaced with $DW_T(s)$.
That gives an upper bound
\begin{equation}
    \frac{c_1^2(s)}{\pi T^2} \frac{1}{1-\cos(\Delta_2(s)/2)} 
\end{equation}
The three bounds together give us
\begin{equation}
    \|D^{(2)}P_T(s)\| \leq \frac{ c_1(s)^2+c_1(s)c_1(s+1/T)}{\pi T^2 (1-\cos(\Delta_2(s)/2))}  + \frac{2c_2(s)}{T^2\Delta_2(s)}.
\end{equation}

\section{Proof of \texorpdfstring{\cref{lem:sum_by_parts}}{Lemma 14}}\label{ap:sum_parts}

Our initial point is noticing the following identity
\begin{equation}
\label{eq:p1}
    Q_T(s)X(s)P_T(s)=-Q_T(s)[W_T(s),\Tilde{X}(s)]P_T(s),
\end{equation}
which follows from 
\begin{align}
 [W_T(s),\Tilde{X}(s)]&= -\frac{1}{2\pi i}\oint_{\Gamma_T(s)}[W_T(s),R_T(s,z)X(s)R_T(s,z)]dz \nonumber\\
    &= -\frac{1}{2\pi i}\oint_{\Gamma_T(s)}[W_T(s)-zI,R_T(s,z)X(s)R_T(s,z)]dz \nonumber\\
    &= -\frac{1}{2\pi i}\oint_{\Gamma_T(s)}(X(s)R_T(s,z)-R_T(s,z)X(s))dz \nonumber\\
    &= [P_T(s),X(s)].\nonumber
\end{align}
Now, using the definition \cref{eq:Wa} for $W_T^A$, one gets $W_T(s)=V_T^{\dagger}(s)W_T^A(s)$. Substituting into \cref{eq:p1}, one gets 

\begin{equation}
\label{eq:p2}
    Q_T(s)X(s)P_T(s)=-Q_T(s)[W^A_T(s),\Tilde{X}(s)]P_T(s)-Q_T(s)[A(s),\Tilde{X}(s)]P_T(s),
\end{equation}
where $A(s)$ is given in \cref{eq:A}. Now, we can use $P_T(s)=U_T^{A}(s)P_0U_T^{A\dagger}(s)$ and $Q_T(s)=U_T^{A}(s)Q_0U_T^{A\dagger}(s)$ in  \cref{eq:p2} to obtain
\begin{equation}
\label{eq:p3}
    Q_0U_T^{A\dagger}(s)X(s)U_T^{A}(s)P_0 = -Q_0 U_T^{A\dagger}(s) [W^A_T(s),\Tilde{X}(s)]U_T^{A}(s)P_0 - Q_0U_T^{A\dagger}(s)[A(s),\Tilde{X}(s)]U_T^{A}(s)P_0,
\end{equation}
then,
\begin{align}
  &U_T^{A\dagger}\left(\frac{n}{T}\right)\left[W_T^{A}\left(\frac{n}{T}\right),\Tilde{X}\left(\frac{n}{T}\right)\right]U_T^A\left(\frac{n}{T}\right)\nonumber
  \\&= U_T^{A\dagger}\left(\frac{n}{T}\right)W_T^{A}\left(\frac{n}{T}\right)\Tilde{X}\left(\frac{n}{T}\right)U_T^A\left(\frac{n}{T}\right) - U_T^{A\dagger}\left(\frac{n}{T}\right)\Tilde{X}\left(\frac{n}{T}\right)W_T^{A}\left(\frac{n}{T}\right)U_T^A\left(\frac{n}{T}\right)\nonumber\\
 &= U_T^{A\dagger}\left(\frac{n}{T}\right)W_T^{A}\left(\frac{n}{T}\right)\Tilde{X}\left(\frac{n}{T}\right)U_T^A\left(\frac{n}{T}\right)-U_T^{A\dagger}\left(\frac{n}{T}\right)\Tilde{X}\left(\frac{n}{T}\right)U_T^A\left(\frac{n_+}{T}\right)\nn
 &= U_T^{A\dagger}\left(\frac{n}{T}\right)W_T^{A}\left(\frac{n_-}{T}\right)\Tilde{X}\left(\frac{n}{T}\right)U_T^A\left(\frac{n}{T}\right)-U_T^{A\dagger}\left(\frac{n}{T}\right)\Tilde{X}\left(\frac{n}{T}\right)U_T^A\left(\frac{n_+}{T}\right)\nonumber\\ 
 &\quad+ U_T^{A\dagger}\left(\frac{n}{T}\right)\left(W_T^{A}\left(\frac{n}{T}\right)-W_T^{A}\left(\frac{n_-}{T}\right)\right)\Tilde{X}\left(\frac{n}{T}\right)U_T^A\left(\frac{n}{T}\right)\nonumber\\
 &= U_T^{A\dagger}\left(\frac{n_-}{T}\right)\Tilde{X}\left(\frac{n}{T}\right)U_T^A\left(\frac{n}{T}\right)-U_T^{A\dagger}\left(\frac{n}{T}\right)\Tilde{X}\left(\frac{n}{T}\right)U_T^A\left(\frac{n_+}{T}\right)\nonumber\\ 
 &\quad+ U_T^{A\dagger}\left(\frac{n}{T}\right)DW_T^{A}\left(\frac{n_-}{T}\right)\Tilde{X}\left(\frac{n}{T}\right)U_T^A\left(\frac{n}{T}\right) \nn
  &= U_T^{A\dagger}\left(\frac{n_-}{T}\right)\Tilde{X}\left(\frac{n}{T}\right)U_T^A\left(\frac{n}{T}\right)-U_T^{A\dagger}\left(\frac{n}{T}\right)\Tilde{X}\left(\frac{n_+}{T}\right)U_T^A\left(\frac{n_+}{T}\right) \nonumber\\
  & \quad + U_T^{A\dagger}\left(\frac{n}{T}\right)DW_T^{A}\left(\frac{n_-}{T}\right)\Tilde{X}\left(\frac{n}{T}\right)U_T^A\left(\frac{n}{T}\right) + U_T^{A\dagger}\left(\frac{n}{T}\right)D\Tilde{X}\left(\frac{n}{T}\right)U_T^A\left(\frac{n_+}{T}\right),\nonumber\\
  &= U_T^{A\dagger}\left(\frac{n_-}{T}\right)\Tilde{X}\left(\frac{n}{T}\right)U_T^A\left(\frac{n}{T}\right) - U_T^{A\dagger}\left(\frac{n}{T}\right)\Tilde{X}\left(\frac{n_+}{T}\right)U_T^A\left(\frac{n_+}{T}\right) +  U_T^{A\dagger}\left(\frac{n}{T}\right)B\left(\frac{n}{T}\right)U_T^A\left(\frac{n}{T}\right),
\end{align}
where $B(s)$ is given in \cref{eq:B}. 
To complete this proof we have to multiply by $Y(s)$ on the right-hand side of Eq.~\eqref{eq:p3} and then do a sum from $1/T$ to $l/T$. First let us look to the boundary term, which is derived from the first part on the right-hand side of Eq. (\ref{eq:p3}), i.e.,
 \begin{align}
     &-\sum_{n=1}^l Q_0U_T^{A\dagger}\left(\frac{n}{T}\right)\left[W^A_T\left(\frac{n}{T}\right),\Tilde{X}\left(\frac{n}{T}\right)\right]U_T^{A}\left(\frac{n}{T}\right) P_0 Y\left(\frac{n}{T}\right)\nonumber\\ 
     &=    \sum_{n=1}^l Q_0U_T^{A\dagger}\left(\frac{n}{T}\right) \Tilde{X}\left(\frac{n_+}{T}\right)  U_T^A\left(\frac{n_+}{T}\right) P_0 Y\left(\frac{n}{T}\right) -\sum_{n=1}^l Q_0 U_T^{A\dagger}\left(\frac{n_-}{T}\right) \Tilde{X}\left(\frac{n}{T}\right) U_T^A\left(\frac{n}{T}\right) P_0 Y\left(\frac{n}{T}\right)\nonumber\\
     &\quad - \sum_{n=1}^l Q_0 U_T^{A\dagger}\left(\frac{n}{T}\right)B\left(\frac{n}{T}\right)U_T^A\left(\frac{n}{T}\right) P_0 Y\left(\frac{n}{T}\right)\nn
&=    \sum_{n=1}^l Q_0U_T^{A\dagger}\left(\frac{n}{T}\right) \Tilde{X}\left(\frac{n_+}{T}\right)  U_T^A\left(\frac{n_+}{T}\right) P_0 Y\left(\frac{n}{T}\right) -\sum_{n=0}^{l-1} Q_0 U_T^{A\dagger}\left(\frac{n}{T}\right) \Tilde{X}\left(\frac{n_+}{T}\right) U_T^A\left(\frac{n_+}{T}\right) P_0 Y\left(\frac{n_+}{T}\right)\nonumber\\
     &\quad - \sum_{n=1}^l Q_0 U_T^{A\dagger}\left(\frac{n}{T}\right)B\left(\frac{n}{T}\right)U_T^A\left(\frac{n}{T}\right) P_0 Y\left(\frac{n}{T}\right) \nn
     & = \mathcal{B}- \sum_{n=1}^l Q_0U_T^{A\dagger}\left(\frac{n}{T}\right) \Tilde{X}\left(\frac{n_+}{T}\right)  U_T^A\left(\frac{n_+}{T}\right) P_0 DY\left(\frac{n}{T}\right) - \sum_{n=1}^l Q_0 U_T^{A\dagger}\left(\frac{n}{T}\right)B\left(\frac{n}{T}\right)U_T^A\left(\frac{n}{T}\right) P_0 Y\left(\frac{n}{T}\right).
 \end{align}
 Here we have combined two sums using $Y(s)=-DY(s)+Y(s+1/T)$, and $\mathcal{B}$ is a correction accounting for the extra term at $n=0$ and the missing term at $n=l$. It is given by
\begin{equation}
    \mathcal{B}=Q_0U^{A\dagger}_T\left(\frac{l}{T}\right)\Tilde{X}\left(\frac{l_+}{T}\right)U^{A}_T\left(\frac{l_+}{T}\right)P_0Y\left(\frac{l_+}{T}\right)
    -Q_0U_T^{A\dagger}(0)\Tilde{X}\left(\frac{1}{T}\right)U^{A}_T\left(\frac{1}{T}\right)P_0Y\left(\frac{1}{T}\right).
\end{equation}
 Then plug the result above into Eq.~\eqref{eq:p3}, to give
 \begin{align}
    &\sum_{n=1}^l Q_0 U_T^{A\dagger}\left(\frac{n}{T}\right) X\left(\frac{n}{T}\right) U_T^{A}\left(\frac{n}{T}\right) P_0 Y\left(\frac{n}{T}\right)\nonumber\\
    & =  \mathcal{B}- \sum_{n=1}^l\left\{ Q_0U_T^{A\dagger}\left(\frac{n}{T}\right) \Tilde{X}\left(\frac{n_+}{T}\right)  U_T^A\left(\frac{n_+}{T}\right) P_0 DY\left(\frac{n}{T}\right)\right. \nonumber\\
     &\quad+ \left. Q_0 U_T^{A\dagger}\left(\frac{n}{T}\right)\left(\left[A\left(\frac{n}{T}\right),\Tilde{X}\left(\frac{n}{T}\right)\right]+B\left(\frac{n}{T}\right)\right) U_T^A\left(\frac{n}{T}\right) P_0 Y\left(\frac{n}{T}\right)\right\},\nonumber\\
     &=\mathcal{B}-\frac{1}{T}\mathcal{S},
 \end{align}
 where
 \begin{align}
     \mathcal{S}&=\sum_{n=1}^l\left\{ Q_0U_T^{A\dagger}\left(\frac{n}{T}\right) \Tilde{X}\left(\frac{n_+}{T}\right)  U_T^A\left(\frac{n_+}{T}\right) P_0 T DY\left(\frac{n}{T}\right)\right.\nonumber\\
     &\quad+ \left. TQ_0 U_T^{A\dagger}\left(\frac{n}{T}\right)\left(\left[A\left(\frac{n}{T}\right),\Tilde{X}(\frac{n}{T})\right]+B\left(\frac{n}{T}\right)\right) U_T^A\left(\frac{n}{T}\right) P_0 Y\left(\frac{n}{T}\right)\right\}. 
 \end{align}
 
\section{Details for the proof of \texorpdfstring{\cref{theoAdia}}{Theorem 15} and its corollary}\label{ap:theoAdia}

\subsection{Diagonal term}\label{ap:sec_diag}

Here we bound the ``diagonal'' term in \cref{eq:diag1}.
For this term (without loss of generality we only consider the term projected on $P_0$), we have 
\begin{align}
        &\quad \left\|\sum_{n=1}^{sT}P_0{U_T^A}^{\dagger}\left(\frac{n}{T}\right)\left(I - V_T^{\dagger}\left(\frac{n_-}{T}\right)\right){U_T^A}\left(\frac{n}{T}\right)P_0\Omega_T\left(\frac{n_-}{T}\right)\right\| \nonumber\\
        & = \left\|\sum_{n=1}^{sT}{U_T^A}^{\dagger}\left(\frac{n}{T}\right)P_T\left(\frac{n}{T}\right)\left(I - V_T^{\dagger}\left(\frac{n_-}{T}\right)\right)P_T\left(\frac{n}{T}\right){U_T^A}\left(\frac{n}{T}\right)\Omega_T\left(\frac{n_-}{T}\right)\right\| \nonumber\\
        & \leq \sum_{n=1}^{sT}\left\|P_T\left(\frac{n}{T}\right)\left(I - V_T^{\dagger}\left(\frac{n_-}{T}\right)\right)P_T\left(\frac{n}{T}\right)\right\|\left\|\Omega_T\left(\frac{n_-}{T}\right)\right\| \nonumber\\
        & = \sum_{n=1}^{sT}\left\|P_T\left(\frac{n}{T}\right)\left(I - V_T^{\dagger}\left(\frac{n_-}{T}\right)\right)P_T\left(\frac{n}{T}\right)\right\|. 
\end{align}
In the second line we have used \cref{eq:projU}.
Using \cref{lem:V_v2}, we have 
\begin{align}
\label{eq:diag_p1}
    & \quad \sum_{n=1}^{sT}\left\|P_T\left(\frac{n}{T}\right)\left(I - V_T^{\dagger}\left(\frac{n_-}{T}\right)\right)P_T\left(\frac{n}{T}\right)\right\| \nonumber\\
    &= \sum_{n=1}^{sT}\left\|P_T\left(\frac{n}{T}\right)\left(I - \mathcal{F}_T\left(\frac{n_-}{T}\right) +  \left(2P_T\left(\frac{n_-}{T}\right)-I\right)DP_T\left(\frac{n_-}{T}\right)\mathcal{F}_T\left(\frac{n_-}{T}\right)\right)P_T\left(\frac{n}{T}\right)\right\|\nonumber\\
    & \leq \sum_{n=1}^{sT}\left\|P_T\left(\frac{n}{T}\right)\left(I - \mathcal{F}_T\left(\frac{n_-}{T}\right)\right)P_T\left(\frac{n}{T}\right)\right\| + \sum_{n=1}^{sT}\left\|P_T\left(\frac{n}{T}\right)\left(2P_T\left(\frac{n_-}{T}\right)-I\right)DP_T\left(\frac{n_-}{T}\right)\mathcal{F}\left(\frac{n_-}{T}\right)P_T\left(\frac{n}{T}\right)\right\| \nonumber\\
    & \leq \sum_{n=0}^{sT-1}\left\|I-\mathcal{F}_T\left(\frac{n}{T}\right)\right\| + \sum_{n=1}^{sT}\left\|P_T\left(\frac{n}{T}\right)\left(2P_T\left(\frac{n_-}{T}\right)-I\right)DP_T\left(\frac{n_-}{T}\right)\left(\mathcal{F}_T\left(\frac{n_-}{T}\right) - I\right)P_T\left(\frac{n}{T}\right)\right\| \nonumber\\
    & \quad + \sum_{n=1}^{sT}\left\|P_T\left(\frac{n}{T}\right)\left(2P_T\left(\frac{n_-}{T}\right)-I\right)DP_T\left(\frac{n_-}{T}\right)P_T\left(\frac{n}{T}\right)\right\| 
\end{align}
From the second term in the last inequality we have
\begin{align}
    &\sum_{n=1}^{sT}\left\|P_T\left(\frac{n}{T}\right)\left(2P_T\left(\frac{n_-}{T}\right)-I\right)DP_T\left(\frac{n_-}{T}\right)\left(\mathcal{F}_T\left(\frac{n_-}{T}\right) - I\right)P_T\left(\frac{n}{T}\right)\right\|\nonumber\\
    &\leq \sum_{n=1}^{sT}\left\|DP_T\left(\frac{n_-}{T}\right)\left(\mathcal{F}_T\left(\frac{n_-}{T}\right) - I\right)\right\|\nonumber\\
    &=\sum_{n=0}^{sT-1}\left\|I-\mathcal{F}_T\left(\frac{n}{T}\right)\right\|  \left\| DP_T\left(\frac{n}{T}\right) \right\|.
\end{align}
Now if we replace  $P_T(n_{-}/T)=P_T(n/T)-DP_T(n_{-}/T)$ in the last term of \cref{eq:diag_p1}, i.e.,
\begin{align}
    \left\|P_T\left(\frac{n}{T}\right)\left(2P_T\left(\frac{n_-}{T}\right) - I\right) DP_T\left(\frac{n_-}{T}\right) P_T\left(\frac{n}{T}\right)\right\| &= \left\|P_T\left(\frac{n}{T}\right ) \left[2\left(P_T\left(\frac{n}{T}\right)-DP_T\left(\frac{n_{-}}{T}\right)\right)-I\right] DP_T\left(\frac{n_-}{T}\right)P_T\left(\frac{n}{T}\right)\right\|\nonumber\\
    &\leq \left\|P_T\left(\frac{n}{T}\right ) \left(2P_T\left(\frac{n}{T}\right)-I\right) DP_T\left(\frac{n_-}{T}\right)P_T\left(\frac{n}{T}\right)\right\|\nonumber\\
    &\quad+2\left\|P_T\left(\frac{n}{T}\right ) DP_T\left(\frac{n_-}{T}\right)^2 P_T\left(\frac{n}{T}\right)\right\|\nonumber\\
    &= \left\|P_T\left(\frac{n}{T}\right ) DP_T\left(\frac{n_-}{T}\right)P_T\left(\frac{n}{T}\right)\right\|\nonumber\\ 
    &\quad+ 2\left\|P_T\left(\frac{n}{T}\right ) DP_T\left(\frac{n_-}{T}\right)^2 P_T\left(\frac{n}{T}\right)\right\|\nonumber\\
    &= 3\left\|P_T\left(\frac{n}{T}\right ) DP_T\left(\frac{n_-}{T}\right)^2 P_T\left(\frac{n}{T}\right)\right\|,
    \end{align}
In the last calculation above, we used the following equality $p(p-q)p=p(p-q)^2p$ when we have $p,q$ as any two projections. Thus,  \begin{align}
\label{eq:diag_p2}
    & \quad \sum_{n=1}^{sT}\left\|P_T\left(\frac{n}{T}\right)\left(I - V_T^{\dagger}\left(\frac{n_-}{T}\right)\right)P_T\left(\frac{n}{T}\right)\right\| \nonumber\\
    & \leq \sum_{n=0}^{sT-1} \left\|I-\mathcal{F}_T\left(\frac{n}{T}\right)\right\|\left(1+\left\|DP_T\left(\frac{n}{T}\right)\right\|\right) + 3\sum_{n=1}^{sT}\left\|P_T\left(\frac{n}{T}\right)DP_T\left(\frac{n_-}{T}\right)^2 P_T\left(\frac{n}{T}\right)\right\| \nonumber\\
    & = \sum_{n=0}^{sT-1} \left\|I-\mathcal{F}_T\left(\frac{n}{T}\right)\right\|\left(1+\left\|DP_T\left(\frac{n}{T}\right)\right\|\right) + 3\sum_{n=1}^{sT}\left\|P_T\left(\frac{n}{T}\right)(DP_T\left(\frac{n_-}{T}\right)^2 P_T\left(\frac{n}{T}\right)\right\| \nonumber\\
    & \leq \sum_{n=0}^{sT-1} \left\|I-\mathcal{F}_T\left(\frac{n}{T}\right)\right\|\left(1+\left\|DP_T\left(\frac{n}{T}\right)\right\|\right) + 3\sum_{n=0}^{sT-1}\left\|DP_T\left(\frac{n}{T}\right)\right\|^2. 
\end{align}
To bound $\|\mathcal{F}_T(s) - I\|$, we can use \cref{lem:P} and the definition of $\mathcal{F}_T(s)$ as follows
\begin{align}
\label{eq:bounF-i}
    \|\mathcal{F}_T(s) - I\| &\leq \sum_{k=1}^{\infty} \frac{\Pi_{i=1}^{k}(2i-1)}{2^k k!}\left\| DP_T(s) \right\|^{2k}\nonumber\\ 
    &\leq \sum_{k=1}^{\infty} \frac{\Pi_{i=1}^{k}(2i-1)}{2^k k!}\left(\frac{2c_1(s)}{T\Delta_1(s)}\right)^{2k}\nonumber\\ 
    &=  \left(1-\frac{4c_1(s)^2}{T^2\Delta_1(s)^2}\right)^{-1/2}-1,\nonumber\\
    &=\mathcal{D}_1\left(\frac{2c_1(s)}{T\Delta_1(s)}\right)-1.
\end{align}

\subsection{Off-diagonal term}\label{ap:offD}

For the ``off-diagonal'' term in \cref{eq:offdiag1} used for \cref{theoAdia}, we have 
\begin{align}
\label{eq:offDiag_B}
    & \quad \left\|\sum_{n=1}^{sT}Q_0{U_T^A}^{\dagger}\left(\frac{n}{T}\right)\left(I - V_T^{\dagger}\left(\frac{n_-}{T}\right)\right){U_T^A}\left(\frac{n}{T}\right)P_0\Omega_T\left(\frac{n_-}{T}\right)\right\| \nonumber\\
    & \leq \frac{1}{T}\left\|\sum_{n=1}^{sT-1}Q_0{U_T^A}^{\dagger}\left(\frac{n}{T}\right)X\left(\frac{n}{T}\right){U_T^A}\left(\frac{n}{T}\right)P_0Y\left(\frac{n}{T}\right)\right\| + \frac{1}{T}\left\|Q_0{U_T^A}^{\dagger}\left(s\right)X\left(s\right){U_T^A}\left(s\right)P_0Y\left(s\right)\right\| \nonumber\\
    & \leq \frac 1T \| \mathcal{B} \| + \frac 1{T^2} \|\mathcal{S} \| + \frac{1}{T} \left\|X\left(s\right)\right\|\left\|Y\left(s\right)\right\| \nn
    & \leq \frac{1}{T}\left\|\tilde{X}\left(\frac{1}{T}\right)\right\|\left\|Y\left(\frac{1}{T}\right)\right\| + \frac{1}{T} \left\|\tilde{X}\left(s\right)\right\|\left\|Y\left(s\right)\right\| 
     + \frac{1}{T^2}\sum_{n=1}^{sT-1} \left\|Z\left(\frac{n}{T}\right)\right\| \left\|Y\left(\frac{n}{T}\right)\right\|+ \frac{1}{T}\sum_{n=1}^{sT-1} \left\|\tilde{X}\left(\frac{n_+}{T}\right)\right\| \left\|DY\left(\frac{n}{T}\right)\right\| \nonumber\\ 
     &\quad+ \frac{1}{T} \left\|X\left(s\right)\right\|\left\|Y\left(s\right)\right\| \nonumber\\
     & = \frac{1}{T}\left\|\tilde{X}\left(\frac{1}{T}\right)\right\| + \frac{1}{T}\left\|\tilde{X}\left(s\right)\right\| 
     + \frac{1}{T^2}\sum_{n=1}^{sT-1} \left\|Z\left(\frac{n}{T}\right)\right\|  + \frac{1}{T}\sum_{n=1}^{sT-1} \left\|\tilde{X}\left(\frac{n_+}{T}\right)\right\| \left\|DY\left(\frac{n}{T}\right)\right\|+ \frac{1}{T}\left\|X\left(s\right)\right\|. 
\end{align}
In the third line we have used the summation by parts result in \cref{lem:sum_by_parts} with $l=sT-1$.
In the fourth line we have used the product rule for norms of products and the fact that the spectral norms of projectors and unitary operators are 1.
In the last line we have used the fact that the choice of $Y_T$ is unitary.
Next we use the previously derived lemmas to provide bounds for the individual operators, $X(s)$, $\tilde{X}(s)$, $Z(s)$ and $DY(s)$.
Starting with $X(s)$ we have from \cref{lem:V_bound2} combined with \cref{lem:P} that
\begin{align}
    \left\|X\left(\frac{n}{T}\right)\right\|&=T\left\|V\left(\frac{n_-}{T}\right)-I\right\|\nonumber\\
    &\leq T\left\|\mathcal{F}_T\left(\frac{n_-}{T}\right) - I\right\| +  T\frac{2c_1(n_-/T)}{T\Delta_1(n_-/T)}\left\|\mathcal{F}_T\left(\frac{n_-}{T}\right)\right\|.
\end{align}
Now we use the upper bound from \cref{eq:bounF-i} to provide a bound on $\mathcal{F}_T(s)$ as 
\begin{align}
    \|\mathcal{F}_T(s)\| &\leq 1+\sum_{k=1}^{\infty} \frac{\Pi_{i=1}^{k}(2i-1)}{2^k k!}\left(\frac{2c_1(s)}{T\Delta_1(s)}\right)^{2k}\nonumber\\ 
    &=  \left(1-\frac{4c_1(n_-/T)^2}{T^2\Delta_1(n_-/T)^2}\right)^{-1/2},
\end{align}
and similarly 
\begin{align}
    \|\mathcal{F}_T(s) - I\| \leq \left(1-\frac{4c_1(n_-/T)^2}{T^2\Delta_1(n_-/T)^2}\right)^{-1/2} - 1.
\end{align}
That gives the following upper bound for $X(s)$
\begin{align}\label{eq:Xbnd}
    \left\|X\left(\frac{n}{T}\right)\right\|&\leq  T\left[\left(1+\frac{2c_1(n_-/T)}{T\Delta_1(n_-/T)}\right)\left(1-\frac{4c_1(n_-/T)^2}{T^2\Delta_1(n_-/T)^2}\right)^{-1/2}-1\right]\nonumber\\
    &=T\left[\left(1+\frac{2c_1(n_-/T)}{T\Delta_1(n_-/T)}\right)^{1/2}\left(1-\frac{2c_1(n_-/T)}{T\Delta_1(n_-/T)}\right)^{-1/2}-1\right]\nonumber\\
    &=T\mathcal{D}_2\left(\frac{2c_1(n_-/T)}{T\Delta_1(n_-/T)}\right) .
\end{align}
Now for $\tilde{X}(s)$ we can use the bound from \cref{lem:Xt}, which gives
\begin{align}\label{eq:Xtildebnd}
    \left\|\tilde{X}\left(\frac{n}{T}\right)\right\|&\leq 
    \frac{2T}{\Delta_0(n/T)}\mathcal{D}_2\left(\frac{2c_1(n_-/T)}{T\Delta_1(n_-/T)}\right).
\end{align}
For the bound on $Z(s)$ we can use \cref{eq:Zbound} from \cref{lem:ABZ2_v2}, but first we need bounds for $DX(s)$ and $DV_T(s-1/T)$, which can be obtained using \cref{lem:DV_v2} in combination with \cref{lem:P} as follows
\begin{align}\label{eq:XVrel}
    \left\|DX\left(\frac{n}{T}\right)\right\| &=T\left\|DV_T\left(\frac{n_-}{T}\right)\right\|\\ 
    &\leq T\left(1+\|DP_T(n/T)\|\right)\left\|D^{(2)}P_T(n_-/T)\right\| \mathcal{D}_3\left( \max(\|DP_T(n/T)\|,\|DP_T(n_-/T)\|) \right)\nonumber\\
    &\quad+T\|\mathcal{F}_T(n_-/T)\|\left(\left\|D^{(2)}P_T(n_-/T)\right\| + 2 \|(DP_T(n_-/T))\|^2\right)\nonumber\\
    &\leq \left(1+\frac{2c_1(n/T)}{T\Delta_1(n/T)}\right)\frac{\mathcal{G}_{T,1}(n_-/T)}{T}\mathcal{D}_3\left( \max\left(\frac{2c_1(n/T)}{T\Delta_1(n/T)},\frac{2c_1(n_-/T)}{T\Delta_1(n_-/T)}\right) \right)\nonumber\\
    &\quad +\mathcal{D}_1\left(\frac{2c_1(n_-/T)}{T\Delta_1(n_-/T)}\right)\left(\frac{\mathcal{G}_{T,1}(n_-/T)}{T} +  \frac{8c_1(n_-/T)^2}{T\Delta_1(n_-/T)^2}\right).
\end{align}
In the first line we have used $X(s) = T(I- V_T^\dagger(s-1/T))$, in the second line we have used \cref{lem:DV_v2}, and at the end we have used \cref{lem:P} in combination with the fact that the functions $\mathcal{D}_1$ and $\mathcal{D}_3$ are monotonically increasing.
Now the functions $\mathcal{G}_{T,2}$ and $\mathcal{G}_{T,3}$ from
\cref{eq:G2,eq:G3} can be used in the last expression above, to give
\begin{align}
    \left\|DX\left(\frac{n}{T}\right)\right\| &\leq \left(1+\frac{2c_1(n/T)}{T\Delta_1(n/T)}\right)\frac{\mathcal{G}_{T,2}(n_-/T)}{T})\nonumber\\
    &+\mathcal{D}_1\left(\frac{2c_1(n_-/T)}{T\Delta_1(n_-/T)}\right)\left(\frac{\mathcal{G}_{T,1}(n_-/T)}{T} +  \frac{8c_1(n_-/T)^2}{T\Delta_1(n_-/T)^2}\right)\nonumber\\
    &=\frac{\mathcal{G}_{T,3}(n_-/T)}{T}.
\end{align}
Then, from \cref{eq:XVrel} we have
\begin{equation}
    \left\|DV_T\left(\frac{n_-}{T}\right)\right\| \leq \frac{\mathcal{G}_{T,3}(n_-/T)}{T^2}.
\end{equation}
Now that we have these bounds we proceed to bound $Z(s)$ using \cref{eq:Zbound}.
Starting with the replacement of the bound of $DV_T(s-1/T)$ we can make use of the function ${\mathcal{G}_{T,4}(s)}$ as defined in \cref{eq:G4}
\begin{align}
       \left\|Z\left(\frac{n}{T}\right)\right\| &\leq \frac{4T}{\Delta_0\left(n/T\right)} \left(\left\|\mathcal{F}_T\left(\frac{n}{T}\right) - I\right\| +  \left\|DP_T\left(\frac{n}{T}\right)\right\|\left\|\mathcal{F}_T\left(\frac{n}{T}\right) \right\|\right)\left\|X\left(\frac{n}{T}\right)\right\| + \frac{2T}{\Delta_1(n/T)}\left\|DX\left(\frac{n}{T}\right)\right\| \nonumber\\
       &\quad+ \frac{2c_1(n/T)}{\pi  (1-\cos(\Delta_1(n/T)/2))} \left\|X\left(\frac{n}{T}\right)\right\| + \frac{2\mathcal{G}_{T,4}(n_-/T)}{\Delta_0(n/T)}\left\|X\left(\frac{n}{T}\right)\right\|. 
   \end{align}
Our next step is the replacement of the bound of $X(s)$,
\begin{align}
       \left\|Z\left(\frac{n}{T}\right)\right\| &\leq \frac{4T^2}{\Delta_0\left(n/T\right)} \left(\left\|\mathcal{F}_T\left(\frac{n}{T}\right) - I\right\| +  \left\|DP_T\left(\frac{n}{T}\right)\right\|\left\|\mathcal{F}_T\left(\frac{n}{T}\right) \right\|\right)\mathcal{D}_2\left(\frac{2c_1(n_-/T)}{\Delta_1(n_-/T)}\right) + \frac{2T}{\Delta_1(n/T)}\left\|DX\left(\frac{n}{T}\right)\right\| \nonumber\\
       &\quad+ \frac{2Tc_1(n/T)}{\pi  (1-\cos(\Delta_1(n/T)/2))} \mathcal{D}_2\left(\frac{2c_1(n_-/T)}{T\Delta_1(n_-/T)}\right) + \frac{2T\mathcal{G}_{T,4}(n_-/T)}{\Delta_0(n/T)}\mathcal{D}_2\left(\frac{2c_1(n_-/T)}{T\Delta_1(n_-/T)}\right). 
   \end{align} 
Now we use
\begin{equation}
\label{eq:comp_upper}
    \left\|\mathcal{F}_T\left(\frac{n}{T}\right) - I\right\| +  \left\|DP_T\left(\frac{n}{T}\right)\right\|\left\|\mathcal{F}_T\left(\frac{n}{T}\right) \right\|\leq \mathcal{D}_2\left(\frac{2c_1(n/T)}{T\Delta_1(n/T)}\right),
\end{equation}
and the bound derived above for $DX(s)$ to yield
\begin{align}\label{eq:Zbnd}
       \left\|Z\left(\frac{n}{T}\right)\right\| &\leq \frac{4T^2}{\Delta_0\left(n/T\right)} \mathcal{D}_2\left(\frac{2c_1(n/T)}{\Delta_1(n/T)}\right)\mathcal{D}_2\left(\frac{2c_1(n_-/T)}{\Delta_1(n_-/T)}\right) + \frac{2\mathcal{G}_{T,3}(n_-/T)}{\Delta_1(n/T)} \nonumber\\
       &\quad+ \frac{2Tc_1(n/T)}{\pi  (1-\cos(\Delta_1(n/T)/2))} \mathcal{D}_2\left(\frac{2c_1(n_-/T)}{T\Delta_1(n_-/T)}\right) + \frac{2T\mathcal{G}_{T,4}(n_-/T)}{\Delta_0(n/T)}\mathcal{D}_2\left(\frac{2c_1(n_-/T)}{T\Delta_1(n_-/T)}\right). 
   \end{align} 
   
 Finally, for the upper bound of $DY(s)$, first notice that $Y(s)=\Omega_T(s)$, so $D\Omega_T(s)=DY(s)$. Therefore, using \cref{lem:DOmega2} and our bound in \cref{eq:comp_upper} we obtain
\begin{equation}\label{eq:Ybnd}
\left\|DY\left(\frac{n}{T}\right) \right\| \leq \mathcal{D}_2\left(\frac{2c_1(n/T)}{T\Delta_1(n/T)}\right).
\end{equation} 

\subsection{Upper bounds for the functions of \texorpdfstring{\cref{theoAdia}}{Theorem 15}}\label{sec:uppMainF}

Starting the assumption that $T \geq \max (4\hat{c}_1(s)/\check{\Delta}(s))$, from the upper bounds given in \cref{eq:upperB_D} and from the inequality $1-\cos(\theta/2) \geq \theta^2/\pi^2$ for all $0\leq \theta \leq \pi$, we have
\begin{align}\label{eq:GT1app}
    \mathcal{G}_{T,1}(n_-/T) &= \frac{ c_1(n_-/T)^2+c_1(n_-/T)c_1(n/T)}{\pi (1-\cos(\Delta_2(n_-/T)/2))}  + \frac{2c_2(n_-/T)}{\Delta_2(n_-/T)}\nonumber \\ 
    & \leq \frac{ \pi c_1(n_-/T)^2 + \pi c_1(n_-/T)c_1(n/T)}{ \Delta_2(n_-/T)^2}  + \frac{2c_2(n_-/T)}{\Delta_2(n_-/T)} \nonumber \\
    & \leq \frac{2\pi\hat{c}_1(n/T)^2}{\check{\Delta}(n/T)^2} + \frac{2\hat{c}_2(n/T)}{\check{\Delta}(n/T)}. 
\end{align}
Therefore 
\begin{align}\label{eq:GT2app}
    \mathcal{G}_{T,2}(n_-/T) & \leq \xi_3 \frac{2\hat{c}_1(n/T)}{T\check{\Delta}(n/T)} \left(\frac{2\pi\hat{c}_1(n/T)^2}{\check{\Delta}(n/T)^2} + \frac{2\hat{c}_2(n/T)}{\check{\Delta}(n/T)}\right) \nonumber \\
    & = \frac{4\pi\xi_3 \hat{c}_1(n/T)^3}{T\check{\Delta}(n/T)^3} + \frac{4\xi_3 \hat{c}_1(n/T)\hat{c}_2(n/T)}{T\check{\Delta}(n/T)^2}, 
\end{align}
\begin{align}\label{eq:GT3app}
    \mathcal{G}_{T,3}(n_-/T) &\leq \frac{3}{2}\left(\frac{4\pi\xi_3 \hat{c}_1(n/T)^3}{T\check{\Delta}(n/T)^3} + \frac{4\xi_3\hat{c}_1(n/T)\hat{c}_2(n/T)}{T\check{\Delta}(n/T)^2}\right) + \xi_1 \left(\frac{2\pi\hat{c}_1(n/T)^2}{\check{\Delta}(n/T)^2} + \frac{2\hat{c}_2(n/T)}{\check{\Delta}(n/T)} + \frac{8\hat{c}_1(n/T)^2}{\check{\Delta}(n/T)^2}\right) \nonumber \\
    & \leq \frac{3}{2}\left(\frac{\pi\xi_3  \hat{c}_1(n/T)^2}{\check{\Delta}(n/T)^2} + \frac{\xi_3 \hat{c}_2(n/T)}{\check{\Delta}(n/T)}\right) + \xi_1\left(\frac{(2\pi+8)\hat{c}_1(n/T)^2}{\check{\Delta}(n/T)^2} + \frac{2\hat{c}_2(n/T)}{\check{\Delta}(n/T)}\right) \nonumber \\
    & \leq \left(3\pi\xi_3/2 + (2\pi+8)\xi_1\right)\frac{\hat{c}_1(n/T)^2}{\check{\Delta}(n/T)^2} + \left(3\xi_3/2+2\xi_1\right)\frac{\hat{c}_2(n/T)}{\check{\Delta}(n/T)}, 
\end{align}
and 
\begin{align}\label{eq:GT4app}
    \mathcal{G}_{T,4}(n_-/T) \leq \left(3\pi\xi_3/2 + (2\pi+8)\xi_1\right)\frac{\hat{c}_1(n/T)^2}{T\check{\Delta}(n/T)^2} + \left(3\xi_3/2+2\xi_1\right)\frac{\hat{c}_2(n/T)}{T\check{\Delta}(n/T)} + \hat{c}_1(n/T). 
\end{align}
These bounds are used in the body of the paper in \cref{eq:GT1body,eq:GT2body,eq:GT3body,eq:GT4body}.

Plugging all these bounds back to \cref{theoAdia} and using $1-\cos(\theta/2) \geq \theta^2/\pi^2$ again, we have 
\begin{align}\label{eq:Uerrbndapp}
    & \quad \|U_T(s) - U_T^A(s)\| \nonumber\\
    &\leq  \frac{8\xi_2 \hat{c}_1(0)}{T\check{\Delta}(0)^2}+  \frac{8\xi_2 \hat{c}_1(s)}{T\check{\Delta}(s)^2} + \frac{4\xi_2 \hat{c}_1(s)}{T\check{\Delta}(s)} + \sum_{n=1}^{sT-1}\frac{12}{\check{\Delta}(n/T)}\xi_2^2\left(\frac{2\hat{c}_1(n/T)}{T\check{\Delta}(n/T)}\right)^2  \nonumber\\
    & \quad + \sum_{n=1}^{sT-1}\left(6\pi\xi_3 + (8\pi+32)\xi_1\right)\frac{\hat{c}_1(n/T)^2}{T^2\check{\Delta}(n/T)^3} + \sum_{n=1}^{sT-1}\left(6\xi_3+8\xi_1\right)\frac{\hat{c}_2(n/T)}{T^2\check{\Delta}(n/T)^2} + \sum_{n=1}^{sT-1} \frac{4\pi \hat{c}_1(n/T)}{ T \check{\Delta}(n/T)^2} \xi_2 \frac{2\hat{c}_1(n/T)}{T\check{\Delta}
       (n/T)} \nonumber \\
       & \quad + \sum_{n=1}^{sT-1}\left(6\pi\xi_3 + (8\pi+32)\xi_1\right) \frac{\hat{c}_1(n/T)^2}{T^2\check{\Delta}(n/T)^3}\left(\xi_2\frac{2\hat{c}_1(n/T)}{T\check{\Delta}
       (n/T)}\right) + \sum_{n=1}^{sT-1}\left(6\xi_3 + 8\xi_1\right) \frac{\hat{c}_2(n/T)}{T^2\check{\Delta}(n/T)^2}\left(\xi_2\frac{2\hat{c}_1(n/T)}{T\check{\Delta}
       (n/T)}\right) \nonumber \\
       & \quad + \sum_{n=1}^{sT-1} \frac{4\hat{c}_1(n/T)}{T\check{\Delta}(n/T)}\left(\xi_2\frac{2\hat{c}_1(n/T)}{T\check{\Delta}
       (n/T)}\right)  + \sum_{n=0}^{sT-1} \frac{24\hat{c}_1(n/T)^2}{T^2\check{\Delta}(n/T)^2} +  \sum_{n=0}^{sT-1}\frac{8\hat{c}_1(n/T)^2}{T^2\check{\Delta}(n/T)^2} \nonumber \\
       & \leq \frac{8\xi_2\hat{c}_1(0)}{T\check{\Delta}(0)^2}+  \frac{8\xi_2\hat{c}_1(s)}{T\check{\Delta}(s)^2} + \frac{4\xi_2\hat{c}_1(s)}{T\check{\Delta}(s)} + \sum_{n=1}^{sT-1}\frac{48\xi_2^2\hat{c}_1(n/T)^2}{T^2\check{\Delta}(n/T)^3} \nonumber\\
    & \quad + \sum_{n=1}^{sT-1}\left(6\pi\xi_3+(8\pi+32)\xi_1\right)\frac{\hat{c}_1(n/T)^2}{T^2\check{\Delta}(n/T)^3} + \sum_{n=1}^{sT-1}(6\xi_3+8\xi_1)\frac{\hat{c}_2(n/T)}{T^2\check{\Delta}(n/T)^2}+ \sum_{n=1}^{sT-1} \frac{8\pi\xi_2 \hat{c}_1(n/T)^2}{ T^2 \check{\Delta}(n/T)^3}  \nonumber \\
       & \quad + \sum_{n=1}^{sT-1}\left(12\pi\xi_2\xi_3+(16\pi+64)\xi_1\xi_2\right) \frac{\hat{c}_1(n/T)^3}{T^3\check{\Delta}(n/T)^4} + \sum_{n=1}^{sT-1} \left(12\xi_2\xi_3+16\xi_1\xi_2\right)\frac{\hat{c}_1(n/T)\hat{c}_2(n/T)}{T^3\check{\Delta}(n/T)^3} \nonumber \\
       & \quad + \sum_{n=1}^{sT-1} \frac{8\xi_2 \hat{c}_1(n/T)^2}{T^2\check{\Delta}(n/T)^2} + \sum_{n=0}^{sT-1} \frac{24\hat{c}_1(n/T)^2}{T^2\check{\Delta}(n/T)^2} +  \sum_{n=0}^{sT-1}\frac{8\hat{c}_1(n/T)^2}{T^2\check{\Delta}(n/T)^2}  \nonumber \\
       & \leq \frac{8\xi_2\hat{c}_1(0)}{T\check{\Delta}(0)^2}+  \frac{8\xi_2\hat{c}_1(s)}{T\check{\Delta}(s)^2} + \frac{4\xi_2\hat{c}_1(s)}{T\check{\Delta}(s)} +  \sum_{n=1}^{sT-1} \left(48\xi_2^2 + 6\pi \xi_3 + (8\pi+32)\xi_1 + 8\pi \xi_2\right) \frac{\hat{c}_1(n/T)^2}{T^2\check{\Delta}(n/T)^3} \nonumber \\
       & \quad + \sum_{n=1}^{sT-1}(6\xi_3+8\xi_1)\frac{\hat{c}_2(n/T)}{T^2\check{\Delta}(n/T)^2}+ \sum_{n=0}^{sT-1} \frac{(32+8\xi_2) \hat{c}_1(n/T)^2}{ T^2 \check{\Delta}(n/T)^2}  \nonumber \\
       & \quad + \sum_{n=1}^{sT-1}\left(12\pi\xi_2\xi_3+(16\pi+64)\xi_1\xi_2\right) \frac{\hat{c}_1(n/T)^3}{T^3\check{\Delta}(n/T)^4}  + \sum_{n=1}^{sT-1} \left(12\xi_2\xi_3+16\xi_1\xi_2\right)\frac{\hat{c}_1(n/T)\hat{c}_2(n/T)}{T^3\check{\Delta}(n/T)^3}.  
\end{align}
This result is used in the body of the paper in \cref{eq:Uerrbndbody}.

\section{Block encoding of \texorpdfstring{$H(s)$}{H(s)}}\label{app:blockHs}
Here we describe how to perform the block encoding of $H(s)$ as given in \cref{eq:Hsencoding}.
We denote the unitary for the block encoding of $A$ as $U_A$, which acts on an ancilla denoted with subscript $a$ and the system such that
\begin{equation}
    {}_a \! \bra{0} U_A \ket{0}_a = A.
\end{equation}
We also denote the unitary oracle for preparing $\ket{b}$ as $U_b$ such that
\begin{equation}
    U_b \ket{0} = \ket{b}.
\end{equation}
As well as the ancilla system used for the block encoding of $A$, we use four ancilla qubits.
These ancilla qubits are used as follows.
\begin{enumerate}
    \item The first selects between the blocks in $A(f)$.
    \item The next is used for preparing the combination of $\sigma_z\otimes I$ and $\mathbf{A}$.
    \item The third is used in implementing $Q_\mathbf{b}$.
    \item The fourth selects between the blocks in $H(s)$.
\end{enumerate}
These three qubits will be denoted with subscripts $a_1$ to $a_4$.

First consider $A(f)$, which can be written as
\begin{equation}
    A(f) = (1-f) \sigma^z_{a_1} \otimes I_N + f (\ketbra{0}{1}_{a_1}\otimes A + \ketbra{1}{0}_{a_1}\otimes A^\dagger).
\end{equation}
Note that the first operators in the tensor products here, $\sigma^z$ and $\ketbra{0}{1}$ or $\ketbra{1}{0}$, act upon the ancilla denoted $a_1$.
To block encode the operation using ancilla $a_2$, we can use the select operation
\begin{equation}
   U_{A(f)} = \ketbra{0}{0}_{a_2} \otimes \sigma^z_{a_1} \otimes I_{N} \otimes I_a +
    \ketbra{1}{1}_{a_2} \otimes (\ketbra{0}{1}_{a_1}\otimes U_A + \ketbra{1}{0}_{a_1}\otimes U_A^\dagger).
\end{equation}
Here we have included $I_a$ to indicate that the operation is acting as the identity on the ancilla system used for the block encoding of $A$.
Note that we require the ability to apply the oracle $U_A$ in a selected way, where we either perform $U_A$, $U_A^\dagger$ or the identity.

Next consider $Q_{\mathbf{b}}$, which is given by
\begin{equation}
    Q_{\mathbf{b}} = I_{a_1} \otimes I_N - \ketbra{1}{1}_{a_1} \otimes \ketbra{b}{b}.
\end{equation}
Here we have used the ancilla $a_1$ to account for using $\mathbf{b}$ which is encoded as $\ket{1}_{a_1}\otimes \ket{b}$.
We can construct this projector using
\begin{equation}
    (I_{a_1} \otimes U_b^\dagger) \left[ I_{a_1} \otimes I_N - \ketbra{1}{1}_{a_1} \otimes \ketbra{0}{0}_N \right] (I_{a_1} \otimes U_b),
\end{equation}
where we are using subscript $N$ on $\ketbra{0}{0}$ to indicate it is on the system.
We can block encode this projector using the ancilla $a_4$.
We simply need to create this ancilla in an equal superposition, and use the unitary operation
\begin{equation}
   U_{Qb} = (I_{a_3} \otimes I_{a_1} \otimes U_b^\dagger) \left[ \ketbra{0}{0}_{a_3} \otimes I_{a_1} \otimes I_N 
    +\ketbra{1}{1}_{a_3} \otimes (2I_{a_1} \otimes I_N 
    - \ketbra{1}{1}_{a_1} \otimes \ketbra{0}{0}_N) \right] (I_{a_3} \otimes I_{a_1} \otimes U_b).
\end{equation}
That gives the projector as a linear combination of the identity and a reflection.
Finally we are prepared to describe the unitary to block encode $H(s)$, which can be written as
\begin{equation}
    H(s) = \ketbra{0}{1}_{a_4} \otimes A(f(s)) Q_{\mathbf{b}} + \ketbra{1}{0}_{a_4} \otimes Q_{\mathbf{b}} A(f(s)).
\end{equation}

In order to select between $A(f(s)) Q_{\mathbf{b}}$ and $Q_{\mathbf{b}} A(f(s))$, we will apply $Q_{\mathbf{b}}$ in a controlled way before and after $A(f(s))$.
We will denote the controlled unitary for $Q_{\mathbf{b}}$, as controlled on 0 or 1, by $CU^0_{Qb}$ or $CU^1_{Qb}$, respectively.
We may make $Q_{\mathbf{b}}$ controlled simply by making the reflection $2I_{a_1} \otimes I_N 
- \ketbra{1}{1}_{a_1} \otimes \ketbra{0}{0}_N$ controlled, and we do not need to make the oracle $U_b$ controlled.
We can therefore apply $CU^1_{Qb}$ as
\begin{align}
    CU^1_{Qb} &= \ketbra{0}{0}_{a_4} \otimes I_{a_3} \otimes I_{a_1} \otimes I_N + \ketbra{1}{1}_{a_4} \otimes U_{Qb} \nn
    &= (I_{a_3} \otimes I_{a_1} \otimes U_b^\dagger) \left[ I_{a_4} \otimes \ketbra{0}{0}_{a_3} \otimes I_{a_1} \otimes I_N 
    +\ketbra{0}{0}_{a_4} \otimes \ketbra{1}{1}_{a_3} \otimes I_{a_1} \otimes I_N \right. \nn  & \quad \left.
    +\ketbra{1}{1}_{a_4} \otimes \ketbra{1}{1}_{a_3} \otimes (2I_{a_1} \otimes I_N 
    - \ketbra{1}{1}_{a_1} \otimes \ketbra{0}{0}_N)  
    \right] (I_{a_3} \otimes I_{a_1} \otimes U_b)
\end{align}
and similarly for $CU^0_{Qb}$.

We also need to perform the rotation $R(s)$ before or after these operations controlled on the ancilla $a_4$.
That is, we will perform at the beginning
\begin{equation}
    CR^0(s) = \ketbra{0}{0}_{a_4} \otimes R(s)_{a_2} + \ketbra{1}{1}_{a_4} \otimes \mathcal{H}_{a_2},
\end{equation}
where $\mathcal{H}$ denotes the Hadamard operation.
Then at the end we perform the controlled operation
\begin{equation}
    CR^1(s) = \ketbra{1}{1}_{a_4} \otimes R(s)_{a_2} + \ketbra{0}{0}_{a_4} \otimes \mathcal{H}_{a_2},
\end{equation}

We are finally ready to provide the complete sequence of operations to block encode $H(s)$.
In the following we will use the various operations defined above on subsets of the ancillas, with the convention that they act as the identity on any ancillas their action has not been described on.
\begin{enumerate}
    \item We apply the Hadamard on $a_3$ to provide the linear combination needed for $Q_\mathbf{b}$.
    \item Next apply $CU^1_{Qb}$ for controlled implementation of $Q_\mathbf{b}$ before $A(f(s))$.
    \item Apply $CR^0(s)$ to provide the rotation on ancilla $a_2$.
    \item Apply $U_{A(f)}$ for the block encoding of $A(f(s))$.
    \item Apply $CR^1(s)$ to provide the symmetric form of the rotation on ancilla $a_2$.
    \item Apply $CU^0_{Qb}$ for controlled implementation of $Q_\mathbf{b}$ after $A(f(s))$.
    \item Apply the Hadamard on $a_3$ again.
    \item Finally, apply $\sigma^x$ on $a_4$ to flip that bit.
\end{enumerate}
Now recall that we require the unitary operation in the block encoding to be self-inverse for the qubitisation.
To see that this sequence of operations is self-inverse, first note that each of the individual operations is self-inverse.
Second, note that $CU^0_{Qb}=\sigma^x_{a_4} CU^0_{Qb} \sigma^x_{a_4}$ and $CR^0(s)=\sigma^x_{a_4}CR^1(s)\sigma^x_{a_4}$.
That is, we may flip between controlling on 0 and 1 by applying the \textsc{not} gate to $a_4$.
Therefore we have the complete operation squared given by
\begin{align}
    &\left[ \sigma^x_{a_4} \,\mathcal{H}_{a_3}\, CU^0_{Qb}\,CR^1(s)\,U_{A(f)}\,CR^0(s)\,CU^1_{Qb}\,\mathcal{H}_{a_3}\right]
\left[    \sigma^x_{a_4}\, \mathcal{H}_{a_3} \,CU^0_{Qb}\,CR^1(s)\,U_{A(f)}\,CR^0(s)\,CU^1_{Qb}\,\mathcal{H}_{a_3} \right]\nn
&= \sigma^x_{a_4} \,\mathcal{H}_{a_3}\, CU^0_{Qb}\,CR^1(s)\,U_{A(f)}\,CR^0(s)\,CU^1_{Qb}\,
 \sigma^x_{a_4}\, CU^0_{Qb}\,CR^1(s)\,U_{A(f)}\,CR^0(s)\,CU^1_{Qb}\,\mathcal{H}_{a_3} \nn
&= \sigma^x_{a_4} \,\mathcal{H}_{a_3}\, CU^0_{Qb}\,CR^1(s)\,U_{A(f)}\,CR^0(s)\,\sigma^x_{a_4}\,CU^0_{Qb}\,
  CU^0_{Qb}\,CR^1(s)\,U_{A(f)}\,CR^0(s)\,CU^1_{Qb}\,\mathcal{H}_{a_3} \nn
&= \sigma^x_{a_4} \,\mathcal{H}_{a_3}\, CU^0_{Qb}\,CR^1(s)\,U_{A(f)}\,CR^0(s)\,\sigma^x_{a_4}\,CR^1(s)\,U_{A(f)}\,CR^0(s)\,CU^1_{Qb}\,\mathcal{H}_{a_3} \nn
&= \sigma^x_{a_4} \,\mathcal{H}_{a_3}\, CU^0_{Qb}\,CR^1(s)\,U_{A(f)}\,\sigma^x_{a_4}\,CR^1(s)\,CR^1(s)\,U_{A(f)}\,CR^0(s)\,CU^1_{Qb}\,\mathcal{H}_{a_3} \nn
&= \sigma^x_{a_4} \,\mathcal{H}_{a_3}\, CU^0_{Qb}\,CR^1(s)\,U_{A(f)}\,\sigma^x_{a_4}\,U_{A(f)}\,CR^0(s)\,CU^1_{Qb}\,\mathcal{H}_{a_3} \nn
&= \sigma^x_{a_4} \,\mathcal{H}_{a_3}\, CU^0_{Qb}\,CR^1(s)\,\sigma^x_{a_4}\,U_{A(f)}\,U_{A(f)}\,CR^0(s)\,CU^1_{Qb}\,\mathcal{H}_{a_3} \nn
&= \sigma^x_{a_4} \,\mathcal{H}_{a_3}\, CU^0_{Qb}\,CR^1(s)\,\sigma^x_{a_4}\,CR^0(s)\,CU^1_{Qb}\,\mathcal{H}_{a_3} \nn
&= \sigma^x_{a_4} \,\mathcal{H}_{a_3}\, CU^0_{Qb}\,\sigma^x_{a_4}\,CR^0(s)\,CR^0(s)\,CU^1_{Qb}\,\mathcal{H}_{a_3} \nn
&= \sigma^x_{a_4} \,\mathcal{H}_{a_3}\, CU^0_{Qb}\,\sigma^x_{a_4}\,CU^1_{Qb}\,\mathcal{H}_{a_3} \nn
&= \sigma^x_{a_4} \,\mathcal{H}_{a_3}\, \sigma^x_{a_4}\,CU^1_{Qb}\,CU^1_{Qb}\,\mathcal{H}_{a_3} \nn
&= \sigma^x_{a_4} \,\mathcal{H}_{a_3}\, \sigma^x_{a_4}\,\mathcal{H}_{a_3} \nn
&= I.
\end{align}
Here we have repeatedly commuted $\sigma^x_{a_4}$ through operators, and used the property that operators are self-inverse to cancel them.
This shows that our sequence of operations is self-inverse as required.

\section{Upper bounds of \texorpdfstring{\cref{cor:adia}}{Theorem 3} with \texorpdfstring{$p=3/2$}{p=3/2}}\label{ap:upperB_cor}

We split the proof of \cref{theo:p15} into three parts: the upper bounds for the three terms without sums, the sums with $\hat{c}_1$, and the summation term with $\hat{c}_2$.
Before we proceed with each calculation, first we note that in \cref{cor:adia} the three gaps are replaced by the minimum one, \cref{eq:fhat}, that is
 \begin{equation}
\check{\Delta}(s) = \min_{s' \in \left\{s-1/T,s,s+1/T\right\} \cap [0,1] } \Delta(s').
\end{equation}
We have \cref{eq:Gaps} using the fact that the gap is monotonically decreasing, so then the fact that $f$ is monotonically increasing gives us
\begin{equation}
\label{eqn:qlsp_gap_check}
    \check{\Delta}(s) = \begin{cases}       
       (1- f(s+3/T) + f(s+3/T)/\kappa), & 0 \leq s \leq 1-3/T, \\
       1/\kappa, & s = 1-2/T, 1-1/T, 1.  \\
    \end{cases}
\end{equation}

Next, in \cref{cor:adia} we have the functions $\hat{c}_1(s)$ and $\hat{c}_2(s)$ as defined in \cref{eq:chat}.
Choices for the functions $c_1(s)$ and $c_2(s)$ are given in \cref{lem:DR}.
Using the monotonicity properties of the function $f$, we find
\begin{equation}
\label{eq:cic1corr}
\hat{c}_1(s) =   \begin{cases}
       2Tf(1/T), & s= 0, \\
       2T(f(s)-f(s-1/T)), & 1/T \leq s \leq 1,
\end{cases}
\end{equation}
and
\begin{equation}
\label{eq:cic2corr}
\hat{c}_2(s) = 2\left(2|f'(s)|^2 + |f''(s)|\right).
\end{equation}
For $\hat{c}_1(s)$ we have used the fact that $f'(s)$ is monotonically decreasing, so a larger difference will be obtained for a smaller value of $s$.
For $\hat{c}_2(s)$ we have also used the fact that $|f''(s)|$ is monotonically decreasing, so again larger values will be obtained for smaller values of $s$.
The monotonicity properties of $f$ are easily checked by checking expressions for the derivatives; $f'(s)$ is positive, $f''(s)$ is negative and $f'''(s)$ is positive.

In the block encoding we need to account for how the gap in $H(s)$ is translated to the gap in the walk operators.
The solution state has eigenvalue $0$, which is translated to the eigenvalues $\pm 1$ for the walk operator.
The eigenvalues $\lambda$ of $H$ are generally translated to $\pm e^{\pm i\arcsin \lambda}$, which means the gap for the walk operator is increased to the arcsine of the gap of the Hamiltonian.
Since the arcsine can only increase the gap, the lower bounds on the gap for $H(s)$ also apply to the walk operator.

\subsection{Single components}\label{appsec:single}

Beginning with the first term from the bound in \cref{cor:adia}, using the expression for $\check{\Delta}(0)$ from \cref{eqn:qlsp_gap_check}, for $\hat{c}_1(0)$ from \cref{eq:cic1corr}, and $f(s)$ from \cref{eq:sched1}, we get
\begin{align}
\label{eq:firstupperB}
    \frac{\hat{c}_1(0)}{T\check{\Delta}(0)^2}& = 2 \frac{f(1/T)}{(1-f(3/T)+f(3/T)/\kappa)^2}\nonumber\\
    &=\frac{2}{ T^4}\frac{\kappa}{\sqrt{\kappa}+1}\frac{\left(3\sqrt{\kappa}-3+T\right)^4\left(\sqrt{\kappa}-1+2T\right)}{(\sqrt{\kappa}-1+T)^2}\nonumber\\
    &=\frac{4}{ T}\frac{\kappa}{\sqrt{\kappa}+1} \frac{(1+2\alpha_1)^4 (1-\alpha_1/2)}{(1-\alpha_1)^3} \nn
    &=\frac{4}{ T}\frac{\kappa}{\sqrt{\kappa}+1} \left[1+\mathcal{O}(\alpha_1)\right] \nn
    &= \frac{4\sqrt{\kappa}}{T} + \mathcal{O}\left( \frac{\kappa}{T^2} \right),
\end{align}
where
\begin{equation}
    \alpha_n:=\frac{\sqrt\kappa-1}{T+n(\sqrt\kappa-1)},
\end{equation}
so $\alpha_n=\mathcal{O}(\sqrt\kappa/T)$,
and we have used $T>\kappa$.
This result is given in \cref{eq:single_c(0)} of the body.

We next show \cref{eq:single_c(1),eq:single_c(1)_2}.
This time we use $\hat{c}_1(s)$ and $\check{\Delta}(s)$ for $s=1$; by \cref{eq:cic1corr}  we get $\hat{c}_1(1)=2(1-f(1-1/T))$ and from \cref{eqn:qlsp_gap_check} we have            $\check{\Delta}(1)=1/\kappa$.
Therefore 
\begin{align}
    \frac{\hat{c}_1(1)}{T\check{\Delta}(1)^2}& = 2\kappa^2 (1-f(1-1/T))\nonumber\\
    & = 2\kappa^2 \left[1+\frac{\kappa}{1-\kappa}\left(1-\frac{1}{(1+(\sqrt{\kappa} -1)(1-1/T))^2}\right)\right].
\end{align}
Now we simplify the terms inside the square brackets to give
\begin{align}
\label{eq:bracket}
    1+\frac{\kappa}{1-\kappa}\left[1-\frac{T^2}{(T+(\sqrt{\kappa} -1)(T-1))^2}\right] &= \frac{(1-\kappa)\left(1+\sqrt{\kappa}(T-1)\right)^2+\kappa\left(1+\sqrt{\kappa}(T-1)\right)^2 - T^2\kappa}{\left(1 - \kappa \right)\left(1+\sqrt{\kappa}(T-1)\right)^2}\nonumber\\
    &= \frac{\left(1+\sqrt{\kappa}(T-1)\right)^2 - T^2\kappa}{\left(1 - \kappa \right)\left(1+\sqrt{\kappa}(T-1)\right)^2}\nonumber\\
    &=\frac{\sqrt{\kappa}( 2 T-1) +1}{(\sqrt{\kappa}+1)(1+\sqrt{\kappa}(T-1))^2} \nn
    &=\frac{2}{T(\kappa+\sqrt{\kappa})} 
    \frac{1-\beta/2}{(1-\beta)^2} \nn
    &=\frac{2}{T(\kappa+\sqrt{\kappa})} \left[1+\mathcal{O}(\beta)\right] \nn
    &=\frac{2}{\kappa T} + \mathcal{O}\left( \frac{1}{\kappa T^2} \right),
\end{align}
with
\begin{equation}
    \beta=\frac{1-1/\sqrt\kappa} T,
\end{equation}
so $\beta=\mathcal{O}(1/T)$.
Therefore, we can conclude
\begin{equation}
\label{eq:SecondupperB}
    \frac{\hat{c}_1(1)}{ T\check{\Delta}(1)^2} 
    =\frac{4\kappa}{T} + \mathcal{O}\left( \frac{\kappa}{T^2} \right).
\end{equation}
This is the result given in \cref{eq:single_c(1)}.
For the other upper bound, we have $\check{\Delta}(1)$ instead of $\check{\Delta}(1)^2$, so get for the upper bound shown in \cref{eq:single_c(1)_2}
\begin{equation}
\label{eq:thirdupperB}
    \frac{\hat{c}_1(1)}{T\check{\Delta}(1)} 
    =\frac{4}{T} + \mathcal{O}\left( \frac{1}{T^2} \right).
\end{equation}

\subsection{\texorpdfstring{$c_1(s)$}{c1(s)} summations}
\label{appsec:summc1}
We start by considering the sum of $\hat{c}_1(s)^2/(T^2\check{\Delta}(s)^3)$ for $1/T\leq s\leq 1-3/T$.
In this range we get
\begin{align}
    \frac{\hat{c}_1(n/T)^2}{T^2\check{\Delta}(n/T)^3} &= 4 \frac{(f(n/T)-f((n-1)/T))^2}{(1-f((n+3)/T)+f((n+3)/T)/\kappa)^3} \nn
&= \frac{16\kappa^2}{(\sqrt\kappa + 1)^2T^2}\frac{\left[(3 + n)(\sqrt{\kappa}-1)+T\right]^6\left[(n - 1/2)(\sqrt{\kappa}-1)+T\right]^2}{\left[n(\sqrt{\kappa}-1)+T\right]^4\left[(n-1)(\sqrt{\kappa}-1)+T\right]^4}\nonumber\\
&= \frac{16\kappa^2}{(\sqrt\kappa + 1)^2T^2}
\frac{(1+3\alpha_n)^6(1-\alpha_n/2)^2}{(1-\alpha_n)^4}
\nonumber\\
&= \frac{16\kappa^2}{(\sqrt\kappa + 1)^2T^2}
\left[1+\mathcal{O}(\alpha_n)\right] \nn
    &=  \frac{16 \kappa}{T^2}+ \mathcal{O} \left( \frac{\kappa^{3/2}}{T^3} \right).
\end{align}
Now for the last two elements of the sum we have
\begin{align}
\label{eq:upper-1/t}
   \frac{\hat{c}_1(1-2/T)^2}{T^2\check{\Delta}(1-2/T)^3} &= 4 \kappa^3(f(1-2/T) - f(1-3/T) )^2\nn
   &=
   \frac{16 \kappa^2}{
 T^2 (1 + \sqrt\kappa)^2} \frac{[1 - 5/(2 T) + 5/(
     2 T \sqrt\kappa)]^2}{[1 + 6/T^2 - 5/T + 6/(T^2 \kappa) - 12/(
     T^2 \sqrt\kappa) + 5/(T \sqrt\kappa)]^4}\nn
     &=\frac{16 \kappa^2}{
 T^2 (1 + \sqrt\kappa)^2} \left[ 1+\mathcal{O}\left( \frac 1T \right) \right] \nn
   &=\frac{16 \kappa}{T^2}+ \mathcal{O} \left( \frac{\kappa}{T^3} \right),
\end{align}
and
\begin{align}
\label{eq:upper-2/t}
    \frac{\hat{c}_1(1-1/T)^2}{T^2\check{\Delta}(1-1/T)^3} &= 4\kappa^3(f(1-1/T) - f(1-2/T) )^2 \nn
   &=
   \frac{16 \kappa^2}{
 T^2 (1 + \sqrt\kappa)^2} \frac{[1 - 3/(2 T) + 3/(
     2 T \sqrt\kappa)]^2}{[1 + 2/T^2 - 3/T + 2/(T^2 \kappa) - 4/(
     T^2 \sqrt\kappa) + 3/(T \sqrt\kappa)]^4}\nn
     &=\frac{16 \kappa^2}{
 T^2 (1 + \sqrt\kappa)^2} \left[ 1+\mathcal{O}\left( \frac 1T \right) \right] \nn    
    &=\frac{16 \kappa}{T^2}+ \mathcal{O} \left( \frac{\kappa}{T^3} \right).
\end{align}
Therefore, for all $n$ in the sum we have an upper bound of $16\kappa/T^2$ up to leading order.
The total upper bound for the sum of $\hat{c}_1(n/T)^2/(T^2\check{\Delta}(n/T)^3)$ from from $n=1$ to $T-1$ is therefore
\begin{equation}
\label{eq:uppersum1}
   \sum_{n=1}^{T-1} \hat{c}_1(n/T)^2/(T^2\check{\Delta}(n/T)^3) = \frac{16 \kappa}{T}+ \mathcal{O} \left( \frac{\kappa^{3/2}}{T^2} \right),
\end{equation}
which is given in \cref{eq:sumc1}.

Next we show the upper bound for the sum of the elements $\hat{c}_1(s)^2/(T^2\check{\Delta}(s)^2)$, which is given in \cref{eq:sumc1_2} above. When $1/T\leq s\leq 1-3/T$ we have
\begin{align}
   \frac{\hat{c}_1(n/T)^2}{T^2\check{\Delta}(n/T)^2}&= 4 \frac{(f(n/T)-f((n-1)/T))^2}{(1-f((n+3)/T)+f((n+3)/T)/\kappa)^2}\nonumber\\
   &= \frac{4\kappa^2}{(\sqrt{\kappa} + 1)^2}\frac{\left[(3+n)(\sqrt{\kappa}-1)+T\right]^4 \left[ \left(2n - 1\right)\left(\sqrt{\kappa} - 1\right) + 2T\right]^2 }{\left[n(\sqrt{\kappa}-1)+T\right]^4\left[(n-1)(\sqrt{\kappa}-1)+T\right]^4}\nonumber\\
&=  \frac{16\kappa^2}{[T+n(\sqrt\kappa-1)]^2(\sqrt{\kappa} + 1)^2} \frac{(1-\alpha_n/2)^2(1+3\alpha_n)^4}{(1-\alpha_n)^4}\nn   
&= \frac{16\kappa^2}{[T+n(\sqrt\kappa-1)]^2(\sqrt{\kappa} + 1)^2} \left[ 1+ \mathcal{O}\left(\alpha_n \right)\right]\nn
&\le \frac{16}{T^2}+\mathcal{O}\left(\frac {\sqrt\kappa}{T^3} \right).
\end{align}

Because the sum starts from $n=0$ we need the following upper bound
\begin{align}
   \frac{\hat{c}_1(0)^2}{T^2\check{\Delta}(0)^2}&=4\frac{f(1/T)^2}{(1-f(3/T)+f(3/T)/\kappa)^2}\nonumber\\
   & = \frac{4\kappa^2}{(\sqrt{\kappa} + 1)^2}\frac{\left(3(\sqrt{\kappa}-1) + T\right)^4\left(\sqrt{\kappa} - 1 + 2T\right)^2}{T^4\left(\sqrt{\kappa}-1+T\right)^4}\nonumber\\
   & = \frac{16\kappa^2}{(\sqrt\kappa+1)^2T^2} \frac{(1+a_0/2)^2(1+3a_0)^4}{(1+a_0)^4} \nn
   & = \frac{16\kappa^2}{(\sqrt\kappa+1)^2T^2} \left[ 1+\mathcal{O}(a_0)\right] \nn
   &= \frac{16\kappa}{T^2} + \mathcal{O} \left( \frac {\kappa^{3/2}}{T^3} \right).
 \end{align}

We also have to upper bound the cases where $s=1-1/T$ and $s=1-2/T$. This upper bound is the same as we had in~\cref{eq:upper-1/t,eq:upper-2/t}, but now with $1/\kappa^2$ in the denominator rather than  $1/\kappa^3$, so we get
\begin{equation}
   \frac{\hat{c}_1(1-2/T)^2}{T^2\check{\Delta}(1-2/T)^2} = 4\kappa^2(f(1-2/T) - f(1-3/T) )^2 = \frac{16}{T^2} + \mathcal{O} \left( \frac {1}{T^3} \right),
\end{equation}
and
\begin{equation}
    \frac{\hat{c}_1(1-1/T)^2}{T^2\check{\Delta}(1-1/T)^2} = 4\kappa^2(f(1-1/T) - f(1-2/T) )^2= \frac{16}{T^2} + \mathcal{O} \left( \frac {1}{T^3} \right).
\end{equation}
There are $T$ terms in the sum, and each is upper bounded by $16/T^2$ to leading order except that at $n=0$.
We therefore get
\begin{equation}
    \sum_{n=0}^{T-1}\frac{\hat{c}_1(n/T)^2}{T^2\check{\Delta}(n/T)^2}
    \leq \frac{16}{T} + \frac{16\kappa}{T^2} + \mathcal{O} \left( \frac {\sqrt\kappa}{T^2} \right)
    = \frac{16}{T} +\mathcal{O} \left( \frac {\kappa}{T^2} \right).
\end{equation}
This is the result given as \cref{eq:sumc1_2} above.

\subsection{\texorpdfstring{$c_2(s)$}{c2(s)} summation}\label{appsec:sumc2}

Next we show the upper bound given in \cref{eq:sumc2}.
Using \cref{eq:cic2corr,eqn:qlsp_gap_check} for $1\leq s\leq 1 - 3/T$ we have
\begin{align}
\frac{\hat{c}_2(n/T)}{T^2\check{\Delta}(n/T)^2}&= \frac{1}{ T^2} \frac{4 |f'(n/T)|^2 + |f''(n/T)|}{(1-f((n+3)/T)+f((n+3)/T)/\kappa)^2}\nonumber\\
&= \frac{2\kappa}{T^2(\sqrt\kappa+1)^2} (1+3a_n)^4((3+8\gamma^2)\kappa-3) \nn
&\leq \frac{22\kappa}{T^2(\sqrt\kappa+1)^2}(1+3a_n)^4\nn
  &= \frac{22\kappa}{T^2} + \mathcal{O}\left(\frac{\kappa^{3/2}}{T^3}\right),
\end{align}
where
\begin{equation}
    \gamma=\frac T{T+n(\sqrt\kappa-1)}<1.
\end{equation}
We need to separately consider the case $s=1-1/T$, which gives
\begin{align}
    \frac{\hat{c}_2(1-1/T)}{T^2\check{\Delta}(1-1/T)^2}&= \frac{\kappa^2}{ T^2} \left(4|f'(1-1/T)|^2 + |f''(1-1/T)|\right)\nonumber\\
    &= \frac{2T^2\kappa^3}{(\sqrt\kappa+1)^2(\sqrt\kappa(T-1)+1)^4} (3\kappa + 5 + 16\delta + 8\delta^2) \nn
    &= \frac {6\kappa}{T^2} + \mathcal{O}\left( \frac{1}{T^2}\right),
\end{align}   
where
\begin{equation}
    \delta = \frac{\sqrt\kappa-1}{\sqrt\kappa(T-1)+1} = \mathcal{O}(1/T).
\end{equation}
Similarly, we obtain
\begin{align}
    \frac{\hat{c}_2(1-2/T)}{T^2\check{\Delta}(1-2/T)^2}&= \frac{\kappa^2}{ T^2} \left(4|f'(1-2/T)|^2 + |f''(1-2/T)|\right)\nonumber\\
    &= \frac{2T^2\kappa^3}{(\sqrt\kappa+1)^2(\sqrt\kappa(T-2)+2)^4} (3\kappa + 5 + 32\delta + 32\delta^2) \nn
    &= \frac {6\kappa}{T^2} + \mathcal{O}\left( \frac{1}{T^2}\right),
\end{align}
where this time
\begin{equation}
    \delta = \frac{\sqrt\kappa-1}{\sqrt\kappa(T-2)+2} = \mathcal{O}(1/T).
\end{equation}
Finally, since there are $T-1$ terms in the sum, and each is upper bounded by $22\kappa/T^2$ to leading order, we get
\begin{equation}
\sum_{n=1}^{T-1} \frac{ \hat{c}_2(n/T)}{ T^2 \check{\Delta}(n/T)^2} \leq \frac {22\kappa}{T} + \mathcal{O}\left( \frac{\kappa^{3/2}}{T^2}\right).
\end{equation}
This is the bound given in \cref{eq:sumc2}.

\section{Phase factors in the adiabatic evolution}
\label{app:phasefactor}
One would normally consider the eigenspace of interest in a single group for the adiabatic theorem.
In contrast, here the eigenspace of interest is separated in two parts, corresponding to $\pm 1$, and the remaining eigenspace is separated by a gap in two parts in the upper and lower halves of the complex plane.
In order to address this, one can instead consider just the eigenvalue $1$ as the eigenspace of interest, which is then in a single group.
Then the discrete adiabatic theorem can be applied unchanged to show that the state is correctly mapped to the final eigenstate.
Similarly, one can just consider the adiabatic theorem with $-1$.
Since using the adiabatic theorem separately on each eigenstate shows that it properly evolves to the final state, the superposition of the two eigenstates must also.

To be more specific, as discussed in Eqs.\ (10) to (13) of \cite{BerryNPJ18}, the eigenvectors of the walk operator are of the form (correcting a missing $i$ in the reference)
\begin{equation}\label{eq:qubeig}
    \frac{1}{\sqrt 2} \left( \ket{0}_a \ket{k}_s \pm i\ket{0k^\perp}_{as} \right),
\end{equation}
where $\ket{0}_a$ is the zero state on the ancilla, $\ket{k}_s$ is the eigenstate of $H$ of energy $E_k$ on the system, and $\ket{0k^\perp}_{as}$ is a state orthogonal to $\ket{0}_a$ on the ancilla.
In our case, the target eigenvalue of the Hamiltonian is $E_k=0$, which yields eigenvalues of $\pm 1$ of the walk operator with these two eigenstates.
When we have a positive superposition of the two eigenstates, then the resulting state is the solution given by $\ket{0}_a \ket{k}_s$.
In contrast, if we have a negative superposition of the two eigenstates, then the result is the non-solution state $\ket{0k^\perp}_{as}$.
In the adiabatic evolution, we start with the positive superposition, and we must maintain that positive superposition at the end in order to obtain the solution.
Therefore, we should show that there is no phase factor introduced by the adiabatic evolution.

To show this, it is again sufficient to consider the evolution of each eigenvalue on its own.
To obtain the phase factor, it is sufficient to consider the exact adiabatic evolution given by the adiabatic walk operators
\begin{align}
    W_T^A(s) &= V_T(s) W_T(s) \nn
    &= v_T(s',s)^{-1} S_T(s',s) W_T(s) \nn
    &= [S_T(s',s) S_T^\dagger(s',s)]^{-1/2} S_T(s',s) W_T(s) \nn
    &= [P_T(s')P_T(s)P_T(s')+ Q_T(s')Q_T(s)Q_T(s')]^{-1/2} [P_T(s')P_T(s)+ Q_T(s')Q_T(s)] W_T(s).
\end{align}
where $s'=s+1/T$. (We are swapping the $s'$ and $s$ from the way $S_T$ and $v_T$ were given originally.)
For the case we are interested in, there may be multiple states within the spectrum of interest, but they are orthogonal.
More specifically, there is the solution state (ground state of the Hamiltonian)
\begin{equation}
    \begin{pmatrix}
     {A}(f(s))^{-1} \bm{b} \\ 0
    \end{pmatrix},
\end{equation}
as well as a non-solution state of the form
\begin{equation}
    \begin{pmatrix}
     0 \\ \bm{b}
    \end{pmatrix}.
\end{equation}
In that case the product of projectors is of the form
\begin{equation}
    P_T(s')P_T(s) = \sum_{j,j'} \ket{\lambda_j(s')}\braket{\lambda_j(s')}{\lambda_{j'}(s)}\bra{\lambda_{j'}(s)} = 
    \sum_{j} \braket{\lambda_j(s')}{\lambda_{j}(s)} \ket{\lambda_j(s')}\bra{\lambda_{j}(s)} ,
\end{equation}
and similarly
\begin{equation}
    P_T(s')P_T(s)P_T(s') = 
    \sum_{j} |\braket{\lambda_j(s')}{\lambda_{j}(s)}|^2 \ket{\lambda_j(s')}\bra{\lambda_{j}(s)} ,
\end{equation}
We use $\ket{\lambda_j(s)}$ to indicate eigenstates, including for degenerate eigenvalues.
What this means is that we cannot flip between orthogonal eigenstates in the spectrum of interest during the (exact) adiabatic evolution, and we do not have the solution state leaking into the non-solution state.

For the eigenstate $\ket{\lambda_j(s)}$ in the spectrum of interest for $W_T(s)$, we have
\begin{align}
    &[P_T(s')P_T(s)P_T(s')+ Q_T(s')Q_T(s)Q_T(s')]^{-1/2} [P_T(s')P_T(s)+ Q_T(s')Q_T(s)] W_T(s) \ket{\lambda_j(s)} \nn
    &=\lambda_j(s) [P_T(s')P_T(s)P_T(s')+ Q_T(s')Q_T(s)Q_T(s')]^{-1/2} [P_T(s')P_T(s)+ Q_T(s')Q_T(s)] \ket{\lambda_j(s)} \nn
    &=\lambda_j(s)\braket{\lambda_j(s')}{\lambda_{j}(s)} [P_T(s')P_T(s)P_T(s')+ Q_T(s')Q_T(s)Q_T(s')]^{-1/2} \ket{\lambda_j(s')} \nn
    &=\lambda_j(s)\braket{\lambda_j(s')}{\lambda_{j}(s)} [|\braket{\lambda_j(s')}{\lambda_{j}(s)}|^2\ket{\lambda_j(s')}\bra{\lambda_{j}(s)}]^{-1/2} \ket{\lambda_j(s')} \nn
    &= \lambda_j(s)\frac{\braket{\lambda_j(s')}{\lambda_{j}(s)}}{|\braket{\lambda_j(s')}{\lambda_{j}(s)}|}\ket{\lambda_j(s')}.
\end{align}
In the second line applying $W_T(s)$ gives the eigenvalue.
If this eigenstate is $\ket{\lambda_j(s)}$, then applying $P_T(s')P_T(s)+Q_T(s')Q_T(s)$ gives the updated state $\ket{\lambda_j(s')}$ times $\braket{\lambda_j(s')}{\lambda_{j}(s)}$.
Then applying $v_T(s',s)^{-1}$ cancels the magnitude of $\braket{\lambda_j(s')}{\lambda_{j}(s)}$, and we only have its phase.

For our application the eigenvalues of $W_T(s)$ are $\pm 1$, but provided the total number of steps of the walk is even then this sign flip cancels out.
Ideally we would show that $\braket{\lambda_j(s')}{\lambda_{j}(s)}$ is real in order to show that there are no spurious phase factors.
However, it is sufficient to just show that any phase factor from $\braket{\lambda_j(s')}{\lambda_{j}(s)}$ is the same between the $\pm 1$ eigenvectors of $W_T(s)$.
The eigenstates of $W_T(s)$ are given as in \cref{eq:qubeig}.
In order to describe the inner products between the eigenstates at successive time steps, let us use $k_1$ and $k_2$.
Then the inner product of eigenstates at successive steps is
\begin{equation}
    \frac 12 \left( {}_a\!\bra{0}\, {}_s\!\bra{k_1} \mp i \, {}_{as}\!\bra{0k_1^\perp}\right)\left( \ket{0}_a \ket{k_2}_s \pm i\ket{0k_2^\perp}_{as} \right) = 
    \frac 12 \left( \, {}_s\!\braket{k_1}{k_2}_s +\, {}_{as}\! \braket{0k_1^\perp}{0k_2^\perp}_{as} \right).
\end{equation}
Here we have used $\pm$ to indicate that we are using $+1$ for both steps or $-1$ for both steps.
The crucial result here is that the inner product does \emph{not} depend on whether we were considering the $+1$ or $-1$ eigenstates.
This means that there may be a phase factor, but it will be the \emph{same} between the $\pm 1$ eigenstates.
The one caveat is that we need to use an even number of steps to avoid a $-1$ factor, but it is always possible to slightly adjust the schedule so that there is an even number of steps, and that does not change the asymptotic scaling of the number of steps needed.

The net result of this is that the adiabatic walk with the qubitised walk operator still works, despite there being separated eigenvalues at $\pm 1$.
Here we have not needed to use any special properties of the Hamiltonian other than that the eigenvalue of interest (for the Hamiltonian) is $0$, so this result may be used for applications other than solving linear equations.
If there were a nonzero eigenvalue that was \emph{known}, then it would be possible to add a multiple of the identity to rezero that eigenvalue, and the above method would again work.

\section{Upper bounds of \texorpdfstring{\cref{thm:qlsp_general_p}}{Theorem 18} with \texorpdfstring{$1<p<2$}{1<p<2}}\label{app:qlsp_general_p}

In order to show the linear dependence in $T$ on $\kappa$, we use \cref{cor:adia} and calculate the scaling of each term. 
There are two main difficulties: estimates of finite differences of the walk operator and different discrete time points used in \cref{cor:adia}. 
To overcome the first difficulty, we establish a connection between the discrete finite difference coefficients $c_k(s)$ and the corresponding continuous derivatives of the schedule function. 
Then, according to the definition of the schedule function \cref{eq:gapCon}, this can be directly related to the spectrum gap and cancel with the denominators in the error bound. 
For the second difficulty, we use continuity and monotonicity of the spectrum gap in the linear system problem to unify the time points with sacrifice of larger preconstants.

We first reformulate the coefficients $c_1$ and $c_2$, which have already been done in \cref{lem:DR}. 
Here, we will use a slightly different version with continuous time values, that 
\begin{equation}
\label{eq:c1}
    c_1(s) = 2\max_{\tau\in[s,s+1/T]\cap[0,1]}|f'(\tau)|,
\end{equation}
and 
\begin{equation}
    c_2(s) = 2\max_{\tau\in[s,s+2/T]\cap[0,1]}(2|f'(\tau)|^2+|f''(\tau)|).
\end{equation}
Notice that the choices of $c_1$ and $c_2$ here are even larger than those in \cref{lem:DR}. 
Then we can use the definition of the schedule function to establish the connection between $c_k(s)$ and the spectrum gap. 
\begin{lemma}\label{lem:qlsp_c_bound}
    Consider solving linear system problem using discrete adiabatic evolution with schedule function defined in \cref{eq:gapCon}. 
    Then the walk operators satisfy 
    \begin{enumerate}
        \item for any $0 \leq s \leq 1-1/T$, we have 
    \begin{equation}
        c_1(s) = 2d_p\Delta_0(s)^{p}, 
    \end{equation}
    \item for any $0 \leq s \leq 1-2/T$, we have 
    \begin{equation}
        c_2(s) = 4 d_p^2\Delta_0(s)^{2p} + 2 d_p^2 p (1-1/\kappa) \Delta_0(s)^{2p-1}. 
    \end{equation}
    \end{enumerate}
\end{lemma}
\begin{proof}
    According to \cref{lem:DR},  we only need to compute the derivatives of the schedule function. 
    The first order derivative directly comes from the definition of the schedule function that 
    \begin{equation}
        f'(\tau) = d_p \Delta_0(s)^p. 
    \end{equation}
    For the second order derivative, we have 
    \begin{align}
        f''(\tau) &= \frac{d}{d\tau} \left(d_p \left(1-f(\tau)+f(\tau)/\kappa\right)^p\right)  \nonumber \\
        & = d_p p\left(1-f(\tau)+f(\tau)/\kappa\right)^{p-1} (-1+1/\kappa)f'(\tau) \nonumber \\
        &= d_p^2 p (-1+1/\kappa) \Delta_0(s)^{2p-1}. 
    \end{align}
    The proof is completed using the monotonicity of $\Delta_0$. 
\end{proof}

In the error estimate in \cref{cor:adia}, we encounter taking maximum or minimum of several consequent time steps, which poses technical difficulty in calculating the scaling of the error. 
In the following lemma, we show how to resolve the different time point issue. 
\begin{lemma}\label{lem:qlsp_time_unify}
    Let $\Delta_0(s)$ denote the spectrum gap of the time-dependent Hamiltonian used in solving linear system problem. 
    Assume $T \geq 16(\sqrt{2})^p \left(\frac{\kappa^{p-1}-1}{p-1}\right) = \mathcal{O}(\kappa^{p-1})$. 
    Then for any $s \leq s' \leq s+4/T$, we have 
    \begin{equation}
        \Delta_0(s) \leq \frac{4}{3}\Delta_0(s'). 
    \end{equation}
\end{lemma} 
\begin{proof}
    For simplicity we only consider the Hermitian positive definite case, since the gap in the general non-Hermitian case only differs from the positive definite case by a multiplication factor of $\sqrt{2}$. 
    We define $\Delta_{\text{linear}}(y) = 1-y+y/\kappa$. 
    Then $\Delta_0(s) = \Delta_{\text{linear}}(f(s))$. 
    Since $\Delta(s)$ is a monotonically decreasing function, it suffices to prove $\Delta_0(s)/\Delta_0(s+4/T) \leq 4/3$. 
    We first compute the derivative of the gap, 
    \begin{align}
        \Delta_0'(s) = \frac{d}{ds} \left(\Delta_{\text{linear}}(f(s))\right) = \Delta_{\text{linear}}'(f(s))f'(s) = (-1+1/\kappa) d_p \Delta_0(s)^p. 
    \end{align}
    Then for any $0\leq s \leq 1-4/T$, 
    \begin{equation}
        |\Delta_0(s) - \Delta_0(s+4/T)| \leq \frac{4}{T} \max_{s'\in[s,s+4/T]} |\Delta_0'(s')| = \frac{4}{T}(1-1/\kappa) d_p\Delta_0(s)^p, 
    \end{equation}
    and thus 
    \begin{align}
        \frac{\Delta_0(s)}{\Delta_0(s+4/T)} &= 1 + \frac{\Delta_0(s)-\Delta_0(s+4/T)}{\Delta_0(s+4/T)} \nonumber \\
        & \leq 1 + \frac{4(1-1/\kappa)d_p \Delta_0(s)^{p-1}}{T}\frac{\Delta_0(s)}{\Delta_0(s+4/T)} \nonumber \\
        & \leq 1 + \frac{4(1-1/\kappa)d_p }{T}\frac{\Delta_0(s)}{\Delta_0(s+4/T)}. 
    \end{align}
    It has been computed in~\cite{an2019quantum} that $d_p = \frac{2^{p/2}}{p-1}\frac{\kappa}{\kappa-1}(\kappa^{p-1}-1)$. 
    Together with the assumption that $T \geq 16(\sqrt{2})^p \left(\frac{\kappa^{p-1}-1}{p-1}\right)$, we have 
    \begin{equation}
        \frac{4(1-1/\kappa)d_p }{T} = \frac{2^{p/2+2}}{T(p-1)}(\kappa^{p-1}-1) \leq \frac{1}{4}, 
    \end{equation}
    and thus 
    \begin{equation}
        \frac{\Delta_0(s)}{\Delta_0(s+4/T)} \leq 1 + \frac{1}{4} \frac{\Delta_0(s)}{\Delta_0(s+4/T)}, 
    \end{equation}
    which implies $\Delta_0(s)/\Delta_0(s+4/T) \leq 4/3$. 
\end{proof}

Now we are ready to prove \cref{thm:qlsp_general_p}, the complexity estimate of using discrete adiabatic evolution to solve linear system problems. 

\begin{proof}[Proof of \cref{thm:qlsp_general_p}]
Without loss of generality, here we only prove the scenario with Hermitian positive definite matrix $A$.
This is because the spectrum gap in the general non-Hermitian case only differs by a multiplication factor of $\sqrt{2}$, which can be absorbed by changing the concrete definition of the constant factor $C_p$. 

First notice that 
\begin{equation}
4\frac{c_1(s)}{\Delta_1(s)} = \frac{8d_{p} \Delta_0(s)^{p}}{\Delta_0(s+1/T)} \leq \frac{32d_p\Delta_0(s)^{p}}{3\Delta_0(s)} \leq \frac{32d_p}{3}. 
\end{equation}
Therefore the assumption that $T \geq 32d_p/3$ ensures that the assumption in \cref{cor:adia} satisfies. 
Now we apply \cref{cor:adia} to bound the error. 
To simplify the computation, we now unify the time in the hat and check notations by applying \cref{lem:qlsp_time_unify}. 
More precisely, 
\begin{equation}
\hat{c}_1(s) = \max_{s'\in\{s-1/T,s,s+1/T\}\cap [0,1-1/T]} c_1(s') = \begin{cases}
2d_p \Delta_0(0)^{p}, & s = 0, \\
2d_p \Delta_0(s-1/T)^{p}, & 1/T \leq s \leq 1. 
\end{cases}
\end{equation}
Applying \cref{lem:qlsp_time_unify} to change all the discrete time to $s$, we have for all $0 \leq s \leq 1$, 
\begin{equation}\label{eqn:qlsp_c1_hat_proof}
\hat{c}_1(s) \leq \frac{2^{2p+1}}{3^p}d_p \Delta_0(s)^{p}. 
\end{equation}
Similarly, 
\begin{equation}
\hat{c}_2(s) = \begin{cases}
4 d_p^2\Delta_0(0)^{2p} + 2p d_p^2  (1-1/\kappa) \Delta_0(0)^{2p-1}, & s = 0\\ 
4 d_p^2\Delta_0(s-1/T)^{2p} + 2p d_p^2 (1-1/\kappa) \Delta_0(s-1/T)^{2p-1}, & 1/T \leq s \leq 1-1/T,
\end{cases}
\end{equation}
and for all $0 \leq s \leq 1-1/T$,
\begin{equation}\label{eqn:qlsp_c2_hat_proof}
\hat{c}_2(s) \leq \frac{2^{4p+2}}{3^{2p}}d_p^2\Delta_0(s)^{2p} + \frac{2^{4p-1}}{3^{2p-1}} p d_p^2 (1-1/\kappa) \Delta_0(s)^{2p-1}. 
\end{equation}
For the spectrum gap, by \cref{eqn:qlsp_gap_check} and \cref{lem:qlsp_time_unify}, we have 
\begin{equation}\label{eqn:qlsp_gap_hat_proof}
\check{\Delta}(s) \geq \frac{3}{4} \Delta_0(s). 
\end{equation}
    
Combining \cref{eqn:qlsp_c1_hat_proof,eqn:qlsp_c2_hat_proof,eqn:qlsp_gap_hat_proof} and the fact that 
\begin{equation}\label{eqn:qlsp_bound_dp}
d_p = \frac{2^{p/2}}{p-1}\frac{\kappa}{\kappa-1} (\kappa^{p-1}-1) \leq \frac{2^{1+p/2}}{p-1} \kappa^{p-1}, 
\end{equation}
we are now ready to bound each term in the error bound in \cref{cor:adia}. 
The first three terms (i.e. boundary terms) in \cref{cor:adia} can be bounded as follows: 
\begin{equation}\label{eqn:qlsp_term1}
\frac{\hat{c}_1(0)}{T\check{\Delta}(0)^2} \leq \frac{2^{2p+1}}{3^p}d_p \Delta_0(0)^p \frac{2^4}{3^2T\Delta_0(0)^2} = \frac{2^{2p+5}d_p}{3^{p+2}T} \leq \frac{2^{6+5p/2}}{3^{p+2}(p-1)} \frac{\kappa^{p-1}}{T}, \end{equation}
\begin{equation}\label{eqn:qlsp_term2}
\frac{\hat{c}_1(1)}{T\check{\Delta}(1)^2} \leq  \frac{2^{2p+1}}{3^p}d_p \Delta_0(1)^p \frac{2^4}{3^2T\Delta_0(1)^2} = \frac{2^{2p+5}d_p}{3^{p+2}T\Delta_0(1)^{2-p}} \leq \frac{2^{6+5p/2}}{3^{p+2}(p-1)} \frac{\kappa}{T}, 
\end{equation}
and
\begin{equation}\label{eqn:qlsp_term3}
\frac{\hat{c}_1(1)}{T\check{\Delta}(1)} \leq \frac{2^{2p+1}}{3^p}d_p \Delta_0(1)^p \frac{2^2}{3 T\Delta_0(1)} = \frac{2^{2p+3}d_p}{3^{p+1}T}\Delta_0(1)^{p-1} \leq \frac{2^{4+5p/2}}{3^{p+1}(p-1)}\frac{1}{T}. 
\end{equation}
Again by \cref{eqn:qlsp_c1_hat_proof,eqn:qlsp_c2_hat_proof,eqn:qlsp_gap_hat_proof}, the last three terms in \cref{cor:adia} can be bounded as 
\begin{align}
\sum_{n=1}^{T-1}\frac{\hat{c}_1(n/T)^2}{T^2\check{\Delta}(n/T)^3} \leq \sum_{n=1}^{T-1} \frac{2^{4p+2}}{3^{2p}}d_p^2\Delta_0(n/T)^{2p}\frac{1}{T^2}\frac{2^6}{3^3\Delta_0(n/T)^3} = \frac{2^{4p+8}d_p^2}{3^{2p+3}T^2}\sum_{n=1}^{T-1} \Delta_0(n/T)^{2p-3} \label{eqn:qlsp_bound_term4_1},  
\end{align}
\begin{align}
\sum_{n=0}^{T-1} \frac{ \hat{c}_1(n/T)^2}{ T^2 \check{\Delta}(n/T)^2}  \leq \sum_{n=0}^{T-1} \frac{2^{4p+2}}{3^{2p}}d_p^2\Delta_0(n/T)^{2p}\frac{1}{T^2}\frac{2^4}{3^2\Delta_0(n/T)^2}
= \frac{2^{4p+6}d_p^2}{3^{2p+2}T^2}\sum_{n=0}^{T-1} \Delta_0(n/T)^{2p-2} \label{eqn:qlsp_bound_term5_1}, 
\end{align}
and 
\begin{align}
\sum_{n=1}^{T-1}\frac{\hat{c}_2(n/T)}{T^2\check{\Delta}(n/T)^2} &\leq \sum_{n=1}^{T-1} \frac{2^{4p+2}}{3^{2p}}d_p^2\Delta_0(n/T)^{2p} \frac{1}{T^2}\frac{2^4}{3^2 \Delta_0(n/T)^2} \nonumber \\
& \quad + \sum_{n=1}^{T-1} \frac{2^{4p-1}}{3^{2p-1}}p d_p^2 (1-1/\kappa) \Delta_0(n/T)^{2p-1} \frac{1}{T^2} \frac{2^4}{3^2 \Delta_0(n/T)^2} \nonumber \\
& = \frac{2^{4p+6}d_p^2}{3^{2p+2}T^2}\sum_{n=1}^{T-1} \Delta_0(n/T)^{2p-2} + \frac{2^{4p+3}pd_p^2}{3^{2p+1}T^2}(1-1/\kappa)\sum_{n=1}^{T-1} \Delta_0(n/T)^{2p-3} \label{eqn:qlsp_bound_term6_1}.  
\end{align}
    
To proceed, we need to bound the summations $\frac{1}{T}\sum_{n=0}^{T-1}\Delta_0(n/T)^{2p-2}$ and $\frac{1}{T}\sum_{n=0}^{T-1}\Delta_0(n/T)^{2p-3}$. 
Notice that the summations are in the Riemann sum form. 
The idea is then to approximate the summations by corresponding integrals and to bound both the integrals and the difference terms. 
More precisely, according to~\cite{BurdenNA}, for any continuously differentiable $g(t)$ on the interval $[a,b]$, we have 
\begin{equation}
\left|\int_a^b g(t)dt - (b-a)g(a)\right| \leq \frac{(b-a)^2}{2} \max_{t\in[a,b]} \left|g'(t)\right|. 
\end{equation}
This implies that
\begin{equation}
\left|\int_0^1 g(t)dt - \frac{1}{T}\sum_{n=0}^{T-1}g(n/T)\right| \leq \frac{1}{2T^2} \sum_{n=0}^{T-1}\max_{t\in[n/T,(n+1)/T]} \left|g'(t)\right|. 
\end{equation}
If we further assume $g(t)>0$ for all $t$, then 
\begin{equation}\label{eqn:qlsp_bound_sum_integral}
\frac{1}{T}\sum_{n=0}^{T-1}g(n/T) \leq \int_0^1 g(t)dt + \frac{1}{2T^2} \sum_{n=0}^{T-1}\max_{t\in[n/T,(n+1)/T]} \left|g'(t)\right|. 
\end{equation}
By taking the function $g(t)$ to be $\Delta_0(t)^{2p-2}$ and $\Delta_0(t)^{2p-3}$ respectively, we can bound the desired summations. 
    
We start with the summation of $\Delta_0^{2p-2}$. By change of variable $x = f(t)$, the integral can be computed as 
\begin{align}
\int_0^1 \Delta_0(t)^{2p-2} dt &= \int_0^1 (1-f(t)+f(t)/\kappa)^{2p-2} dt \nonumber \\
& = \int_0^1 (1-f+f/\kappa)^{2p-2} \frac{1}{d_p (1-f+f/\kappa)^p} df \nonumber \\
& = \frac{1}{d_p} \int_0^1 (1-f+f/\kappa)^{p-2}df \nonumber \\
& = \frac{1}{d_p}\frac{\kappa^{2-p}}{p-1}\frac{\kappa^{p-1}-1}{\kappa-1}. 
\end{align}
The derivative can be computed as 
    \begin{equation}
        \frac{d}{dt}\Delta_0(t)^{2p-2} =  (2p-2)(-1+1/\kappa)d_p(1-f(t)+f(t)/\kappa)^{3p-3} = (2p-2)(-1+1/\kappa)d_p\Delta_0(t)^{3p-3}. 
    \end{equation}
    Therefore, according to \cref{eqn:qlsp_bound_sum_integral} and the fact that $\Delta_0(t)$ is bounded by $1$, we have 
    \begin{align}
        \frac{1}{T}\sum_{n=0}^{T-1}\Delta_0(n/T)^{2p-2} &\leq \frac{1}{d_p}\frac{\kappa^{2-p}}{p-1}\frac{\kappa^{p-1}-1}{\kappa-1} + \frac{(p-1)d_p}{T^2}\frac{\kappa-1}{\kappa} \sum_{n=0}^{T-1} \Delta_0(n/T)^{3p-3} \nonumber \\
        & \leq \frac{1}{d_p}\frac{\kappa^{2-p}}{p-1}\frac{\kappa^{p-1}-1}{\kappa-1} + \frac{(p-1)d_p}{T^2}\frac{\kappa-1}{\kappa} \sum_{n=0}^{T-1} \Delta_0(n/T)^{2p-2} \nonumber \\
        & \leq \frac{1}{d_p}\frac{\kappa^{2-p}}{p-1}\frac{\kappa^{p-1}-1}{\kappa-1} + \frac{d_p}{T} \frac{1}{T}\sum_{n=0}^{T-1} \Delta_0(n/T)^{2p-2}. 
    \end{align}
    By the assumption that $T > 32d_p/3$, we have $d_p/T \leq 3/32$ and thus 
    \begin{equation}
        \frac{1}{T}\sum_{n=0}^{T-1}\Delta_0(n/T)^{2p-2} \leq \frac{1}{d_p}\frac{\kappa^{2-p}}{p-1}\frac{\kappa^{p-1}-1}{\kappa-1} + \frac{3}{32} \frac{1}{T}\sum_{n=0}^{T-1} \Delta_0(n/T)^{2p-2}. 
    \end{equation}
    Solving the summation from the above inequality leads to 
    \begin{align}
        \frac{1}{T}\sum_{n=0}^{T-1}\Delta_0(n/T)^{2p-2} &\leq \frac{32}{29}\frac{1}{d_p}\frac{\kappa^{2-p}}{p-1}\frac{\kappa^{p-1}-1}{\kappa-1}  \nonumber \\
        & \leq \frac{2^5}{3^3}\frac{1}{d_p}\frac{\kappa^{2-p}}{p-1}\frac{\kappa^{p-1}-1}{\kappa-1} \nonumber \\
        & \leq  \frac{2^6}{3^3(p-1)}\frac{1}{d_p} \label{eqn:qlsp_bound_sum_2p-2},
    \end{align}
    where the last inequality follows from $\kappa^{2-p}\frac{\kappa^{p-1}-1}{\kappa-1} \leq 2$. 
    
    The summation of $\Delta_0^{2p-3}$ can be bounded similarly but requiring some more delicate computations. 
    We first assume that $p \neq 1.5$ such that $2p-3 \neq 0$. 
    Again, the integral and the derivative can be computed as 
    \begin{align}
        \int_0^1 \Delta_0(t)^{2p-3} dt &= \int_0^1 (1-f(t)+f(t)/\kappa)^{2p-3} dt \nonumber \\
        & = \int_0^1 (1-x+x/\kappa)^{2p-3} \frac{1}{d_p (1-x+x/\kappa)^p} dx \nonumber \\
        & = \frac{1}{d_p} \int_0^1 (1-x+x/\kappa)^{p-3}dx \nonumber \\
        & = \frac{1}{d_p} \frac{1}{2-p} \frac{\kappa}{\kappa-1} (\kappa^{2-p}-1),
    \end{align}
    and 
    \begin{equation}
        \frac{d}{dt}\Delta_0(t)^{2p-3} =  (2p-3)(-1+1/\kappa)d_p(1-f(t)+f(t)/\kappa)^{3p-4} = (2p-3)(-1+1/\kappa)d_p\Delta_0(t)^{3p-4}. 
    \end{equation}
    According to \cref{eqn:qlsp_bound_sum_integral}, we have 
    \begin{align}
        \frac{1}{T}\sum_{n=0}^{T-1}\Delta_0(n/T)^{2p-3} & \leq  \frac{1}{d_p} \frac{1}{2-p} \frac{\kappa}{\kappa-1} (\kappa^{2-p}-1) + \frac{|2p-3|d_p}{2T^2}\frac{\kappa-1}{\kappa} \sum_{n=0}^{T-1} \max_{t\in[n/t,(n+1)/T]} \Delta_0(t)^{3p-4} \nonumber \\
        & \leq \frac{1}{d_p} \frac{1}{2-p} \frac{\kappa}{\kappa-1} (\kappa^{2-p}-1) + \frac{|2p-3|d_p}{2T^2}\frac{\kappa-1}{\kappa} \sum_{n=0}^{T-1} \max_{t\in[n/t,(n+1)/T]} \Delta_0(t)^{2p-3}.
    \end{align}
    Since $\Delta_0(t)^{2p-3}$ is always monotonic, $\max_{t\in[n/T,(n+1)/T]} \Delta_0(t)^{2p-3}$ becomes either $\Delta_0(n/T)^{2p-3}$ or $\Delta_0((n+1)/T)^{2p-3}$. 
    The corresponding summation is then bounded by either $\sum_{n=0}^{T-1}\Delta_0(n/T)^{2p-3}$ or $\sum_{n=0}^{T-1}\Delta_0((n+1)/T)^{2p-3}$, of which both can be bounded by $\sum_{n=0}^{T}\Delta_0(n/T)^{2p-3}$. 
    Then 
    \begin{align}
        \frac{1}{T}\sum_{n=0}^{T-1}\Delta_0(n/T)^{2p-3} & \leq \frac{1}{d_p} \frac{1}{2-p} \frac{\kappa}{\kappa-1} (\kappa^{2-p}-1) + \frac{|2p-3|d_p}{2T^2}\frac{\kappa-1}{\kappa} \sum_{n=0}^{T}  \Delta_0(n/T)^{2p-3} \nonumber \\
        & \leq \frac{1}{d_p} \frac{1}{2-p} \frac{\kappa}{\kappa-1} (\kappa^{2-p}-1) + \frac{d_p}{2T^2} \sum_{n=0}^{T}  \Delta_0(n/T)^{2p-3}. 
    \end{align}
    Again using the fact that $d_p/T \leq \frac{3}{32}$ and separating the term with $n = T$ in the summation on the right hand side, we obtain 
    \begin{align}
        \frac{1}{T}\sum_{n=0}^{T-1}\Delta_0(n/T)^{2p-3} & \leq \frac{1}{d_p} \frac{1}{2-p} \frac{\kappa}{\kappa-1} (\kappa^{2-p}-1) + \frac{3}{32} \frac{1}{T} \sum_{n=0}^{T}  \Delta_0(n/T)^{2p-3} \nonumber \\
        & \leq \frac{1}{d_p} \frac{1}{2-p} \frac{\kappa}{\kappa-1} (\kappa^{2-p}-1) + \frac{3}{32} \frac{1}{T} \sum_{n=0}^{T-1}  \Delta_0(n/T)^{2p-3} + \frac{3}{32} \frac{\kappa^{3-2p}}{T}. 
    \end{align}
    Solving the summation gives 
    \begin{align}
        \frac{1}{T}\sum_{n=0}^{T-1}\Delta_0(n/T)^{2p-3} &\leq \frac{32}{29}\frac{1}{d_p} \frac{1}{2-p} \frac{\kappa}{\kappa-1} (\kappa^{2-p}-1) + \frac{3}{29} \frac{\kappa^{3-2p}}{T} \nonumber \\
        & \leq \frac{2^5}{3^3}\frac{1}{d_p} \frac{1}{2-p}\frac{\kappa}{\kappa-1} (\kappa^{2-p}-1) + \frac{1}{3^3} \frac{\kappa^{3-2p}}{T} \label{eqn:qlsp_bound_sum_2p-3}. 
    \end{align}
    Notice that the above estimate also holds for $p = 1.5$ since when $p=1.5$, the left hand side is a constant $1$ and the right hand hand is always larger than $1$. 
    
    Now we are ready to bound the last three terms in \cref{cor:adia}. 
    By plugging \cref{eqn:qlsp_bound_sum_2p-2,eqn:qlsp_bound_sum_2p-3} back into \cref{eqn:qlsp_bound_term4_1,eqn:qlsp_bound_term5_1,eqn:qlsp_bound_term6_1} and using the representation of $d_p$ in \cref{eqn:qlsp_bound_dp}, we have 
    \begin{align}
        \sum_{n=1}^{T-1}\frac{\hat{c}_1(n/T)^2}{T^2\check{\Delta}(n/T)^3} & \leq \frac{2^{4p+8}d_p}{3^{2p+3}T} \frac{2^6}{3^3}\frac{1}{2-p}(\kappa^{2-p}-1) + \frac{2^{4p+8}d_p^2}{3^{2p+3}T} \frac{1}{3^3}\frac{\kappa^{3-2p}}{T} \nonumber \\
        & = \frac{2^{4p+14}d_p}{3^{2p+6}(2-p)T} (\kappa^{2-p}-1) + \frac{2^{4p+8}d_p^2}{3^{2p+6}T} \frac{\kappa^{3-2p}}{T} \nonumber \\
        & \leq \frac{2^{16+9p/2}}{3^{2p+6}(2-p)(p-1)} \frac{\kappa}{T} +  \frac{2^{5p+10}}{3^{2p+6}(p-1)^2} \frac{\kappa}{T^2} , \label{eqn:qlsp_term4}
    \end{align}
    \begin{align}
        \sum_{n=0}^{T-1} \frac{ \hat{c}_1(n/T)^2}{ T^2 \check{\Delta}(n/T)^2}  \leq \frac{2^{4p+6}d_p^2}{3^{2p+2}T}\frac{2^6}{3^3(p-1)}\frac{1}{d_p} 
        \leq \frac{2^{13+9p/2}}{3^{2p+5}(p-1)^2 }\frac{\kappa^{p-1}}{T} ,   \label{eqn:qlsp_term5}
    \end{align}
    and 
    \begin{align}
        \sum_{n=1}^{T-1}\frac{\hat{c}_2(n/T)}{T^2\check{\Delta}(n/T)^2} &\leq  \frac{2^{4p+6}d_p^2}{3^{2p+2}T} \frac{2^6}{3^3(p-1)}\frac{1}{d_p} + \frac{2^{4p+3}pd_p^2}{3^{2p+1}T}\frac{\kappa-1}{\kappa}\left(\frac{2^5}{3^3}\frac{1}{d_p} \frac{1}{2-p}\frac{\kappa}{\kappa-1} (\kappa^{2-p}-1) + \frac{1}{3^3} \frac{\kappa^{3-2p}}{T}\right) \nonumber \\
        & = \frac{2^{4p+12}d_p}{3^{2p+5}(p-1)T} + \frac{2^{4p+8}pd_p(\kappa^{2-p}-1)}{3^{2p+4}(2-p)T} + \frac{2^{4p+3}pd_p^2}{3^{2p+4}T}\frac{\kappa-1}{\kappa} \frac{\kappa^{3-2p}}{T} \nonumber \\
        & \leq \frac{2^{13+9p/2}}{3^{2p+5}(p-1)^2}\frac{\kappa^{p-1}}{T} +  \frac{2^{9+9p/2}p}{3^{2p+4}(2-p)(p-1)}\frac{\kappa}{T} + \frac{2^{5p+4}p}{3^{2p+4}(p-1)^2 } \frac{\kappa}{T^2} .  \label{eqn:qlsp_term6}
    \end{align}
    
    Finally, plugging \cref{eqn:qlsp_term1,eqn:qlsp_term2,eqn:qlsp_term3,eqn:qlsp_term4,eqn:qlsp_term5,eqn:qlsp_term6} into the estimate in \cref{cor:adia}, the adiabatic error can be bounded by 
    \begin{align}
        & \quad 12 \frac{2^{6+5p/2}}{3^{p+2}(p-1)} \frac{\kappa^{p-1}}{T} + 12 \frac{2^{6+5p/2}}{3^{p+2}(p-1)} \frac{\kappa}{T} + 6 \frac{2^{4+5p/2}}{3^{p+1}(p-1)}\frac{1}{T} \nonumber \\
        & \quad\quad + 305 \frac{2^{16+9p/2}}{3^{2p+6}(2-p)(p-1)} \frac{\kappa}{T} + 305 \frac{2^{5p+10}}{3^{2p+6}(p-1)^2} \frac{\kappa}{T^2} + 44\frac{2^{13+9p/2}}{3^{2p+5}(p-1)^2 }\frac{\kappa^{p-1}}{T} \nonumber \\
        & \quad \quad + 32\frac{2^{13+9p/2}}{3^{2p+5}(p-1)^2}\frac{\kappa^{p-1}}{T} +  32\frac{2^{9+9p/2}p}{3^{2p+4}(2-p)(p-1)}\frac{\kappa}{T} + 32\frac{2^{5p+4}p}{3^{2p+4}(p-1)^2 } \frac{\kappa}{T^2} \nonumber \\
        & \leq  C_p^{(1)}\frac{\kappa}{T} + C_p^{(2)}\frac{\kappa^{p-1}}{T} +C_p^{(3)}\frac{\kappa}{T^2} + C_p^{(4)}\frac{1}{T}, \label{eqn:qlsp_general_p_error}
    \end{align}
    where
    \begin{align}
        C_p^{(1)} &:= 12 \frac{2^{6+5p/2}}{3^{p+2}(p-1)} + 305 \frac{2^{16+9p/2}}{3^{2p+6}(2-p)(p-1)} + 32\frac{2^{9+9p/2}p}{3^{2p+4}(2-p)(p-1)},   \\
        C_p^{(2)} &:= 12 \frac{2^{6+5p/2}}{3^{p+2}(p-1)} +44\frac{2^{13+9p/2}}{3^{2p+5}(p-1)^2 } + 32\frac{2^{13+9p/2}}{3^{2p+5}(p-1)^2} , \\
        C_p^{(5)} &:= 305 \frac{2^{5p+10}}{3^{2p+6}(p-1)^2}+32\frac{2^{5p+4}p}{3^{2p+4}(p-1)^2 } , \\
        C_p^{(4)} &:= 6 \frac{2^{4+5p/2}}{3^{p+1}(p-1)} .
    \end{align}
    This completes the proof of the first part by defining $C_p$ to be the largest constant factor in \cref{eqn:qlsp_general_p_error},
    \begin{equation}\label{eq:Cpdef}
        C_p := \max_j C_p^{(j)}.
    \end{equation}
    The second part of \cref{thm:qlsp_general_p}, which is $T = \mathcal{O}\left({\kappa}/{\epsilon}\right)$,   
    directly follows from this bound by noticing that each term of the adiabatic error in the first part can be bounded by $\mathcal{O}(\kappa/T)$.    
\end{proof}

\section{Additional details for filtering}
\label{ap:filtering}

Here we give a proof of the upper bound on the norm of the difference of states for filtering.
We are assuming that $\tilde w(\phi)=0$ for the desired part of the spectrum, and the initial probability for the desired part of the spectrum is at least $1/2$.
Then the squared norm for the undesired part of the state is
\begin{equation}
    P(\perp) = \left\| 
    \sum_{k\in \perp} \tilde w(\phi_k) \psi_k \ket{k}
    \right\|^2 \le \left(\max_{k\in\{\perp\}} \tilde w(\phi_k) \right)^2 \left( \sum_{k\in \perp} |\psi_k|^2 \right) ,
\end{equation}
where we are using $\perp$ to denote the set of undesired states.
Recall that this was part of a state that is not normalised.
The squared norm for the desired part of the spectrum is
\begin{equation}
    P(\not\perp) = \sum_{k\in\not\perp} |\psi_k|^2,
\end{equation}
where $\not\perp$ is indicating the desired part of the spectrum.
As a result, the normalised probability for the desired part is lower bounded by
\begin{equation}
\frac{P(\not\perp)}{P(\not\perp) + P(\perp)} \ge
\frac{\sum_{k\in\not\perp} |\psi_k|^2}{\sum_{k\in\not\perp} |\psi_k|^2 + \left(\max_{k\in\perp} \tilde w(\phi_k) \right)^2 \left( \sum_{k\in \perp} |\psi_k|^2 \right)} \ge \frac{1}{1 + \left(\max_{k\in\perp} \tilde w(\phi_k) \right)^2 } ,
\end{equation}
where the second inequality comes from assuming that the initial probability for the desired part of the spectrum is at least $1/2$.
Given this probability, the norm of the difference from the desired state is
\begin{equation}
    \sqrt{2-\frac 2{1 + \left(\max_{k\in\perp} \tilde w(\phi_k) \right)^2}} \le \max_{k\in\perp} \tilde w(\phi_k).
\end{equation}

Next we give a more explicit description of the sequence of rotations needed for the filtering.
The first rotation prepares the state
\begin{equation}
    \frac 1{\sqrt{\sum_j w_j}} \left( \sqrt{w_0} \ket{0} + \sqrt{\sum_{j>0} w_j} \ket{1} \right) .
\end{equation}
The first controlled rotation gives
\begin{equation}
    {\sqrt{\sum_{j>0} w_j}} \ket{10}\mapsto \sqrt{w_1} \ket{10} + {\sqrt{\sum_{j>1} w_j}} \ket{11} .
\end{equation}
In general, the controlled rotation with qubit $k$ as control and $k+1$ as target maps
\begin{equation}
    {\sqrt{\sum_{j\ge k} w_j}} \ket{10}\mapsto \sqrt{w_k} \ket{10} + {\sqrt{\sum_{j>k} w_j}} \ket{11} .
\end{equation}
If one were to perform the rotations for the preparation in the reverse order, one would use a rotation on the last qubit to take zero to
\begin{equation}
    \frac 1{\sqrt{\sum_j w_j}} \left( \sqrt{\sum_{j=0}^{\ell-1} w_j} \ket{0} + \sqrt{w_{\ell}} \ket{1} \right) .
\end{equation}
Then the controlled rotation would take
\begin{equation}
    \sqrt{\sum_{j=0}^{\ell-1} w_j} \ket{10} \mapsto \sqrt{\sum_{j=0}^{\ell-2} w_j} \ket{00} + \sqrt{w_{\ell-1}} \ket{10} .
\end{equation}
Inverting this rotation gives
\begin{align}
    \sqrt{\sum_{j=0}^{\ell-1} w_j} \ket{00} &\mapsto \sqrt{\sum_{j=0}^{\ell-2} w_j} \ket{10} - \sqrt{w_{\ell-1}} \ket{00}, \\
    \sqrt{\sum_{j=0}^{\ell-1} w_j} \ket{10} &\mapsto \sqrt{w_{\ell-1}} \ket{10} + \sqrt{\sum_{j=0}^{\ell-2} w_j} \ket{00}.
\end{align}
More generally, the rotation with qubit $k+1$ as control and $k$ as target gives
\begin{align}
    \sqrt{\sum_{j=0}^{k} w_j} \ket{00} &\mapsto \sqrt{\sum_{j=0}^{k-1} w_j} \ket{10} - \sqrt{w_{k}} \ket{00}, \\
    \sqrt{\sum_{j=0}^{k} w_j} \ket{10} &\mapsto \sqrt{w_{k}} \ket{10} + \sqrt{\sum_{j=0}^{k-1} w_j} \ket{00} .
\end{align}

This means that the sequence of two controlled rotations gives
\begin{align}
    {\sqrt{\sum_{j\ge k} w_j}} \ket{1} &\mapsto \sqrt{w_k} \ket{10} + {\sqrt{\sum_{j>k} w_j}} \ket{11}  \nn
    & \mapsto \frac{ \sqrt{w_k} }{\sqrt{\sum_{j=0}^{k} w_j}} \left( \sqrt{w_{k}} \ket{10} + \sqrt{\sum_{j=0}^{k-1} w_j} \ket{00} \right) + {\sqrt{\sum_{j>k} w_j}} \ket{11}
\end{align}
and
\begin{align}
     \ket{0} &\mapsto \frac 1{\sqrt{\sum_{j=0}^{k} w_j}} \left(\sqrt{\sum_{j=0}^{k-1} w_j} \ket{10} - \sqrt{w_{k}} \ket{00} \right).
\end{align}
Projecting onto one on the first qubit then gives the mapping
\begin{align}
     \ket{1} &\mapsto \frac 1{\sqrt{\sum_{j\ge k} w_j}} \left( \frac{ w_k }{\sqrt{\sum_{j=0}^{k} w_j}} \ket{0}  + {\sqrt{\sum_{j>k} w_j}} \ket{1} \right) \\
     \ket{0} &\mapsto \frac {\sqrt{\sum_{j=0}^{k-1} w_j}}{\sqrt{\sum_{j=0}^{k} w_j}} \ket{0}.
\end{align}
To see the effect of this, let us consider $k=1$, so we are considering the operation immediately after the qubit rotation and controlled $W$ on the target system.
Assuming the target system is in an eigenstate with eigenvalue $e^{i\phi}$, the state at this point will be
\begin{equation}
    \frac 1{\sqrt{\sum_j w_j}} \left( \sqrt{w_0} \ket{0} + \sqrt{\sum_{j>0} w_j} e^{i\phi} \ket{1} \right) .
\end{equation}
The above mapping then gives
\begin{equation}
    \frac 1{\sqrt{\sum_j w_j}} \left( \frac{w_0}{\sqrt{\sum_{j=0}^1 w_j}} \ket{0} + e^{i\phi} \left( \frac{w_1}{\sqrt{\sum_{j=0}^1 w_j}} \ket{0} + \sqrt{\sum_{j>1} w_j}\ket{1} \right) \right) .
\end{equation}
This can be written as
\begin{equation}
    \frac 1{\sqrt{\sum_j w_j}} \left( \frac{w_0+e^{i\phi} w_1}{\sqrt{\sum_{j=0}^1 w_j}} \ket{0} + e^{i\phi} \sqrt{\sum_{j>1} w_j}\ket{1} \right) .
\end{equation}
Thus we can see that we have the desired weights $w_0$ and $w_1$ on the $\ket{0}$ state, and the $\ket{1}$ state is flagging the remainder of the linear combination still to be obtained.
More generally, after performing the controlled rotations between qubits $k$ and $k+1$ and the projection onto $\ket{1}$ on the ancilla qubit, the state will be of the form
\begin{equation}
    \frac 1{\sqrt{\sum_j w_j}} \left( \frac{\sum_{j=0}^k e^{ij\phi} w_j}{\sqrt{\sum_{j=0}^k w_j}} \ket{0} + e^{ik\phi} \sqrt{\sum_{j>k} w_j}\ket{1} \right) .
\end{equation}


\begin{thebibliography}{26}%
\makeatletter
\providecommand \@ifxundefined [1]{%
 \@ifx{#1\undefined}
}%
\providecommand \@ifnum [1]{%
 \ifnum #1\expandafter \@firstoftwo
 \else \expandafter \@secondoftwo
 \fi
}%
\providecommand \@ifx [1]{%
 \ifx #1\expandafter \@firstoftwo
 \else \expandafter \@secondoftwo
 \fi
}%
\providecommand \natexlab [1]{#1}%
\providecommand \enquote  [1]{``#1''}%
\providecommand \bibnamefont  [1]{#1}%
\providecommand \bibfnamefont [1]{#1}%
\providecommand \citenamefont [1]{#1}%
\providecommand \href@noop [0]{\@secondoftwo}%
\providecommand \href [0]{\begingroup \@sanitize@url \@href}%
\providecommand \@href[1]{\@@startlink{#1}\@@href}%
\providecommand \@@href[1]{\endgroup#1\@@endlink}%
\providecommand \@sanitize@url [0]{\catcode `\\12\catcode `\$12\catcode
  `\&12\catcode `\#12\catcode `\^12\catcode `\_12\catcode `\%12\relax}%
\providecommand \@@startlink[1]{}%
\providecommand \@@endlink[0]{}%
\providecommand \url  [0]{\begingroup\@sanitize@url \@url }%
\providecommand \@url [1]{\endgroup\@href {#1}{\urlprefix }}%
\providecommand \urlprefix  [0]{URL }%
\providecommand \Eprint [0]{\href }%
\providecommand \doibase [0]{http://dx.doi.org/}%
\providecommand \selectlanguage [0]{\@gobble}%
\providecommand \bibinfo  [0]{\@secondoftwo}%
\providecommand \bibfield  [0]{\@secondoftwo}%
\providecommand \translation [1]{[#1]}%
\providecommand \BibitemOpen [0]{}%
\providecommand \bibitemStop [0]{}%
\providecommand \bibitemNoStop [0]{.\EOS\space}%
\providecommand \EOS [0]{\spacefactor3000\relax}%
\providecommand \BibitemShut  [1]{\csname bibitem#1\endcsname}%
\let\auto@bib@innerbib\@empty
\bibitem [{\citenamefont {Harrow}\ \emph {et~al.}(2009)\citenamefont {Harrow},
  \citenamefont {Hassidim},\ and\ \citenamefont {Lloyd}}]{Harrow_2009}%
  \BibitemOpen
  \bibfield  {author} {\bibinfo {author} {\bibfnamefont {A.~W.}\ \bibnamefont
  {Harrow}}, \bibinfo {author} {\bibfnamefont {A.}~\bibnamefont {Hassidim}}, \
  and\ \bibinfo {author} {\bibfnamefont {S.}~\bibnamefont {Lloyd}},\ }\href
  {\doibase 10.1103/physrevlett.103.150502} {\bibfield  {journal} {\bibinfo
  {journal} {Physical Review Letters}\ }\textbf {\bibinfo {volume} {103}},\
  \bibinfo {pages} {150502} (\bibinfo {year} {2009})}\BibitemShut {NoStop}%
\bibitem [{\citenamefont {Clader}\ \emph {et~al.}(2013)\citenamefont {Clader},
  \citenamefont {Jacobs},\ and\ \citenamefont {Sprouse}}]{Clader13}%
  \BibitemOpen
  \bibfield  {author} {\bibinfo {author} {\bibfnamefont {B.~D.}\ \bibnamefont
  {Clader}}, \bibinfo {author} {\bibfnamefont {B.~C.}\ \bibnamefont {Jacobs}},
  \ and\ \bibinfo {author} {\bibfnamefont {C.~R.}\ \bibnamefont {Sprouse}},\
  }\href {\doibase 10.1103/PhysRevLett.110.250504} {\bibfield  {journal}
  {\bibinfo  {journal} {Physical Review Letters}\ }\textbf {\bibinfo {volume}
  {110}},\ \bibinfo {pages} {250504} (\bibinfo {year} {2013})}\BibitemShut
  {NoStop}%
\bibitem [{\citenamefont {Berry}(2014)}]{BerryJPA14}%
  \BibitemOpen
  \bibfield  {author} {\bibinfo {author} {\bibfnamefont {D.~W.}\ \bibnamefont
  {Berry}},\ }\href {\doibase 10.1088/1751-8113/47/10/105301} {\bibfield
  {journal} {\bibinfo  {journal} {Journal of Physics A: Mathematical and
  Theoretical}\ }\textbf {\bibinfo {volume} {47}},\ \bibinfo {pages} {105301}
  (\bibinfo {year} {2014})}\BibitemShut {NoStop}%
\bibitem [{\citenamefont {Berry}\ \emph {et~al.}(2017)\citenamefont {Berry},
  \citenamefont {Childs}, \citenamefont {Ostrander},\ and\ \citenamefont
  {Wang}}]{BerryCMP17}%
  \BibitemOpen
  \bibfield  {author} {\bibinfo {author} {\bibfnamefont {D.~W.}\ \bibnamefont
  {Berry}}, \bibinfo {author} {\bibfnamefont {A.~M.}\ \bibnamefont {Childs}},
  \bibinfo {author} {\bibfnamefont {A.}~\bibnamefont {Ostrander}}, \ and\
  \bibinfo {author} {\bibfnamefont {G.}~\bibnamefont {Wang}},\ }\href {\doibase
  10.1007/s00220-017-3002-y} {\bibfield  {journal} {\bibinfo  {journal}
  {Communications in Mathematical Physics}\ }\textbf {\bibinfo {volume}
  {356}},\ \bibinfo {pages} {1057} (\bibinfo {year} {2017})}\BibitemShut
  {NoStop}%
\bibitem [{\citenamefont {Wiebe}\ \emph {et~al.}(2012)\citenamefont {Wiebe},
  \citenamefont {Braun},\ and\ \citenamefont {Lloyd}}]{Wiebe2012}%
  \BibitemOpen
  \bibfield  {author} {\bibinfo {author} {\bibfnamefont {N.}~\bibnamefont
  {Wiebe}}, \bibinfo {author} {\bibfnamefont {D.}~\bibnamefont {Braun}}, \ and\
  \bibinfo {author} {\bibfnamefont {S.}~\bibnamefont {Lloyd}},\ }\href
  {\doibase 10.1103/PhysRevLett.109.050505} {\bibfield  {journal} {\bibinfo
  {journal} {Physical Review Letters}\ }\textbf {\bibinfo {volume} {109}},\
  \bibinfo {pages} {050505} (\bibinfo {year} {2012})}\BibitemShut {NoStop}%
\bibitem [{\citenamefont {Rebentrost}\ \emph {et~al.}(2014)\citenamefont
  {Rebentrost}, \citenamefont {Mohseni},\ and\ \citenamefont {Lloyd}}]{supvec}%
  \BibitemOpen
  \bibfield  {author} {\bibinfo {author} {\bibfnamefont {P.}~\bibnamefont
  {Rebentrost}}, \bibinfo {author} {\bibfnamefont {M.}~\bibnamefont {Mohseni}},
  \ and\ \bibinfo {author} {\bibfnamefont {S.}~\bibnamefont {Lloyd}},\ }\href
  {\doibase 10.1103/PhysRevLett.113.130503} {\bibfield  {journal} {\bibinfo
  {journal} {Physical Review Letters}\ }\textbf {\bibinfo {volume} {113}},\
  \bibinfo {pages} {130503} (\bibinfo {year} {2014})}\BibitemShut {NoStop}%
\bibitem [{\citenamefont {Lloyd}\ \emph {et~al.}(2013)\citenamefont {Lloyd},
  \citenamefont {Mohseni},\ and\ \citenamefont
  {Rebentrost}}]{lloyd2013quantum}%
  \BibitemOpen
  \bibfield  {author} {\bibinfo {author} {\bibfnamefont {S.}~\bibnamefont
  {Lloyd}}, \bibinfo {author} {\bibfnamefont {M.}~\bibnamefont {Mohseni}}, \
  and\ \bibinfo {author} {\bibfnamefont {P.}~\bibnamefont {Rebentrost}},\
  }\href {https://arxiv.org/abs/1307.0411} {\bibfield  {journal} {\bibinfo
  {journal} {arXiv:1307.0411}\ } (\bibinfo {year} {2013})}\BibitemShut
  {NoStop}%
\bibitem [{\citenamefont {Montanaro}\ and\ \citenamefont
  {Pallister}(2016)}]{Ashley2016}%
  \BibitemOpen
  \bibfield  {author} {\bibinfo {author} {\bibfnamefont {A.}~\bibnamefont
  {Montanaro}}\ and\ \bibinfo {author} {\bibfnamefont {S.}~\bibnamefont
  {Pallister}},\ }\href {\doibase 10.1103/PhysRevA.93.032324} {\bibfield
  {journal} {\bibinfo  {journal} {Physical Review A}\ }\textbf {\bibinfo
  {volume} {93}},\ \bibinfo {pages} {032324} (\bibinfo {year}
  {2016})}\BibitemShut {NoStop}%
\bibitem [{\citenamefont {Low}(2019)}]{rootd}%
  \BibitemOpen
  \bibfield  {author} {\bibinfo {author} {\bibfnamefont {G.~H.}\ \bibnamefont
  {Low}},\ }in\ \href {\doibase 10.1145/3313276.3316386} {\emph {\bibinfo
  {booktitle} {Proceedings of the 51st Annual ACM SIGACT Symposium on Theory of
  Computing}}},\ \bibinfo {series and number} {STOC 2019}\ (\bibinfo
  {publisher} {Association for Computing Machinery},\ \bibinfo {address} {New
  York, NY, USA},\ \bibinfo {year} {2019})\ p.\ \bibinfo {pages}
  {491–502}\BibitemShut {NoStop}%
\bibitem [{\citenamefont {Ambainis}(2010)}]{ambainis2010variable}%
  \BibitemOpen
  \bibfield  {author} {\bibinfo {author} {\bibfnamefont {A.}~\bibnamefont
  {Ambainis}},\ }\href {https://arxiv.org/abs/1010.4458} {\bibfield  {journal}
  {\bibinfo  {journal} {arXiv:1010.4458}\ } (\bibinfo {year}
  {2010})}\BibitemShut {NoStop}%
\bibitem [{\citenamefont {Childs}\ \emph {et~al.}(2017)\citenamefont {Childs},
  \citenamefont {Kothari},\ and\ \citenamefont {Somma}}]{CKS}%
  \BibitemOpen
  \bibfield  {author} {\bibinfo {author} {\bibfnamefont {A.~M.}\ \bibnamefont
  {Childs}}, \bibinfo {author} {\bibfnamefont {R.}~\bibnamefont {Kothari}}, \
  and\ \bibinfo {author} {\bibfnamefont {R.~D.}\ \bibnamefont {Somma}},\ }\href
  {\doibase 10.1137/16M1087072} {\bibfield  {journal} {\bibinfo  {journal}
  {SIAM Journal on Computing}\ }\textbf {\bibinfo {volume} {46}},\ \bibinfo
  {pages} {1920} (\bibinfo {year} {2017})}\BibitemShut {NoStop}%
\bibitem [{\citenamefont {Farhi}\ \emph {et~al.}(2000)\citenamefont {Farhi},
  \citenamefont {Goldstone}, \citenamefont {Gutmann},\ and\ \citenamefont
  {Sipser}}]{farhi2000quantum}%
  \BibitemOpen
  \bibfield  {author} {\bibinfo {author} {\bibfnamefont {E.}~\bibnamefont
  {Farhi}}, \bibinfo {author} {\bibfnamefont {J.}~\bibnamefont {Goldstone}},
  \bibinfo {author} {\bibfnamefont {S.}~\bibnamefont {Gutmann}}, \ and\
  \bibinfo {author} {\bibfnamefont {M.}~\bibnamefont {Sipser}},\ }\href
  {https://arxiv.org/abs/quant-ph/0001106} {\bibfield  {journal} {\bibinfo
  {journal} {arXiv:quant-ph/0001106}\ } (\bibinfo {year} {2000})}\BibitemShut
  {NoStop}%
\bibitem [{\citenamefont {Aharonov}\ \emph {et~al.}(2007)\citenamefont
  {Aharonov}, \citenamefont {Kempe}, \citenamefont {Lloyd}, \citenamefont
  {Dam}, \citenamefont {Landau}, \citenamefont {Regev}, \citenamefont
  {Van~Dam}, \citenamefont {Kempe}, \citenamefont {Landau}, \citenamefont
  {Lloyd},\ and\ \citenamefont {Regev}}]{Aharonov2007}%
  \BibitemOpen
  \bibfield  {author} {\bibinfo {author} {\bibfnamefont {D.}~\bibnamefont
  {Aharonov}}, \bibinfo {author} {\bibfnamefont {J.}~\bibnamefont {Kempe}},
  \bibinfo {author} {\bibfnamefont {S.}~\bibnamefont {Lloyd}}, \bibinfo
  {author} {\bibfnamefont {W.~V.}\ \bibnamefont {Dam}}, \bibinfo {author}
  {\bibfnamefont {Z.}~\bibnamefont {Landau}}, \bibinfo {author} {\bibfnamefont
  {O.}~\bibnamefont {Regev}}, \bibinfo {author} {\bibfnamefont
  {W.}~\bibnamefont {Van~Dam}}, \bibinfo {author} {\bibfnamefont
  {J.}~\bibnamefont {Kempe}}, \bibinfo {author} {\bibfnamefont
  {Z.}~\bibnamefont {Landau}}, \bibinfo {author} {\bibfnamefont
  {S.}~\bibnamefont {Lloyd}}, \ and\ \bibinfo {author} {\bibfnamefont
  {O.}~\bibnamefont {Regev}},\ }\href {\doibase 10.1137/S0097539705447323}
  {\bibfield  {journal} {\bibinfo  {journal} {SIAM Journal on Computing}\
  }\textbf {\bibinfo {volume} {37}},\ \bibinfo {pages} {166} (\bibinfo {year}
  {2007})}\BibitemShut {NoStop}%
\bibitem [{\citenamefont {Suba{\c{s}}{\i}}\ \emph {et~al.}(2019)\citenamefont
  {Suba{\c{s}}{\i}}, \citenamefont {Somma},\ and\ \citenamefont
  {Orsucci}}]{PhysRevLett.122.060504}%
  \BibitemOpen
  \bibfield  {author} {\bibinfo {author} {\bibfnamefont {Y.}~\bibnamefont
  {Suba{\c{s}}{\i}}}, \bibinfo {author} {\bibfnamefont {R.~D.}\ \bibnamefont
  {Somma}}, \ and\ \bibinfo {author} {\bibfnamefont {D.}~\bibnamefont
  {Orsucci}},\ }\href {\doibase 10.1103/PhysRevLett.122.060504} {\bibfield
  {journal} {\bibinfo  {journal} {Physical Review Letters}\ }\textbf {\bibinfo
  {volume} {122}},\ \bibinfo {pages} {060504} (\bibinfo {year}
  {2019})}\BibitemShut {NoStop}%
\bibitem [{\citenamefont {An}\ and\ \citenamefont {Lin}(2019)}]{an2019quantum}%
  \BibitemOpen
  \bibfield  {author} {\bibinfo {author} {\bibfnamefont {D.}~\bibnamefont
  {An}}\ and\ \bibinfo {author} {\bibfnamefont {L.}~\bibnamefont {Lin}},\
  }\href {https://arxiv.org/abs/1909.05500} {\bibfield  {journal} {\bibinfo
  {journal} {arXiv:1909.05500}\ } (\bibinfo {year} {2019})}\BibitemShut
  {NoStop}%
\bibitem [{\citenamefont {Lin}\ and\ \citenamefont
  {Tong}(2020)}]{Lin2020optimalpolynomial}%
  \BibitemOpen
  \bibfield  {author} {\bibinfo {author} {\bibfnamefont {L.}~\bibnamefont
  {Lin}}\ and\ \bibinfo {author} {\bibfnamefont {Y.}~\bibnamefont {Tong}},\
  }\href {\doibase 10.22331/q-2020-11-11-361} {\bibfield  {journal} {\bibinfo
  {journal} {{Quantum}}\ }\textbf {\bibinfo {volume} {4}},\ \bibinfo {pages}
  {361} (\bibinfo {year} {2020})}\BibitemShut {NoStop}%
\bibitem [{\citenamefont {Harrow}\ and\ \citenamefont
  {Kothari}(2021)}]{RobinAram}%
  \BibitemOpen
  \bibfield  {author} {\bibinfo {author} {\bibfnamefont {A.~W.}\ \bibnamefont
  {Harrow}}\ and\ \bibinfo {author} {\bibfnamefont {R.}~\bibnamefont
  {Kothari}},\ }\href@noop {} {\bibfield  {journal} {\bibinfo  {journal} {In
  preparation}\ } (\bibinfo {year} {2021})}\BibitemShut {NoStop}%
\bibitem [{\citenamefont {Dranov}\ \emph {et~al.}(1998)\citenamefont {Dranov},
  \citenamefont {Kellendonk},\ and\ \citenamefont {Seiler}}]{DKS98}%
  \BibitemOpen
  \bibfield  {author} {\bibinfo {author} {\bibfnamefont {A.}~\bibnamefont
  {Dranov}}, \bibinfo {author} {\bibfnamefont {J.}~\bibnamefont {Kellendonk}},
  \ and\ \bibinfo {author} {\bibfnamefont {R.}~\bibnamefont {Seiler}},\ }\href
  {\doibase 10.1063/1.532382} {\bibfield  {journal} {\bibinfo  {journal}
  {Journal of Mathematical Physics}\ }\textbf {\bibinfo {volume} {39}},\
  \bibinfo {pages} {1340} (\bibinfo {year} {1998})}\BibitemShut {NoStop}%
\bibitem [{\citenamefont {Childs}\ \emph {et~al.}(2018)\citenamefont {Childs},
  \citenamefont {Maslov}, \citenamefont {Nam}, \citenamefont {Ross},\ and\
  \citenamefont {Su}}]{Childs18}%
  \BibitemOpen
  \bibfield  {author} {\bibinfo {author} {\bibfnamefont {A.~M.}\ \bibnamefont
  {Childs}}, \bibinfo {author} {\bibfnamefont {D.}~\bibnamefont {Maslov}},
  \bibinfo {author} {\bibfnamefont {Y.}~\bibnamefont {Nam}}, \bibinfo {author}
  {\bibfnamefont {N.~J.}\ \bibnamefont {Ross}}, \ and\ \bibinfo {author}
  {\bibfnamefont {Y.}~\bibnamefont {Su}},\ }\href {\doibase
  10.1073/pnas.1801723115} {\bibfield  {journal} {\bibinfo  {journal}
  {Proceedings of the National Academy of Sciences}\ }\textbf {\bibinfo
  {volume} {115}},\ \bibinfo {pages} {9456} (\bibinfo {year}
  {2018})}\BibitemShut {NoStop}%
\bibitem [{\citenamefont {Haah}(2019)}]{Haah2019product}%
  \BibitemOpen
  \bibfield  {author} {\bibinfo {author} {\bibfnamefont {J.}~\bibnamefont
  {Haah}},\ }\href {\doibase 10.22331/q-2019-10-07-190} {\bibfield  {journal}
  {\bibinfo  {journal} {{Quantum}}\ }\textbf {\bibinfo {volume} {3}},\ \bibinfo
  {pages} {190} (\bibinfo {year} {2019})}\BibitemShut {NoStop}%
\bibitem [{\citenamefont {Dong}\ \emph {et~al.}(2021)\citenamefont {Dong},
  \citenamefont {Meng}, \citenamefont {Whaley},\ and\ \citenamefont
  {Lin}}]{Dong21}%
  \BibitemOpen
  \bibfield  {author} {\bibinfo {author} {\bibfnamefont {Y.}~\bibnamefont
  {Dong}}, \bibinfo {author} {\bibfnamefont {X.}~\bibnamefont {Meng}}, \bibinfo
  {author} {\bibfnamefont {K.~B.}\ \bibnamefont {Whaley}}, \ and\ \bibinfo
  {author} {\bibfnamefont {L.}~\bibnamefont {Lin}},\ }\href {\doibase
  10.1103/PhysRevA.103.042419} {\bibfield  {journal} {\bibinfo  {journal}
  {Physical Review A}\ }\textbf {\bibinfo {volume} {103}},\ \bibinfo {pages}
  {042419} (\bibinfo {year} {2021})}\BibitemShut {NoStop}%
\bibitem [{\citenamefont {Chao}\ \emph {et~al.}(2020)\citenamefont {Chao},
  \citenamefont {Ding}, \citenamefont {Gilyen}, \citenamefont {Huang},\ and\
  \citenamefont {Szegedy}}]{Chao2020}%
  \BibitemOpen
  \bibfield  {author} {\bibinfo {author} {\bibfnamefont {R.}~\bibnamefont
  {Chao}}, \bibinfo {author} {\bibfnamefont {D.}~\bibnamefont {Ding}}, \bibinfo
  {author} {\bibfnamefont {A.}~\bibnamefont {Gilyen}}, \bibinfo {author}
  {\bibfnamefont {C.}~\bibnamefont {Huang}}, \ and\ \bibinfo {author}
  {\bibfnamefont {M.}~\bibnamefont {Szegedy}},\ }\href
  {https://arxiv.org/abs/2003.02831} {\bibfield  {journal} {\bibinfo  {journal}
  {arXiv:2003.02831}\ } (\bibinfo {year} {2020})}\BibitemShut {NoStop}%
\bibitem [{\citenamefont {Jansen}\ \emph {et~al.}(2007)\citenamefont {Jansen},
  \citenamefont {Ruskai},\ and\ \citenamefont {Seiler}}]{jansen2007bounds}%
  \BibitemOpen
  \bibfield  {author} {\bibinfo {author} {\bibfnamefont {S.}~\bibnamefont
  {Jansen}}, \bibinfo {author} {\bibfnamefont {M.-B.}\ \bibnamefont {Ruskai}},
  \ and\ \bibinfo {author} {\bibfnamefont {R.}~\bibnamefont {Seiler}},\ }\href
  {\doibase https://doi.org/10.1063/1.2798382} {\bibfield  {journal} {\bibinfo
  {journal} {Journal of Mathematical Physics}\ }\textbf {\bibinfo {volume}
  {48}},\ \bibinfo {pages} {102111} (\bibinfo {year} {2007})}\BibitemShut
  {NoStop}%
\bibitem [{\citenamefont {Sanders}\ \emph {et~al.}(2020)\citenamefont
  {Sanders}, \citenamefont {Berry}, \citenamefont {Costa}, \citenamefont
  {Tessler}, \citenamefont {Wiebe}, \citenamefont {Gidney}, \citenamefont
  {Neven},\ and\ \citenamefont {Babbush}}]{Sanders2020}%
  \BibitemOpen
  \bibfield  {author} {\bibinfo {author} {\bibfnamefont {Y.~R.}\ \bibnamefont
  {Sanders}}, \bibinfo {author} {\bibfnamefont {D.~W.}\ \bibnamefont {Berry}},
  \bibinfo {author} {\bibfnamefont {P.~C.~S.}\ \bibnamefont {Costa}}, \bibinfo
  {author} {\bibfnamefont {L.~W.}\ \bibnamefont {Tessler}}, \bibinfo {author}
  {\bibfnamefont {N.}~\bibnamefont {Wiebe}}, \bibinfo {author} {\bibfnamefont
  {C.}~\bibnamefont {Gidney}}, \bibinfo {author} {\bibfnamefont
  {H.}~\bibnamefont {Neven}}, \ and\ \bibinfo {author} {\bibfnamefont
  {R.}~\bibnamefont {Babbush}},\ }\href
  {https://journals.aps.org/prxquantum/abstract/10.1103/PRXQuantum.1.020312}
  {\bibfield  {journal} {\bibinfo  {journal} {PRX Quantum}\ }\textbf {\bibinfo
  {volume} {1}},\ \bibinfo {pages} {020312} (\bibinfo {year}
  {2020})}\BibitemShut {NoStop}%
\bibitem [{\citenamefont {Berry}\ \emph {et~al.}(2018)\citenamefont {Berry},
  \citenamefont {Kieferov\'a}, \citenamefont {Scherer}, \citenamefont
  {Sanders}, \citenamefont {Low}, \citenamefont {Wiebe}, \citenamefont
  {Gidney},\ and\ \citenamefont {Babbush}}]{BerryNPJ18}%
  \BibitemOpen
  \bibfield  {author} {\bibinfo {author} {\bibfnamefont {D.~W.}\ \bibnamefont
  {Berry}}, \bibinfo {author} {\bibfnamefont {M.}~\bibnamefont {Kieferov\'a}},
  \bibinfo {author} {\bibfnamefont {A.}~\bibnamefont {Scherer}}, \bibinfo
  {author} {\bibfnamefont {Y.~R.}\ \bibnamefont {Sanders}}, \bibinfo {author}
  {\bibfnamefont {G.~H.}\ \bibnamefont {Low}}, \bibinfo {author} {\bibfnamefont
  {N.}~\bibnamefont {Wiebe}}, \bibinfo {author} {\bibfnamefont
  {C.}~\bibnamefont {Gidney}}, \ and\ \bibinfo {author} {\bibfnamefont
  {R.}~\bibnamefont {Babbush}},\ }\href
  {https://doi.org/10.1038/s41534-018-0071-5} {\bibfield  {journal} {\bibinfo
  {journal} {npj Quantum Information}\ }\textbf {\bibinfo {volume} {4}},\
  \bibinfo {pages} {22} (\bibinfo {year} {2018})}\BibitemShut {NoStop}%
\bibitem [{\citenamefont {Burden}\ \emph {et~al.}(2000)\citenamefont {Burden},
  \citenamefont {Faires},\ and\ \citenamefont {Reynolds}}]{BurdenNA}%
  \BibitemOpen
  \bibfield  {author} {\bibinfo {author} {\bibfnamefont {R.~L.}\ \bibnamefont
  {Burden}}, \bibinfo {author} {\bibfnamefont {J.~D.}\ \bibnamefont {Faires}},
  \ and\ \bibinfo {author} {\bibfnamefont {A.~C.}\ \bibnamefont {Reynolds}},\
  }\href@noop {} {\emph {\bibinfo {title} {Numerical Analysis}}}\ (\bibinfo
  {publisher} {Brooks Cole},\ \bibinfo {year} {2000})\BibitemShut {NoStop}%
\bibitem [{\citenamefont {Changhao}(2021)}]{Changhao2021}%
  \BibitemOpen
  \bibfield  {author} {\bibinfo {author} {\bibfnamefont {Y.}~\bibnamefont
  {Changhao}},\ }\href {\doibase 10.1103/PhysRevA.104.052603} {\bibfield
  {journal} {\bibinfo  {journal} {Physical Review A}\ }\textbf {\bibinfo
  {volume} {104}},\ \bibinfo {pages} {052603} (\bibinfo {year}
  {2021})}\BibitemShut {NoStop}%
\end{thebibliography}
\end{document}